\documentclass[a4paper]{amsart}
\pdfoutput=1
\newtheorem{theorem}{Theorem}[section]

\newtheorem{proposition}[theorem]{Proposition}
\newtheorem{corollary}[theorem]{Corollary}

\theoremstyle{definition}

\theoremstyle{remark}
\newtheorem{remark}[theorem]{Remark}
\usepackage{mathrsfs}
\usepackage{hyperref}
\usepackage[all,cmtip]{xy}
\usepackage{xypic}
\usepackage{tikz,tikz-cd}
\usepackage{pgf}
\usetikzlibrary{positioning}
\usetikzlibrary{calc}
\usetikzlibrary{arrows}
\usetikzlibrary{shapes.multipart}
\usepackage{graphicx}
\usepackage{color}
\usepackage{amsmath}
\usepackage{amsfonts}
\usepackage{euler}
\usepackage{amssymb}
\usepackage{epsfig}
\usepackage{rotating}
\usepackage{graphics}
\DeclareMathSymbol{\shortminus}{\mathbin}{AMSa}{"39}
\usepackage{mathpazo}
\usepackage{caption}
\newcommand{\be}{\begin{equation}}
\newcommand{\ee}{\end{equation}}

\newcommand{\al}{\alpha}

\newcommand{\dz}{\wedge}

\newcommand{\ba}{\begin{array}}

\newcommand{\ea}{\end{array}}

\newcommand{\beq}{\begin{eqnarray}}

\newcommand{\eeq}{\end{eqnarray}}

\newtheorem{lm}{lemma}

\newtheorem{thee}{theorem}

\newtheorem{proo}{proposition}

\newtheorem{co}{corollary}

\newtheorem{rem}{remark}

\newtheorem{deff}{definition}

\newcommand{\bd}{\begin{deff}}

\newcommand{\ed}{\end{deff}}

\newcommand{\bl}{\begin{lm}}

\newcommand{\el}{\end{lm}}

\newcommand{\bp}{\begin{proo}}

\newcommand{\ep}{\end{proo}}

\newcommand{\bt}{\begin{thee}}

\newcommand{\et}{\end{thee}}

\newcommand{\bc}{\begin{co}}

\newcommand{\ec}{\end{co}}

\newcommand{\brm}{\begin{rem}}

\newcommand{\erm}{\end{rem}}

\newcommand{\der}{{\rm d}}

\usepackage{t1enc}
\def\frak{\mathfrak}

\newcommand{\newc}{\newcommand}

\let\ccdot\cdot
\def\cdot{\hbox to 2.5pt{\hss$\ccdot$\hss}}

\newc{\aR}{\mbox{\boldmath{$ R$}}}
\newc{\aS}{\mbox{\boldmath{$ S$}}}
\newc{\aT}{\mbox{\boldmath{$ T$}}}
\newc{\aW}{\mbox{\boldmath{$ W$}}}

\newc{\aK}{\mbox{\boldmath{$ K$}}}
\newc{\aL}{\mbox{\boldmath{$ L$}}}

\newcommand{\bbN}{\mathbb{N}}
\newcommand{\bbM}{\mathbb{M}}
\newcommand{\bbH}{\mathbb{H}}

\usepackage{amssymb}
\usepackage{amscd}

\newcommand{\hook}{\raisebox{-0.35ex}{\makebox[0.6em][r]
{\scriptsize $-$}}\hspace{-0.15em}\raisebox{0.25ex}{\makebox[0.4em][l]{\tiny
 $|$}}}

\newcommand{\bma}{\begin{pmatrix}}
\newcommand{\ema}{\end{pmatrix}}

\let\t=\tau
\let\m=\mu

\newc{\obstrn}[2]{B^{#1}_{#2}}

\newcommand{\rpl}                         
{\mbox{$
\begin{picture}(12.7,8)(-.5,-1)
\put(0,0.2){$+$}
\put(4.2,2.8){\oval(8,8)[r]}
\end{picture}$}}

\newcommand{\lpl}                         
{\mbox{$
\begin{picture}(12.7,8)(-.5,-1)
\put(2,0.2){$+$}
\put(6.2,2.8){\oval(8,8)[l]}
\end{picture}$}}

\usepackage{ifthen}

\newc{\tensor}[1]{#1}
\newc{\Mvariable}[1]{\mbox{#1}}
\newc{\down}[1]{{}_{#1}}
\newc{\up}[1]{{}^{#1}}

\newc{\JulyStrut}{\rule{0mm}{6mm}}
\newc{\midtenPan}{\mbox{\sf S}}
\newc{\midten}{\mbox{\sf T}}
\newc{\midtenEi}{\mbox{\sf U}}
\newc{\ATen}{\mbox{\sf E}}
\newc{\BTen}{\mbox{\sf F}}
\newc{\CTen}{\mbox{\sf G}}

\def\sideremark#1{\ifvmode\leavevmode\fi\vadjust{\vbox to0pt{\vss
 \hbox to 0pt{\hskip\hsize\hskip1em
 \vbox{\hsize3cm\tiny\raggedright\pretolerance10000
 \noindent #1\hfill}\hss}\vbox to8pt{\vfil}\vss}}}%

\newcommand{\Span}{\mathrm{Span}}

\numberwithin{equation}{section}

\newcounter{romenumi}
\newcommand{\labelromenumi}{(\roman{romenumi})}

\newcommand{\bbT}{\mathbb{T}}
\newcommand{\bbR}{\mathbb{R}}

\begin{document}
\title[The Beltrami -- de Sitter model]{The Beltrami -- de Sitter model:\\Penrose's CCC, Radon transform and a hidden ${\bf G}_2$ symmetry}
\vskip 1.truecm
\author{Pawe\l~ Nurowski} \address{Centrum Fizyki Teoretycznej,
  Polska Akademia Nauk, Al. Lotnik\'ow 32/46, 02-668 Warszawa, Poland, and Guangdong Technion -- Israel Institue of Technology, No. 241, Daxue Road, Jinping District, Shantou, Guangdong Province, China}
\email{nurowski@cft.edu.pl}
\date{\today}
\begin{abstract}
We combine the well -- known Beltrami -- Klein model of non -- Euclidean geometry on a two -- dimensional disk, where the geodesics are the chords of the disk, with the two -- dimensional de Sitter space. The geometry of the de Sitter space is defined on the complement of the Beltrami -- Klein disk in the plane, with the de Sitter metric being the unique Lorentzian Einstein metric whose light cones form cones tangent to the disk in this complement. This leads to a Beltrami -- de Sitter model on the plane $\bbR^2$, which is endowed with the Riemannian Beltrami metric on the disk and the Lorentzian de Sitter metric outside the disk in $\bbR^2$.

We explore the relevance of this model for Penrose's Conformal Cyclic Cosmology, first in the two -- dimensional setting and subsequently in higher dimensions, including the physically significant case of four dimensions. In this context, we define a Radon-like transform between the de Sitter and Beltrami spaces, facilitating the purely geometric transformation of physical fields from the Lorentzian de Sitter space to the Riemannian Beltrami space.

In the two --, and three -- dimensional cases, we also uncover a hidden ${\bf G}_2$ symmetry associated with the de Sitter spaces in these dimensions, which is related to a certain vector distribution naturally defined by the geometry of the model. We suggest the potential for discovering similar hidden symmetries in the $n$ -- dimensional Beltrami-de Sitter model.
  \end{abstract}
\maketitle
\tableofcontents
\newcommand{\bbS}{\mathbb{S}}
\newcommand{\sog}{\mathbf{SO}}
\newcommand{\slg}{\mathbf{SL}}
\newcommand{\og}{\mathbf{O}}
\newcommand{\soa}{\frak{so}}
\newcommand{\sla}{\frak{sl}}
\newcommand{\sua}{\frak{su}}
\newcommand{\dr}{\mathrm{d}}
\newcommand{\sug}{\mathbf{SU}}
\newcommand{\ug}{\mathbf{U}}
\newcommand{\gat}{\tilde{\gamma}}
\newcommand{\Gat}{\tilde{\Gamma}}
\newcommand{\thet}{\tilde{\theta}}
\newcommand{\Thet}{\tilde{T}}
\newcommand{\rt}{\tilde{r}}
\newcommand{\st}{\sqrt{3}}
\newcommand{\kat}{\tilde{\kappa}}
\newcommand{\kz}{{K^{{~}^{\hskip-3.1mm\circ}}}}
\newcommand{\bv}{{\bf v}}
\newcommand{\di}{{\rm div}}
\newcommand{\curl}{{\rm curl}}
\newcommand{\cs}{(M,{\rm T}^{1,0})}
\newcommand{\tn}{{\mathcal N}}
\newcommand{\rank}{\mathrm{rank}}
\section{Introduction}
The Conformal Cyclic Cosmology (CCC) proposed by Roger Penrose \cite{Pen} has garnered significant interest among theoretical physicists \cite{Hill,MN,newman,nurpoincare,tod1,tod2} and has sparked a search for its potential astronomical manifestations \cite{AMP,AMNP,lopez,MNR}. In its simplest form, CCC postulates that the Universe consists of a sequence of consecutive eons, each characterized by its own Lorentzian metric. At both the beginning and end of each eon, the metric is degraded to a mere conformal class. This conformal class is then promoted to contain a Lorentzian metric for the next eon, and this process of conformally gluing consecutive eons continues, possibly infinitely. Within the framework of CCC, the currently observed Universe represents just one of the Lorentzian eons within this infinite sequence.

A major challenge to Penrose’s proposal is the absence of a satisfactory formulation for how a strict Lorentzian structure is imposed on the subsequent eon from the full conformal structure of the preceding eon. While this may eventually be addressed, it is possible that resolving this issue could require modifications to the very formulation of CCC, potentially extending beyond the scope of conventional General Relativity.

In this paper, we present an example of a CCC model in which the General Relativity paradigm is altered, allowing the metric to change signature when transitioning from one eon to the next. This model, although not strictly consistent with orthodox GR, has at least two advantages: \begin{itemize} \item It naturally glues eons -- the two consecutive eons reside on the same manifold by the model's very formulation, thereby making the "gluing problem" mentioned above moot; \item It is mathematically elegant and can be traced back to the works of Cayley, Beltrami, and Klein. \end{itemize}

In our model, referred to as the \emph{Beltrami -- de Sitter model}, we consider two distinct metrics: the $n$-dimensional Lorentzian de Sitter metric, which,
for 
$n=4$, is believed to asymptotically describe the metric to which the observed Universe converges, and the $n$-dimensional Riemannian Beltrami metric. These two metrics coexist on the same manifold, which is simply $\bbR^n$.
The two regions -- Lorentzian and Riemannian -- are separated by an $(n-1)$-dimensional sphere $\bbS^{(n -1)}$, which represents the boundary $\partial B^n$ of the unit ball $B^n$ centered at the origin of $\bbR^n$. Inside the ball, the metric is the Riemannian Beltrami metric, while outside the ball, the metric is the Lorentzian de Sitter metric.

Both metrics exhibit singularities, but these singularities only appear in their Einstein conformal factor at the boundary sphere $\bbS^{(n - 1)}$, where the signature change -- from Lorentzian to Riemannian -- occurs.

  The naturalness of this Riemannian -- Lorentzian hybrid structure on $\bbR^n$ arises from the fact that the metric formula for each region, whether Lorentzian or Riemannian, is formally valid in the other region as well. A simple geometric definition of the Lorentzian structure outside the ball $B^n$ states that this metric is the unique Lorentzian Einstein metric for which the cones in $\bbR^n$ tangent to the ball $B^n$ are light cones (see Section \ref{Belde} and the beginning of Section \ref{sec7}). This defined metric, when written in the standard coordinates of $\bbR^n$, turns out to be valid everywhere in 
$\bbR^n$, except at the points on the boundary sphere $\partial B^n$. Once the Lorentzian metric is defined according to the above rule in the complement of the ball $B^n$ in $\bbR^n$, it is recognized as the de Sitter metric in this region, and it continues to the Riemannian Beltrami metric inside the ball.

  The hybrid Lorentzian–Riemannian metric $\der s^2$, defined in the complement of the sphere $\partial B^n$, defines a conformal structure $[\der s^2]$ that is continuous everywhere in $\bbR^n$ , with the exception of a degenerate behavior at the boundary $\partial B^n$, where the metric degenerates, allowing for the signature transition.

  The model has many interesting features, which are discussed in the paper, first in a toy model for 
$n=2$ (Sections \ref{bemel}–\ref{comp_tw}), and then in the general $n\geq 2$ case, including the physically important case of $n=4$ (Section \ref{sec7}).

In this Introduction, we list the most important of these features:

\begin{itemize} \item The Universe in the model has two main phases (see Sections \ref{congeod}, Figure \ref{comp_geo}): \begin{itemize} \item In the first phase, which lasts infinitely long in the Universe's usual time, it is the standard de Sitter spacetime, with Lorentzian causality, enabling timelike, optical, and spacelike geodesics. There are de Sitter massive observers, who move on straight lines in $\bbR^n$, traversing the ball 
$B^n$ and, after an infinite amount of time, reach the boundary sphere $\partial B^n$. There are also massless particles, whose worldlines are straight lines in $\bbR^n$ tangent to this sphere. Although it takes infinite Universe time for a de Sitter massive observer to reach the Big Crunch at the boundary of the ball 
    $B^n$, from the perspective of $\bbR^n$, or using appropriately adjusted conformal time in the de Sitter first phase of the Universe, the life of such an observer does not end at the Big Crunch boundary sphere $\partial B^n$. The \emph{Universe does not end there}.
  \item The Universe then enters its \emph{second, Riemannian phase}, where the massive de Sitter observers continue their conformal travel along perfectly smooth paths --straight lines in $\bbR^n$, where the first section was the observer's de Sitter worldline, which now starts to become a \emph{cord} in the ball $B^n$. This happens to be a \emph{Riemannian geodesic} in the $n$-dimensional Beltrami space within the ball $B^n$. There is no Lorentzian time in this \emph{Riemannian purgatory}, but the former de Sitter observer, whose `worldline' is now the Riemannian geodesic cord in $B^n$, has at her disposal the Riemannian \emph{affine time}, instead. This time must again grow to infinity for the observer to reach the other side of the ball. However, with \emph{appropriately conformally adjusted time}, the observer will eventually reach the boundary of the ball, and after traversing it, will smoothly continue on the same straight line in $\bbR^n$ to begin his \emph{new de Sitter life}, now as a free de Sitter observer traveling backward in the de Sitter Universe time. (Sections \ref{congeod}, Figure \ref{comp_geo}). \end{itemize}
\item There is a \emph{duality} between \emph{points} in the Riemannian Beltrami space inside the ball $B^n$ and certain \emph{spacelike hyperplanes} in the Lorentzian de Sitter space, as well as a \emph{duality} between \emph{points} in the Lorentzian de Sitter space outside the ball $B^n$ and certain \emph{hyperdisks} in the Riemannian Beltrami space inside the ball $B^n$ (see Section \ref{proj_du} and Section \ref{rlt_n}).
\item This duality allows us to define a \emph{Radon -- like transform} that maps physical fields from the Lorentzian de Sitter spacetime to fields on the Riemannian Beltrami background (see Section \ref{radon_sec} and Section \ref{rlt_n}). Unlike the \emph{Wick rotation}, which is frequently used in particle physics to transform fields between Lorentzian and Riemannian regimes, this \emph{transform is purely geometric} and is an intrinsic feature of the de Sitter space.
\item If $n=2$, in a similar vein, the $n=2$ Beltrami -- de Sitter model has the interesting feature of associating a path of a \emph{ghost Riemannian observer} with any differentiable path of a de Sitter observer. This also holds in the opposite direction: every differentiable path of a Riemannian observer has a corresponding path of a \emph{ghost de Sitter observer} (see Section \ref{sec210}).
\item In the $n=2$ case, this enables us to define a \emph{coordinated movement} for each pair of observers, one in the de Sitter part and the other in the Beltrami part of the Beltrami -- de Sitter model. We call this coordinated movement the \emph{dancing movement} of a pair $(P(t),p(t))$ of respective observers, one inside and the other outside the ball $B^2$. Its defining principle states that an \emph{observer} with a path $p(t)$ in the de Sitter part of the model always \emph{moves in} $\bbR^2$ 
\emph{in the direction of the ghost observer} of the Riemannian observer $P(t)$, and that the Riemannian \emph{observer} with a path $P(t)$ in the Beltrami part of the model always \emph{moves in} $\bbR^2$ \emph{in the direction of the ghost observer} of the Lorentzian observer 
$p(t)$.
\item These \emph{dancing pairs} $(P(t),p(t))$ are curves $\gamma(t)=(p(t),P(t))$ in the \emph{Cartesian product of the Riemannian Beltrami space inside the ball} $B^2$ \emph{and the Lorentzian de Sitter space outside the ball} $B^2$. It follows that there is a \emph{unique Einstein split--signature metric} on this 4-dimensional Cartesian product, in which \emph{all dancing pairs are null curves}. This metric is by definition the \emph{dancing metric}, and it is naturally associated with the 2-dimensional Beltrami -- de Sitter model  (see Section \ref{sec3}).
\item This dancing (2, 2)-signature metric can be extended to a 4-dimensional Penrose-like \emph{correspondence space} \cite{eastwood} of the Beltrami -- de Sitter model. On its \emph{twistor bundle of real totally null planes}, one encounters a certain \emph{rank 2 distribution}, which has a $(2, 3, 5)$ \emph{growth vector} and exhibits a \emph{symmetry of a split real form of the simple exceptional Lie group} ${\bf G}_2$,   (see Section \ref{sec3}, especially Theorem \ref{dancing}, and Section \ref{hidden} with culmination in Theorem \ref{the_dist_g2}). This \emph{hidden} ${\bf G}_2$ \emph{symmetry of the Beltrami–de Sitter model} is a surprising result. 
\item 
Although this synopsis restricts the discussion to the case $n= 2$ from a certain point onward, the story of the \emph{dancing pairs}, the \emph{correspondence space} on the Cartesian product of the Beltrami part of the model and the de Sitter part, the \emph{dancing metric}, its \emph{twistor bundle} of real totally null planes, and the \emph{twistor distribution can be defined for every} $n\geq 2$ (see Section \ref{sec7}). In particular, in Section \ref{again}), we focused on the description of the twistor bundle for the 3-dimensional Beltrami -- de Sitter model. And there, the hidden ${\bf G}_2$ \emph{symmetry} appeared again as a symmetry of a certain distribution canonically associated with the $n=3$ Beltrami -- de Sitter twistor bundle.  This is also very surprising. The question of the hidden symmetry of the twistor distribution for $n\geq 4$ will be discussed in a paper with Ian Anderson \cite{ian}. 
\end{itemize}

Being convinced that our Riemannian -- Lorentzian hybrid, which we named the  Beltrami -- de Sitter model, is an \emph{appropriate model for physical reality} (since the change of signature is not really a major issue, especially in theories where the fundamental object in spacetime is not a metric but a connection), the Beltrami -- de Sitter model can be slightly modified to align with the standard CCC proposal. This is discussed in Section 6, where we use \emph{inversion with respect to the circle} to wrap the de Sitter model from outside the ball $B^n$ into the inside of the ball, and use the hybrid of the Lorentzian de Sitter outside the ball and the Lorentzian conformally inverted de Sitter inside as the Lorentzian two consecutive eons in the standard CCC. This model, although more physically appealing at first glance, suffers from the same main problem as all the toy models of CCC we know. While the appropriately rescaled de Sitter metric and its first fundamental form smoothly transition to the respective, rescaled, conformally inverted metric and its first fundamental form, the null geodesics passing through the boundary at $\partial B^n$ are only once differentiable. This is not the case in the Beltrami -- de Sitter hybrid, where the geodesics in both parts of the model are straight lines in $\bbR^n$ , which are obviously infinitely many times differentiable curves everywhere. Apart from the ad hoc gluing of the Lorentzian -- Lorentzian hybrid in Section 6, this feature of non-differentiable light rays in the purely Lorentzian model supports our belief that we should abandon the restriction of Lorentzian signatures in all eons and instead adopt the Beltrami -- de Sitter model as a prototype for a less restrictive Lorentzian -- Riemannian CCC.
\section{Beltrami-de Sitter model in 2 dimensions}\label{bemel}
\subsection{Beltrami's metric}\label{bemek}
The Beltrami-Klein \cite{bel1,bel2,klein} model is a model of non-Euclidean geometry in dimension 2, in which the 2-dimensional arena for the geometry consists of an open disk on the plane. It is convenient to use the Cartesian coordinates $(x,y)$ on the plane, center this disk at the origin, and assume that it has radius 1. Then the points of this geometry are just $p=(x,y)\in \bbR^2$ such that $x^2+y^2<1$. 

Let us denote by $d_E(u,v)=\sqrt{(u_x-v_x)^2+(u_y-v_y)^2}$ the standard Euclidean distance between points $u=(u_x,u_y)$ and $v=(v_x,v_y)$ on the plane. Then, the distance $d(p,q)$ between two points $p=(x,y)$ and $q=(x+\der x,y+\der y)$ in the Beltrami-Klein model is given by \cite{cayley}:
\be
d(p,q)=\tfrac12\log\Big(~\frac{d_E(q,a)}{d_E(q,b)}:\frac{d_E(p,a)}{d_E(p,b)}~\Big).\label{met}\ee
Here $a=(a_x,a_y)$ and $b=(b_x,b_y)$ are the two endpoints of the unique cord of the 
\begin{figure}[h!]
\centering
\includegraphics[scale=0.14]{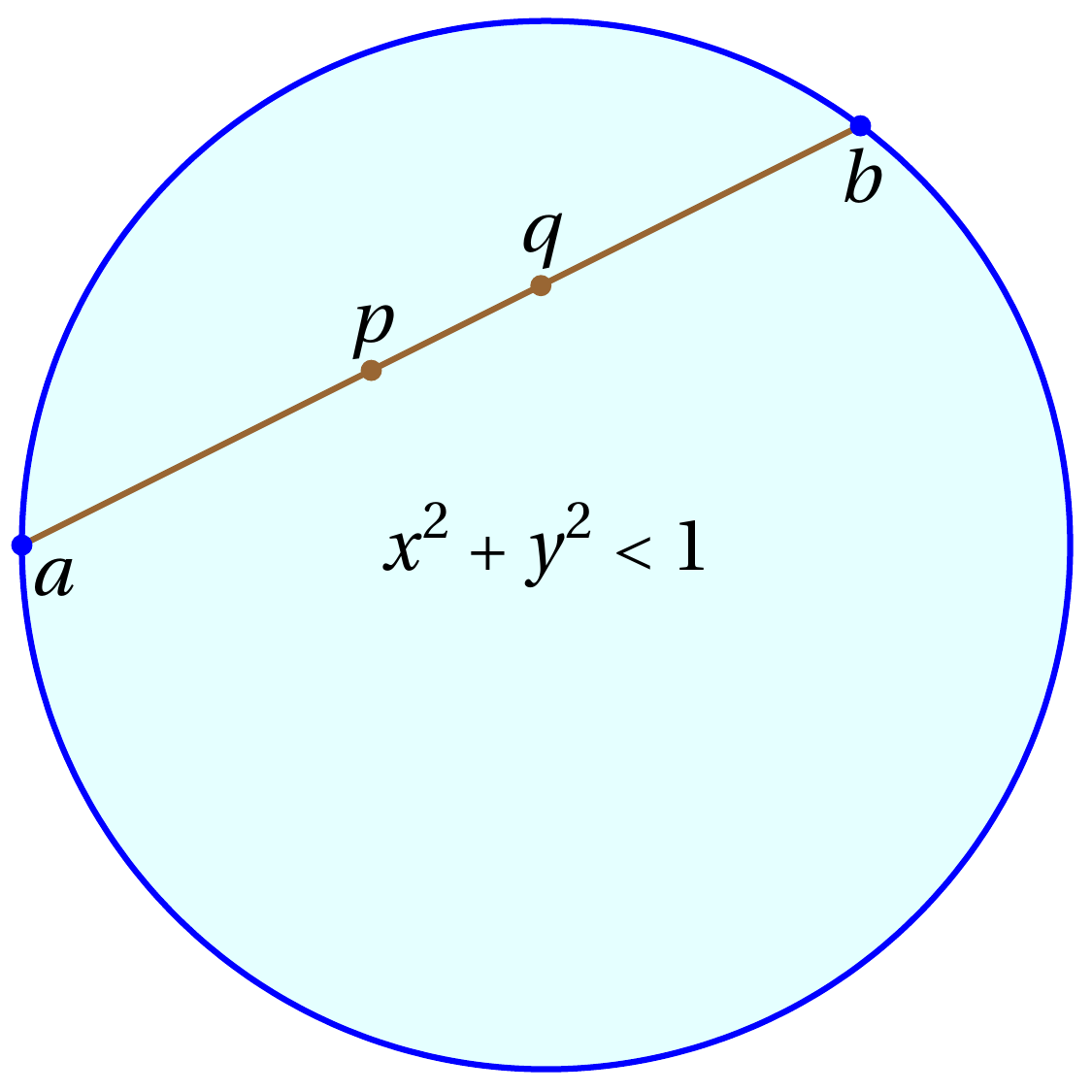}
\caption{\tiny{The Beltrami-Klein distance between points $p$ and $q$ is determined by the cross-ratio $\frac{d_E(q,a)}{d_E(q,b)}:\frac{d_E(p,a)}{d_E(p,b)}$ of the Euclidean distances between the four points $(a,p,q,b)$ lying on circle's cord passing through $p$ and $q$, for which points $a$ and $b$ are the end points.}}
\end{figure}
circle
$x^2+y^2=1$, which passes through both points $p$ and $q$.

Working infinitesimally, with $\der y=\tfrac{\der y}{\der x}\der x=y'\der x$, after some algebra, we easily obtain that:
$$\tiny{\begin{aligned}
a=(\frac{-y'(y-xy')-\sqrt{1-(y+(1-x)y')(y-(1+x)y')}}{1+{y'}^2},
\frac{y-xy'-y'\sqrt{1-(y+(1-x)y')(y-(1+x)y')}}{1+{y'}^2})\\
b=(\frac{-y'(y-xy')+\sqrt{1-(y+(1-x)y')(y-(1+x)y')}}{1+{y'}^2},
\frac{y-xy'+y'\sqrt{1-(y+(1-x)y')(y-(1+x)y')}}{1+{y'}^2}).
\end{aligned}}$$
Inserting this into (\ref{met}) we find the expression for the Beltrami-Klein distance written in terms of the coordinates of $p$ and $q$:
$$
\tiny{\begin{aligned}d(p,q)=&\tfrac12 \log\Big(\frac{1-x^2-y^2-(x+yy'-\sqrt{1-(y+(1-x)y')(y-(1+x)y')})\der x}{1-x^2-y^2-(x+yy'+\sqrt{1-(y+(1-x)y')(y-(1+x)y')})\der x}\Big).\end{aligned}}$$
Thus, looking for an infinitesimal distance $\der s$ between the points $(x,y)$ and $(x+\der x,y+\der y)$, separated infinitesimally by a vector $(\der x,\der y)$, we obtain
$$\tiny{\begin{aligned}
\der s=
&\tfrac12 \log\Big(1+\frac{2\sqrt{1-(y+(1-x)y')(y-(1+x)y')})\der x}{1-x^2-y^2-(x+yy'+\sqrt{1-(y+(1-x)y')(y-(1+x)y')})\der x}\Big)\backsimeq\\
&\tfrac12 \log\Big(1+\frac{2\sqrt{1-(y+(1-x)y')(y-(1+x)y')})\der x}{1-x^2-y^2}\Big)\backsimeq
\frac{\sqrt{1-(y+(1-x)y')(y-(1+x)y')})}{1-x^2-y^2}\der x.
\end{aligned}}
$$
Squaring this, and returning to $y'=\tfrac{\der y}{\der x}$, we obtain
\be
\der s^2=\frac{(1-y^2)\der x^2+2xy \der x \der y+(1-x^2)\der y^2}{(1-x^2-y^2)^2}.\label{beme}\ee
This is the \emph{Beltrami metric} for the Beltrami-Klein model. The metric is Riemannian everywhere within the disk $x^2+y^2<1$, and has the property that a geodesic passing through any two points $p$ and $q$ within the disk is an Euclidean interval between these two points. One can also easily verify that the metric has \emph{negative constant} Gaussian curvature $\kappa=-\tfrac12$. As such it has a maximal number 3 of Killing vectors, which are given by :
\be
X_1=y\partial_x-x\partial_y,\quad\quad X_2=xy\partial_x+(y^2-1)\partial_y,\quad\quad X_3=(x^2-1)\partial_x+xy\partial_y,\label{kil}\ee
and which satisfy the commutation relations $[X_1,X_2]=-X_3$, $[X_3,X_1]=-X_2$, $[X_2,X_3]=X_1$, proper for the Lie algebra $\sla(2,\bbR)$.

Looking at the determinant 
$$\Delta=\det\bma\tfrac{1-y^2}{(1-x^2-y^2)^2}&\tfrac{xy}{(1-x^2-y^2)^2}\\
\tfrac{xy}{(1-x^2-y^2)^2}&\tfrac{1-x^2}{(1-x^2-y^2)^2}\ema=\tfrac{1}{(1-x^2-y^2)^{3}}$$
of the Beltrami metric $\der s^2$ we see that the metric, which was originally defined only within the disk $x^2+y^2<1$ is also well defined in the \emph{complement of the disk in the plane}, i.e. in the region where $x^2+y^2>1$. However, because of the \emph{third} power appearing in the denominator of the determinant, determinant is negative there, and the metric in this region is not Riemannian anymore. It has \emph{Lorentzian signature} there. 
\subsection{Projective duality}\label{proj_du}
To get a geometric meaning to this Lorentzian metric we may use \emph{projective duality} between, respectively, \emph{points and lines} in the disk $x^2+y^2<1$, and \emph{lines and points} in the region outside the disk, where $x^2+y^2>1$.
\begin{figure}[h!]
\centering
\includegraphics[scale=0.3]{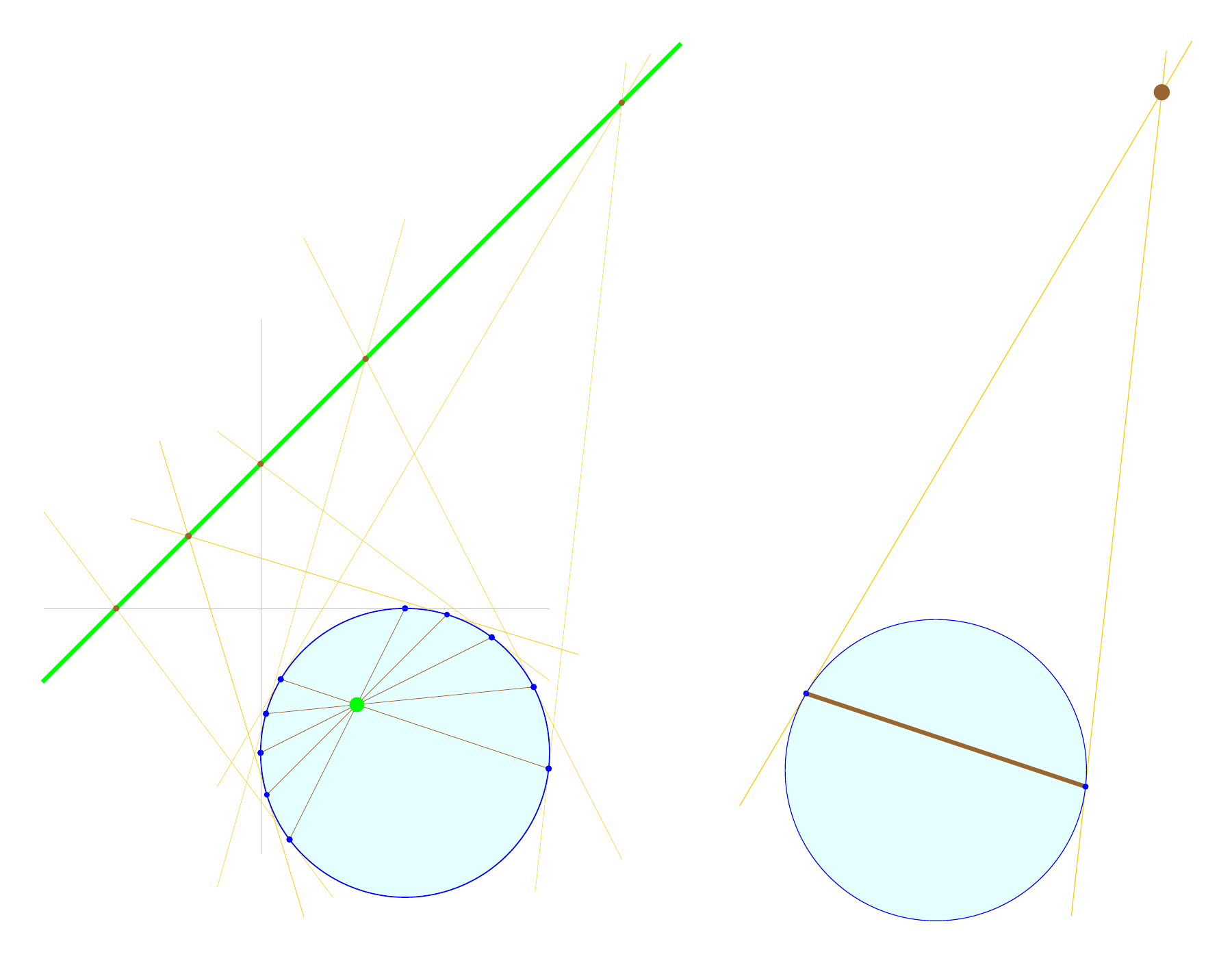}
\caption{\tiny{The  $point\leftrightarrow line$ correspondence in the Beltrami -- de Sitter model.\\
    {\bf On the left:} \emph{Every point} inside the Beltrami-Klein disk, as for example the green point in this figure, \emph{defines a line} outside the disk - the green line on the figure. The green line is obtained as the locus of points at which pairs of red lines intersect. Each pair of red lines consists of lines tangent to the circle $x^2+y^2=1$ at the end points of some cord passing through the green point. The (green) lines outside the disk, corresponding to the (green) points inside the disk, are \emph{spacelike geodesics} in the Lorentzian Beltrami metric (\ref{beme}) outside the disk.\\
    {\bf On the right:} \emph{Every point} outside the disk, i.e. an \emph{event} in the Lorentzian Beltrami spacetime (a green point), defines a (green) cord in the Beltrami-Klein disk with the end points belonging to the two tangent lines to the circle outgoing from the green event. The corresponding (green) cord is a geodesic in the Riemannian Beltrami metric (\ref{beme}) inside the disk.\\
Compare this picture with the last picture in Section 2.4 in \cite{Thurston}.}}\label{fi2}
\end{figure}

Given a point within the disk $x^2+y^2<1$ there is an infinite number of cords of the circle $x^2+y^1=1$ passing through it (see the left side of Figure \ref{fi2}). The end points of each of these cords determine two straight lines tangent to the circle at the end points of the cord. Except the situation when the cord is the diameter of the circle the two lines intersect outside the disk. It is a remarkable fact that \emph{all the intersection points of all pairs of these tangent lines lie on the same straight line outside the disk!} It further follows that each of these straight lines determined in the region outside the disk by a choice of a point inside the disk are \emph{spacelike geodesics} in the Lorentzian Beltrami metric (\ref{beme}).

The converse is also true: \emph{each point outside the disk}, i.e. each point in the Beltrami 2-dimensional spacetime $x^2+y^2>1$, uniquely \emph{defines a geodesic} for the Riemannian Beltrami metric (\ref{beme}), \emph{in the inside region} $x^2+y^2<1$. Here the construction of the geodesic corresponding to a spacetime point is easier than before (see right side of Figure \ref{fi2}): Given a point outside the disk, consider the two tangents to the circle outgoing from this point. The points of the tangency define a cord of the circle, which in this way is uniquely associated with the point outside. This cord is a geodesic in the Riemannian Beltrami metric (\ref{beme}) inside the disk. 

\subsection{Beltrami outside the disk is de Sitter}\label{Belde}
The well known $point\longleftrightarrow{line}$ duality described above enables for yet another approach \cite{bor} to the Beltrami metric (\ref{beme}). To discuss this we start from the \emph{outside} of the Beltrami-Klein disk.
\begin{figure}[h!]
\centering
\includegraphics[scale=0.3]{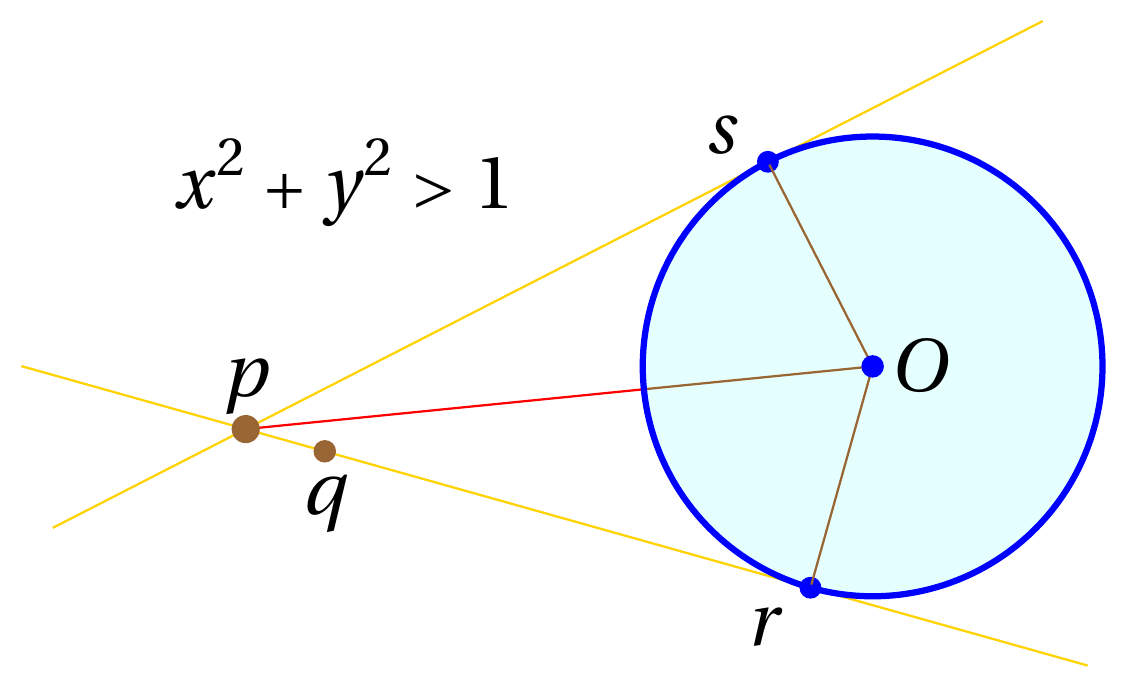}
\caption{\tiny{The conformal structure as defined by cones tangent to the unit disk. \\The larger brown point outside the disk defines a cone in the region $x^2+y^2>1$. This cone is interpreted as a light cone of a spacetime event marked by the larger brown point. We interpret the white area in the cone between the point $p$ and the part of the circle $x^2+y^2=1$ between $r$ and $s$ as the possible future of the larger brown point. This is the region of spacetime which the larger brown point can achieve by moving along timelike curves. Every point outside the Beltrami-Klein disk defines its own cone in the same manner. This defines a conformal structure in the region $x^2+y^2>1$. In this conformal structure the null geodesics are just the straight lines tangent to the circle $x^2+y^2=1$, and the time of observers in $x^2+y^2>1$ region has an arrow from a particular event towards the circle $x^2+y^2=1$. In particular, the vector $\stackrel{\longrightarrow}{pO}$ shows the time arrow within the cone of the larger brown point.}}\label{fi3}
\end{figure}

Choosing a point $p=(x,y)$ outside the Beltrami-Klein disk, $x^2+y^2>1$ we have precisely two straight lines that pass through $p$ and are tangent to the circle $x^2+y^2=1$. We interpret these lines as \emph{light rays} (null geodesics) outgoing from $p$. We now look for a \emph{conformal} metric in which the lines $\stackrel{\longrightarrow}{pr}$ and $\stackrel{\longrightarrow}{ps}$ form its \emph{light cone} with a tip at $p$ (see Figure \ref{fi3}). We are interested in an infinitesimal vector $\stackrel{\longrightarrow}{pq}=(\der x,\der y)$, associated with the pair of points $p$ and $q=(x+\der x,y+\der y)$, which we want to be tangent to the light cone wit tip at $p$. This condition is equivalent to the condition that
\be
\stackrel{\longrightarrow}{pq}\times\stackrel{\longrightarrow}{pr}=0,\label{con}\ee
where $\times$ is the usual vector product in $\bbR^3$, and we extended the vectors appearing in this conditions to vectors in $\bbR^3$ with the same $x$and $y$ components as in the plane of Figure \ref{fi3}, and with the $z$ components equal to 0. Since from the elementary geometry visible in Figure \ref{fi3} we have
$$d_E(0,r)=1\quad\&\quad d_E(p,r)^2+d_E(0,r)^2=d_E(0,p)^2\quad\&\quad\stackrel{\longrightarrow}{pq}\cdot \stackrel{\longrightarrow}{Or}=0,$$
where by `$\cdot$' we denoted the standard scalar product in $\bbR^2$, squaring (\ref{con}) we get:
$$\begin{aligned}
0=&(\stackrel{\longrightarrow}{pq})^2 (\stackrel{\longrightarrow}{pr})^2-(\stackrel{\longrightarrow}{pq}\cdot\stackrel{\longrightarrow}{pr})^2=\\&d_E(p,q)^2d_E(p,r)^2-(\stackrel{\longrightarrow}{pq}\cdot(\stackrel{\longrightarrow}{0r}-\stackrel{\longrightarrow}{0p}))^2=\\
&(\der x^2+\der y^2)(d_E(0,p)^2-1)-(\stackrel{\longrightarrow}{pq}\cdot\stackrel{\longrightarrow}{0p})^2=\\&(\der x^2+\der y^2)(x^2+y^2-1)-(x\der x+y\der y)^2=\\
&(y^2-1)\der x^2-2xy\der x\der y+(x^2-1)\der y^2.
\end{aligned}$$
This means that the infinitesimal vector $\stackrel{\longrightarrow}{pq}=(\der x,\der y)$ is tangent to the cone with tip at $p=(x,y)$, if and only if it is \emph{null} in the metric 
\be \der s_0^2=(y^2-1)\der x^2-2xy\der x\der y+(x^2-1)\der y^2.\label{conmet2}\ee
Note that this metric is \emph{conformal} to the Lorentzian Beltrami metric (\ref{beme}). To get the Lorentzian Beltrami metric one \emph{scales} $\der s_0^2$ by $-(1-x^2-y^2)^{(-2)}$:
$$\der s^2=\tfrac{-1}{(1-x^2-y^2)^2}\der s_0^2.$$ 
One may worry why we prefer to consider the metric $\der s^2$, which is apparently singular at the circle $x^2+y^2=1$, rather than metric $\der s_0^2$ which does not exhibit an obvious singularity. There are at least two reasons for that: 
\begin{itemize}
\item The metric $\der s_0^2$ has non-constant Gaussian curvature. After rescaling from $\der s_0^2$ to $\der s^2$ the so obtained Beltrami metric has a constant Gaussian curvature equal to $\kappa=-\tfrac12$. It has $\sla(2,\bbR)$ Lie algebra as the algebra of its Killing symmetries (with the generators $X_1$, $X_2$, $X_3$ as in(\ref{kil}), but now with $(x,y)$ in the region $x^2+y^2>1$). As such this metric is a 2-dimensional \emph{de Sitter} metric \cite{deS1,deS2}, and the region $\{(x,y)\ni\bbR^2~|~x^2+y^2>1\}$ equipped with the Beltrami metric (\ref{beme}) becomes yet another model for the 2-dimensional de Sitter space.
\item Short calculation of the determinant of the \emph{a'priori} nonsingular metric $\der s_0^2$ shows that this determinant equals:
$$\Delta_0=\det\bma y^2-1&-xy\\
-xy&x^2-1\ema=1-x^2-y^2,$$
and is obviously singular on the circle $x^2+y^2=1$, where the metric switches signature from Lorentzian to Riemannian.
\end{itemize}

\subsection{The model}\label{bds} Thus, on $BdS=\bbR^2\setminus\bbS^1$ we have a \emph{Beltrami--de Sitter hybrid metric} $\der s^2$, which is Riemannian inside $\bbS^1$ and Lorentzian outside. The metric structure on $\bbR^2\setminus\bbS^1$ can either be first defined in the inner disk
$$B^2\,=\,\{\,\bbR^2\ni(x,y)\,|\, x^2+y^2<1\,\}$$ in terms of the Riemannian Beltrami metric by the Beltrami-Cayley-Klein procedure and then extended to the Lorentzian de Sitter structure in the outer region of $\bbS^1$, or it can first be defined as a Lorentzian structure of the de Sitter metric in the outer region
$$M^2\,=\,\{\,\bbR^2\ni(x,y)\,|\,x^2+y^2>1\,\}$$
of $\bbS^1$ by a natural cone structure there, and then extended to the Riemannian Beltrami metric inside the disk. The hybrid Beltrami -- de Sitter metric structure with a regular $\der s^2$, as in \eqref{beme}, on $BdS=\bbR^2\setminus\bbS^1$ will be called the \emph{Beltrami -- de Sitter model} $(BdS,\der s^2)$.
\subsection{Hyperbolic spaces and Beltrami -- de Sitter model}
The Beltrami -- de Sitter model is related to the geometries, which hyperboloids
$$-T^2+(X^1)^2+(X^2)^2=-\epsilon, \quad\quad \epsilon=\pm1,$$
acquire from the ambient Minkowski three-space
$$\bbM=\{(T,X^1,X^2)\,|\,T,X^1,X^2\in\bbR\},$$ equipped with the flat Minkowski metric $$\eta=-\der T^2+(\der X^1)^2+(\der X^2)^2.$$ There are two kinds of these hyperboloids:
\begin{itemize}
\item the one-sheeted $\bbH_{\shortminus1}=\{(T,X^1,X^2)\,|\,T^2-(X^1)^2-(X^2)^2=-1\}$, \\and
  \item the two-sheeted $\bbH_{1}=\{(T,X^1,X^2)\,|\,T^2-(X^1)^2-(X^2)^2=1\}$,
  \end{itemize}
and from now on, we will only consider them for $T>0$.

After restriction to these hyperboloids, the Minkowski metric $\eta$ becomes
\be \eta_{|\bbH_\epsilon}=-\frac{(X^1\der X^1+X^2\der X^2)^2}{\epsilon + (X^1)^2+(x^2)^2}+(\der X^1)^2+(\der X^2)^2.\label{etah}\ee
This metric has the Riemannian signature on the hyperboloid $\bbH_1$, and it has Lorentzian signature on the one-sheeted hyperboloid $\bbH_{\shortminus1}$. It further follows that the pair $(\bbH_{\shortminus1},\eta_{|\bbH_{\shortminus1}})$ is locally isometric to the de Sitter space $(M^2,\der s^2)$, and that the entire one sheet of $\bbH_1$, with its metric $\eta_{\bbH_1}$, is locally isometric to the Beltrami space $(B^2,\der s^2)$.

\begin{figure}[h!]
\centering
\includegraphics[scale=0.27]{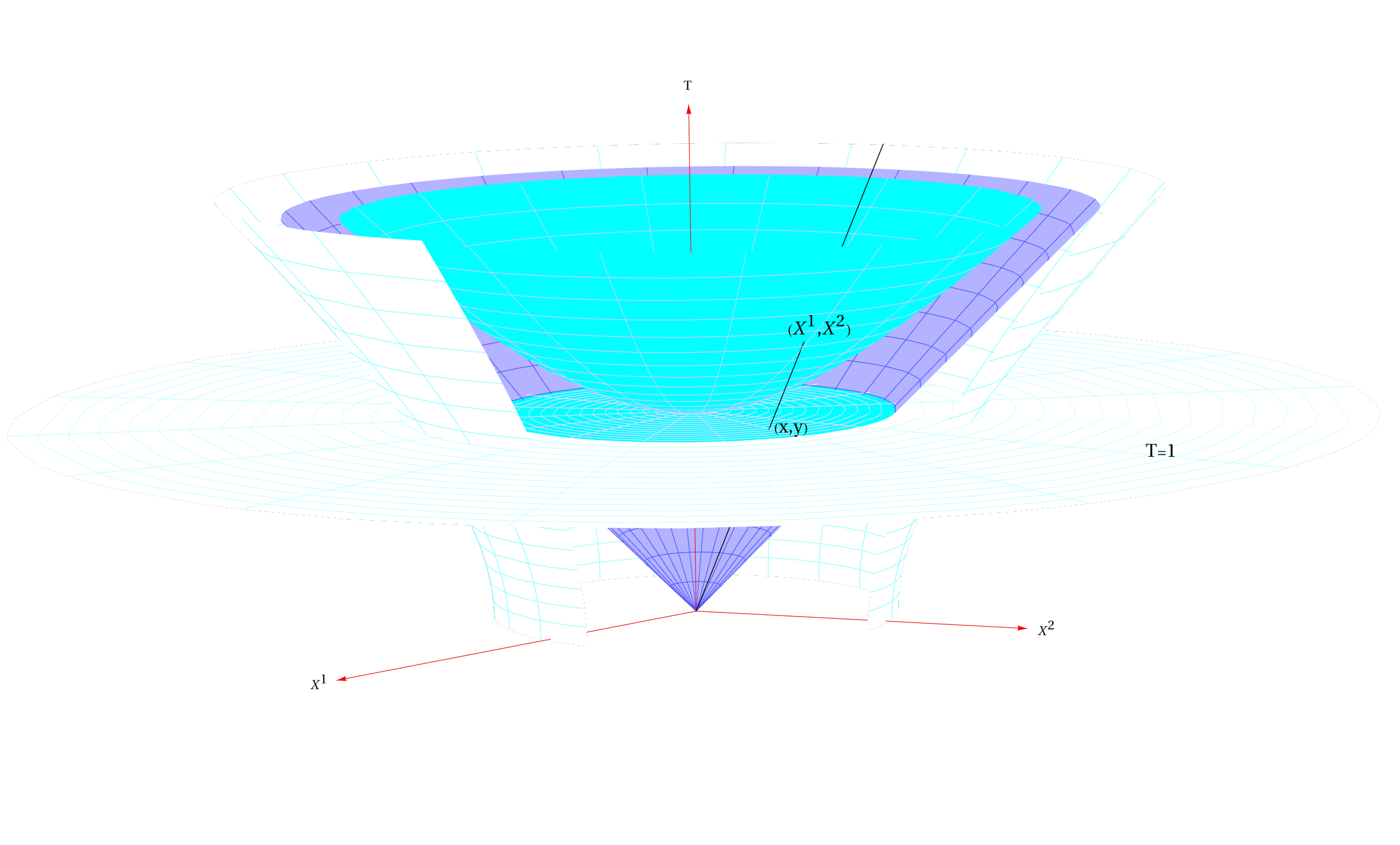} \includegraphics[scale=0.27]{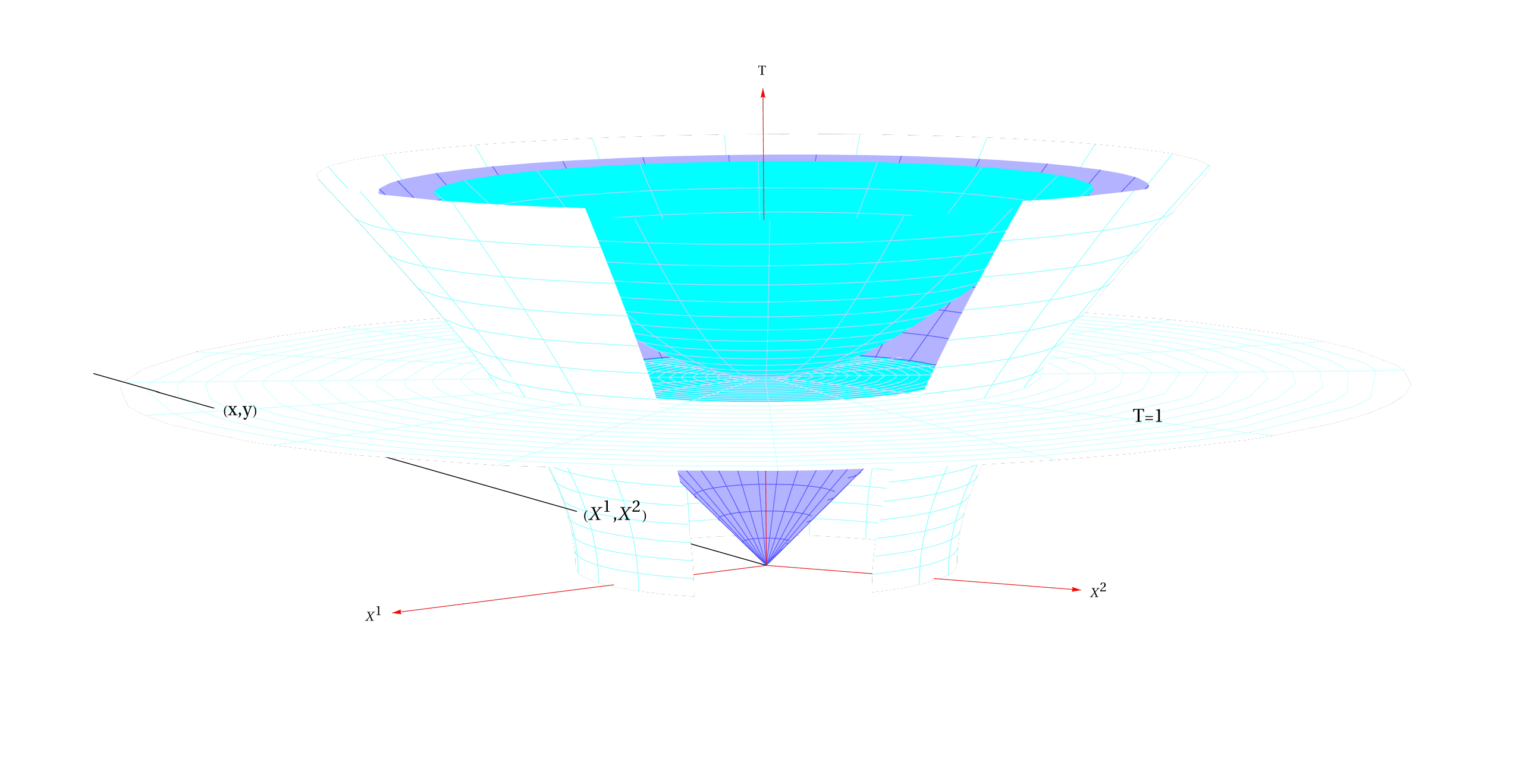}
\caption{\tiny{Two types of hyperboloids in the 3-dimensional Minkowski space with the metric $-\der T^2+\der X^2+\der Y^2$. The points $(T,X,Y)$ of the cyan colored hyperboloid satisfy $-T^2+X^2+Y^2=-1$ and $T>0$, whereas the points of the white hyperboloid satisfy $-T^2+X^2+Y^2=1$. The purple cone is the future light cone of the observer at the origin $(0,0,0)$.\\
    In the top figure it is visible that there is a one--to--one correspondence between the points of the cyan hyperboloid and the points on the disk of points $(T,X,Y)$ such that $X^2+Y^2<1$ and $T=1$. The correspondence is achieved by the projection of a point $(T,X,Y)$ from the hyperboloid, along a black line passing through this point and the origin, to the point of the intersection of this line with the disk.\\
In the bottom figure we show that there is also a one--to--one correspondence between the points of the white hyperboloid and the points on the complement to the cyan disk in the plane $T=1$. Again, the correspondence is achieved by connecting a point $(T,X,Y)$ on a white hyperboloid by a black line passing through the origin, and getting point of intersection of this line with the plane $T=1$ as the point corresponding to $(T,X,Y)$.}}\label{fig454}
\end{figure}
To see this one considers the \emph{central projection}
$$pr:\bbH_\epsilon\to \Pi$$ that maps points $(\sqrt{\epsilon+(X^1)^2+(X^2)^2},X^1,X^2)$ of the hyperboloids $\bbH_\epsilon$ to points $(1,x,y)$ in $\bbM$, which lie on the plane $$\Pi=\{(T,X^1,X^2)\,|\,T=1, X^1=x,X^2=y\}$$ tangent to the hyperboloid $\bbH_1$ at its tip $(T,X^1,X^2)=(1,0,0)$. This projection is defined as follows:\\
  \emph{Given a point $(T,X^1,X^2)$ on $\bbH_\epsilon$, consider a line $\ell$ in the Minkowski space $\bbM$ passing through this point and the origin $(T,X^1,X^2)=(0,0,0)$. This defines a point $(1,x,y)$ on the plane $T=1$ which is the intersection of the line $\ell$ and this plane.}

  The explicit formula is given by:
  $$pr
  \big(
  \sqrt{\epsilon+(X^1)^2+(X^2)^2},X^1,X^2
  \big)
  =
  \big(
  1,\frac{X^1}{\sqrt{\epsilon+(X^1)^2+(X^2)^2}},\frac{X^2}{\sqrt{\epsilon+(X^1)^2+(X^2)^2}}\big)$$
  which maps the parameters $(X^1,X^2)$ of points of the hyperboloid $H_\epsilon$  to the parameters $(x,y)$ of the plane $\Pi$ via:
  $$pr
  \big(X^1,X^2\big)
  =
  \big(x,y\big)$$
with 
$$x=\frac{X^1}{\sqrt{\epsilon+(X^1)^2+(X^2)^2}},\quad y=\frac{X^2}{\sqrt{\epsilon+(X^1)^2+(X^2)^2}}.$$
See the geometry of this transformation on Figure \ref{fig454}.

Now, passing from coordinates $(X^1,X^2)$ to the projected-to-the-plane-$\Pi$ coordinates $(x,y)$ in the hyperbolic metric \eqref{etah} we obtain:
$$\eta_{\bbH_\epsilon}=\epsilon\frac{(1-x^2-y^2)(\der x^2+\der y^2)+(x\der x+y \der y)^2}{(1-x^2-y^2)^2},$$
i.e., after a short algebra, the metric $\der s^2$ of the Beltrami -- de Sitter model.
\begin{corollary}\label{colo}
  The plane $\Pi=\{\bbM\ni(T,X^1,X^2)\,|\,T=1\}$ in the Minkowski space $(\bbM,\eta)$ is a realization of the Beltrami -- de Sitter model $BdS$.
  
  In this realization the Beltrami space $B^2$ consists of points $(1,x,y)\in \bbM$ such that $x^2+y^2<1$, and is equipped with the \emph{Riemannian} Beltrami metric $\eta_{|\bbH_1}$. As such, it can be identified with the space of all \emph{inertial observers} passing through the origin of the Minkowski space $(\bbM,\eta)$.

  The de Sitter space $M^2$ consists of points $(1,x,y)\in \bbM$ such that $x^2+y^2>1$, and is equipped with the \emph{Lorentzian} de Sitter metric $-\eta_{|\bbH_{\shortminus1}}$. It can be identified with  the space of all \emph{tachyons} passing through the origin of the Minkowski space $(\bbM,\eta)$. 
  \end{corollary}
\subsection{Special conformal geodesics in the Beltrami--de Sitter model}\label{congeod}
The interesting feature of the Beltrami--de Sitter model is that the \emph{chords} of the disk $x^2+y^2\leq 1$ (the brown chords on Figure \ref{fig444}) are the \emph{Riemannian geodesics} for $(B^2,\der s^2)$, and \emph{straight lines}, or \emph{half straight lines ending at the circle} $x^2+y^2=1$, are \emph{Lorentzian geodesics} for $(M^2,\der s^2$). Actually in $(M^2,\der s^2)$ there are three types of geodesics: \emph{null} geodesics -- the half straight lines \emph{tangent} to the circle $x^2+y^2=1$ at some point (the orange (half)lines on Figure \ref{fig444}), \emph{timelike} geodesics -- the straight half lines ending at the circle $x^2+y^2=1$ but \emph{not} tangent to it (the red half lines outside the cyan disk in Figure \ref{fig444}), and \emph{spacelike} geodesics -- the straight lines having zero intersection with the circle $x^2+y^2=1$ (the green lines on the Figure \ref{fig444}).

\begin{figure}[h!]
\centering
\includegraphics[scale=0.28]{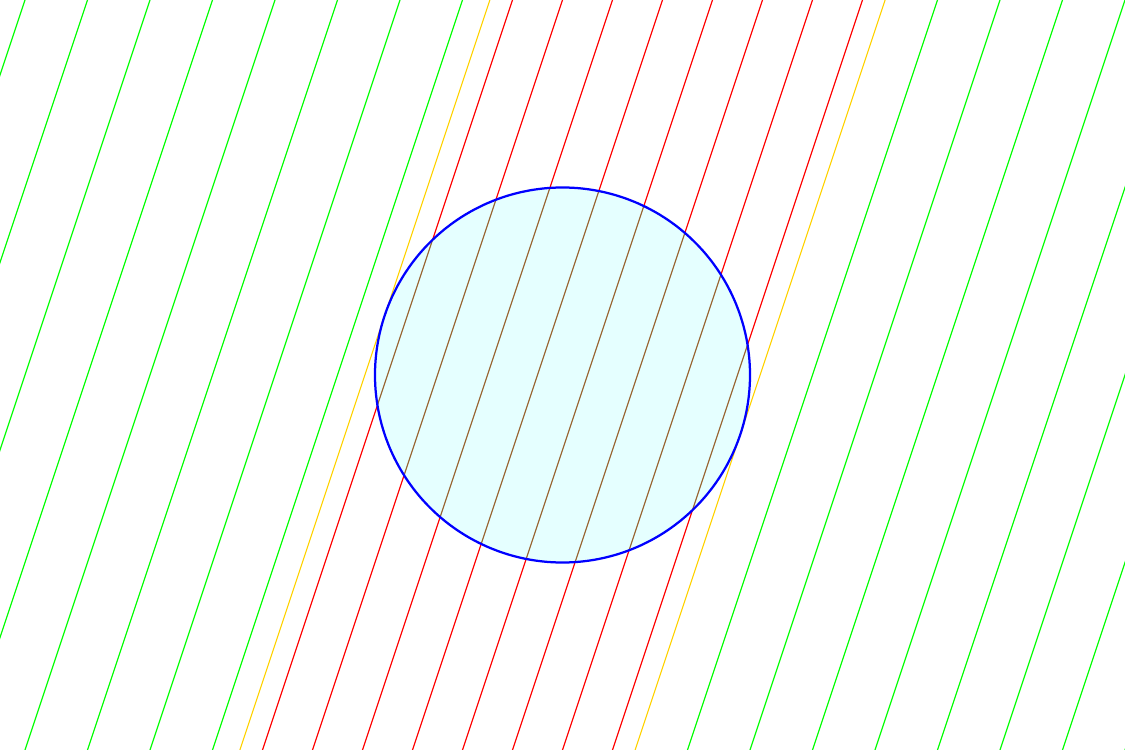}
\caption{\tiny{Four types of geodesics in the Beltrami--de Sitter model}}\label{fig444}
\end{figure}

This demonstrates that every straight line in the entire Euclidean plane \( \mathbb{R}^2 \), which intersects the Beltrami disk \( x^2 + y^2 \leq 1 \), i.e. every line in the \((x, y)\)-plane satisfying  
\[
Ax + By + C = 0 \quad \text{and} \quad A^2 + B^2 > C^2,
\]  
has the following behavior:  

\begin{itemize}
    \item It begins as a timelike geodesic at the past infinity of de Sitter space \( M^2 \) (one of the top red half lines on Figure \ref{fig444}).  
    \item It then reaches the Big Crunch circle of de Sitter space.  
    \item After crossing the Big Crunch circle, it continues as a chord (the brown interval being the continuation of the chosen red half line on Figure \ref{fig444}) within the Beltrami disk. In the Beltrami-Riemannian framework, this chord represents the shortest path (in the Beltrami metric) between the entry and exit points of the line within the disk.  
    \item Finally, it returns to the Lorentzian de Sitter space \( M^2 \), traveling backward in time as a timelike geodesic (the red line being the continuation of the brown chord on Figure \ref{fig444}) from the Big Crunch circle to the past infinity in \( M^2 \).  
\end{itemize}  

It should be noted that while a de Sitter observer's journey from the past infinity of \( M^2 \) to the Big Crunch of de Sitter space requires infinite proper time in the de Sitter metric \( ds^2 \), it takes only a finite time in the conformally related metric \( ds_0^2 \) defined by Equation (1.5).  

Furthermore, every complete Euclidean line \( Ax + By + C = 0 \) with \( A^2 + B^2 > C^2 \) is a geodesic in every metric belonging to the Lorentzian-Riemannian conformal class \([ds_0^2]\) in \( \mathbb{R}^2 \setminus \{\bbR^2\ni(x, y) \mid x^2 + y^2 = 1\} \). However, it is not a geodesic on the circle \( x^2 + y^2 = 1 \) because the conformal class changes its signature at this set, and the very concept of a geodesic breaks down there.  

We will return to this discussion in Section \ref{comp_two}, in particular in the caption of Figure \ref{comp_geo}.

\subsection{Radon--like transform between functions on opposite sides of the disk}\label{radon_sec}

We now use the \emph{projective duality} between the disk region $B^2$  and the open part of its complement $M^2$ in the Beltrami--de Sitter model to define a \emph{transform} between integrable functions defined on both spaces. For this we need some preparations: 

Consider \emph{pairs} $(p,P)$ of points in the Beltrami--de Sitter plane $\bbR^2$ such that $p$ is a point in the external region $M^2=\{(X,Y)\in \bbR^2~|~X^2+Y^2>1\}$ of the unit circle, $p\in M^2$, and $P$ is a point in its internal region $B^2=\{(X,Y)\in\bbR^2~|~X^2+Y^2<1\}$, $P\in B^2$. To distinguish between points in $M^2$ and $B^2$ we use letters from the end of the Latin alphabet to denote coordinates of points in $M^2$, i.e. if $p\in M^2$ it has coordinates $p=(x,y)$ with $x^2+y^2>1$, and we use letters form the beginning of the Latin alphabet to denote coordinates of points $P$ in $B^2$, i.e. $P$ is in $B^2$ if $P=(a,b)$ with $a^2+b^2<1$.

Recall from Figure \ref{fi2} on its right side, that every point $p=(x,y)$ in the de Sitter space $M^2$ defines a cord $c_p$ in the Beltrami space $B^2$. This is the cord connecting two points of tangency of the cone passing through $p$ and tangent to the circle $X^2+Y^2=1$. Given $p=(x,y)$ this cord is a set
\be
c_{(x,y)}=\{(X,Y)\in \bbR^2~|~x X +y Y =1\}\subset B^2.\label{cordcp}\ee
One sees that, because $x^2+y^2>1$ the line $xX+yY=1$ intersects the circle $X^2+Y^2=1$ in two points
\be\begin{aligned}
  s=&(\frac{x-y\sqrt{x^2+y^2-1}}{x^2+y^2},\frac{y+x\sqrt{x^2+y^2-1}}{x^2+y^2}),\\
  r=&(\frac{x+y\sqrt{x^2+y^2-1}}{x^2+y^2},\frac{y-x\sqrt{x^2+y^2-1}}{x^2+y^2}),\end{aligned}\label{cordcp1}\ee and the cord $c_{(x,y)}$ is the interval $sr$ on this line.
\begin{figure}[h!]
\centering
\includegraphics[scale=0.3]{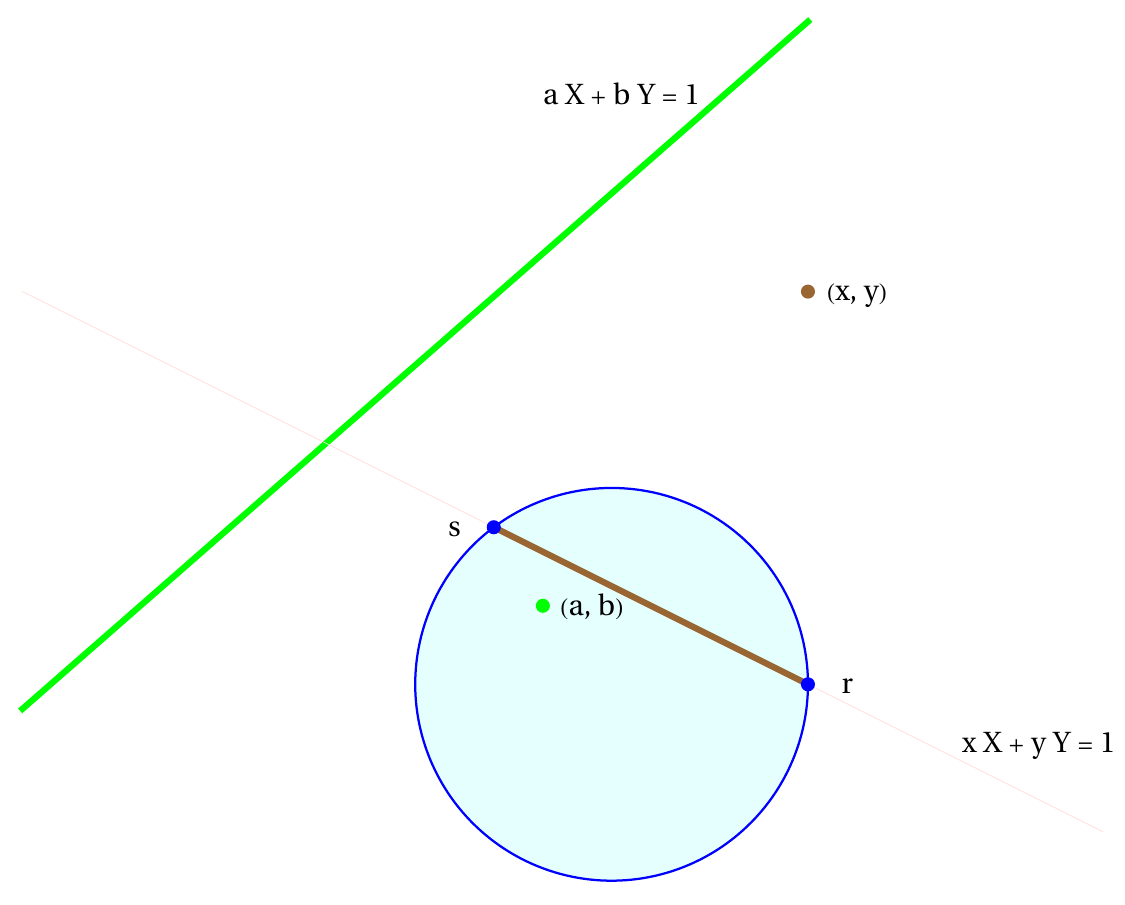}
\caption{\tiny{Distinguished lines in the correspondence space. \\ A line $l_{(a,b)}$ in $M^2$ corresponding to the point $(a,b)\in B^2$ (in green), and a cord $c_{(x,y)}$ of $B^2$ corresponding to a point $(x,y)\in M^2$ (in brown).}}\label{fii4i}
\end{figure}

Likewise, from the left side of Figure \ref{fi2}, we know that every point $P=(a,b)$ in the Beltrami space $B^2$ defines a line $\ell_P$ in the de Sitter space $M^2$, as the line consisting of all tips of the cones tangent at the circle $X^2+Y^2=1$ to the ends of all cords passing through $P=(a,b)$. This line is a set \be
\ell_{(a,b)}=\{(X,Y)\in \bbR^2~|~X a+Y b=1\}\subset M^2.\label{linelp}\ee
Note that because $(a,b)$ satisfy $a^2+b^2<1$, any such line does not intersect $X^2+Y^2=1$, and thus lies entirely in $M^2$.

\subsubsection{Transforming functions from $M^2$ to $B^2$} Having an integrable function $f:M^2\to\bbR$ defined in the de Sitter part $M^2$ of the plane we \emph{define its transformed counterpart} $\hat{f}: B^2\to \bbR$ as follows.
    Given $f$ on $M^2$ we need to specify the value of $\hat{f}(a,b)$ at every point $P=(a,b)$. For that we use the corresponding line $\ell_{(a,b)}$ in $M^2$. We define $\hat{f}(a,b)$ as the value of the line integral $\hat{f}(a,b)=\int_{\ell_{(a,b)}}f$, 
    of $f$ along the line $\ell_{(a,b)}$ in $M^2$. Explicitly, in our coordinates $(x,y,a,b)$ we have:
    \be\hat{f}(a,b)=\int_{-\infty}^{+\infty}f\Big(x,\frac{1}{b}-x\frac{a}{b}\Big)\sqrt{1+\frac{a^2}{b^2}}\der x.\label{traf}\ee

\subsubsection{Transforming functions from $B^2$ to $M^2$} Likewise, given an integrable function $F: B^2\to\bbR$ on the Beltrami part $B^2$ of the Beltrami--de Sitter model we get the value $\hat{F}(x,y)$ of its \emph{transform} $\hat{F}$ at the point $p=(x,y)$ in $M^2$ as the line integral $\hat{F}(a,y)=\int_{c_{(x,y)}}F$  of the function $F$ along the cord $c_{(x,y)}$, which in $B^2$ corresponds to $p=(x,y)$. Explicitly, this value is:
\be\hat{F}(x,y)=\int_{\frac{x-y\sqrt{x^2+y^2-1}}{x^2+y^2}}^{\frac{x+y\sqrt{x^2+y^2-1}}{x^2+y^2}}F\Big(a,\frac{1}{y}-a\frac{x}{y}\Big)\sqrt{1+\frac{x^2}{y^2}}\der a.\label{traF}\ee
These two transforms are \emph{analogs of the Radon transform}  \cite{Funk,Radon} for the Beltrami--de Sitter model. The first one enables to transform fields from the de Sitter spacetime to the Riemannian regime, and the second one enables to transform Riemannian fields to the spacetime.

\section{Correspondence space for the 2D Beltrami--de Sitter model}\label{sec3}
Using the same notation as in Section \ref{radon_sec}, we now consider \emph{pairs} $(p,P)$ of points on the plane $\bbR^2$ such that $p\in M^2$, i.e. it is a point in the external region $M^2=\{(X,Y)\in \bbR^2~|~X^2+Y^2>1\}$ of the unit circle $X^2+Y^2=1$, and $P\in B^2$, i.e it is a point in the internal region $B^2=\{(X,Y)\in\bbR^2~|~X^2+Y^2<1\}$ of the circle. The space $$T=M^2\times B^2$$ of all such pairs $\xi=(p,P)$ is 4-dimensional. Its points $\xi$ can be parameterized by $$\xi=(x,y,a,b)\quad\mathrm{with}\quad x^2+y^2>1\quad\mathrm{and}\quad a^2+b^2<1.$$
We call it the \emph{correspondence space} for the 2-dimensional Beltrami--de Sitter model. As any Cartesian product $T$ is a  trivial double fibration 
\be
\begin{aligned}
   \xymatrix{
        &T  \ar[dl]_{\pi_1} \ar[dr]^{\pi_2} & \\
     {M^2}&            & {B^2}\  }
    \end{aligned}\label{df}\ee
over $M^2$ and $B^2$, with projections $\pi_1$ and $\pi_2$ such that  
$\pi_1(x,y,a,b)=(x,y)$ and $\pi_2(x,y,a,b)=(a,b)$.

\begin{remark}\label{space_of_lines}
  Because every point $\xi=(x,y,a,b)$ in $T$ defines a line $l_{(a,b)}$ in $M^2$, as in \eqref{linelp}, then the manifold $T$ can be also viewed as the Cartesian product of the \emph{space of all points in the plane} $\bbR^2$ \emph{inside the circle} $X^2+Y^2=1$ and of the \emph{space of all lines in} $\bbR^2$, \emph{which do not intersect the disk} $X^2+Y^2\leq 1$.

  Similarly, $T$ can be considered as the Cartesian product of \emph{the space of all points in} $\bbR^2$ \emph{from outside the disk} $X^2+Y^2\leq 1$ and \emph{the space of all cords of the circle} $X^2+Y^2=1$.

  Another possible interpretation is that $T$ is the Cartesian product of the space of all lines outside the disk $X^2+Y^2\leq 1$ and all cords of the corresponding circle $X^2+Y^2=1$. 
\end{remark}

\subsection{Conformal structure} \label{sec210} We now show that the correspondence space $T$ has a natural \emph{conformal structure}. To see this we need some preparations.

Consider a curve $\xi(t)=(p(t),P(t))=(x(t),y(t),a(t),b(t))$ in $T$. Since $T$ is a (trivial) double fibration \eqref{df}, with $\pi_1(x,y,a,b)=(x,y)$ and $\pi_2(x,y,a,b)=(a,b)$, then the curve $\xi(t)\subset T$ defines a curve $p(t)=\pi_1(\xi(t))=(x(t),y(t))$ in $M^2$ and $P(t)=\pi_2(\xi(t))=(a(t),b(t))$ in $B^2$. The curve $p(t)$ describes a movement of a point $p$ in the de Sitter spacetime $M^2$, i.e. describes the worldline of this point there, and the curve $P(t)$ describes a movement of a point $P$ in the Beltrami (Riemannian signature) space.

Recall from Section \ref{bemek}, Fig. \ref{fi2}, or Fig. \ref{fii4i} in Section \ref{radon_sec}, that every point $p=(x,y)$ in the de Sitter space $M^2$ defines a cord $c_p$ in the Beltrami space $B^2$, which is given by
$$c_{p}=\{(X,Y)\in \bbR^2~|~xX+ yY=1\},$$
as in\eqref{cordcp}. Now, if the point $p=(x,y)$ in $M^2$ moves to $p+\der p=(x+\der x,y+\der y)$, then its corresponding cord $c_p$ moves to a new cord:
$$c_{p+\der p}=\{(X,Y)\in \bbR^2~|~(x+\der x)X+ (y+\der y)Y=1\},$$
and in turn defines a point $P_{p}$ in $B^2$, being the intersection of the infinitesimally moved cords $c_{p}$ and $c_{p+\der p}$. Again simple algebra shows that the coordinates of the cords' intersection point $P_p$ are:
$P_{p}=(\frac{\der y}{x\der y-y\der x},\frac{\der x}{y\der x-x\der y})$.

So a point $p(t)=(x,y)$, moving in the de Sitter spacetime $M^2$ with the velocity $v(t)=(\frac{\der x}{\der t},\frac{\der y}{\der t})=(\dot{x},\dot{y})$, defines a point $P_P(t)$ in the Beltrami space $B^2$, providing us with the assignment:
$$M^2\ni p(t)=(x,y)\quad \to \quad P_p(t)=\Big(\frac{\dot{ y}}{x\dot{ y}-y\dot{ x}},\frac{\dot{ x}}{y\dot{ x}-x\dot{ y}}\Big)\in B^2.$$

Likewise, recall from Section \ref{bemek}, Figure \ref{fi2}, that every point $P=(a,b)$ in the Beltrami space $B^2$ defines a line $\ell_P$ in the de Sitter space $M^2$ , as the line consisting of all tips of the cones tangent at the circle $X^2+Y^2=1$ to the ends of all cords passing through $P$. Here again a simple algebra brings the expression for this line: $$\ell_P=\{(X,Y)\in \bbR^2~|~X a+Y b=1\}.$$
\begin{figure}[h!]
\centering
\includegraphics[scale=0.42]{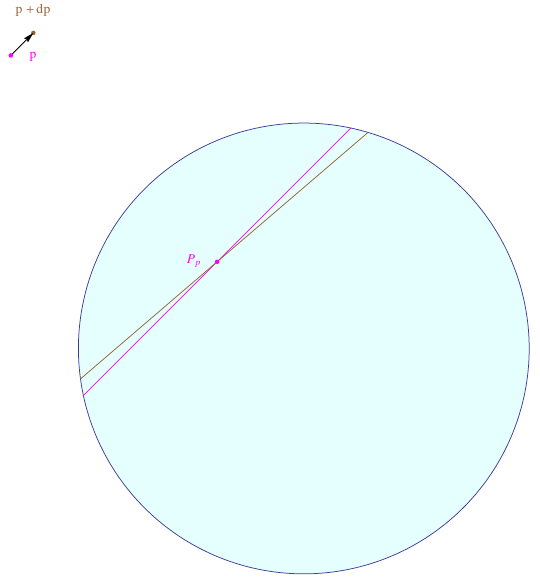}\hspace{1.5cm}\includegraphics[scale=0.6]{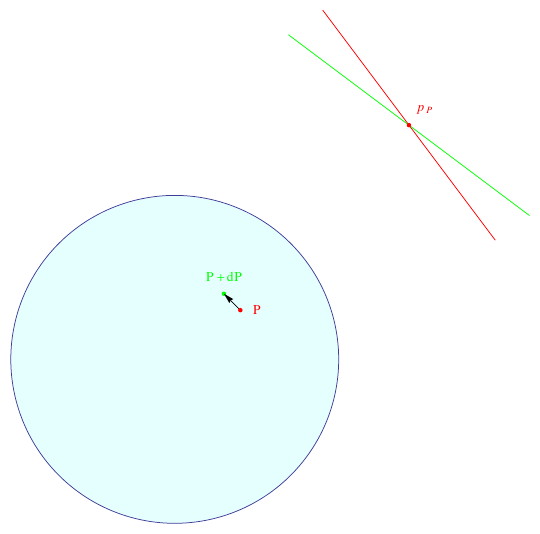}
\caption{\tiny{Construction of ghost points. \\{\bf On the left:} A point $P_p(t)\in B^2$ as obtained from $p(t)\in M^2$. \\ {\bf On the right:} A point $p_P(t)\in M^2$ as obtained from $P(t)\in B^2$.}}\label{fii4}
\end{figure}
Now, if the point $P=(a,b)$ in $B^2$ moves to $P+\der P=(a+\der a,b+\der b)$, then its corresponding line $\ell_P$ moves to another line:
$$\ell_{P+\der P}=\{(X,Y)\in \bbR^2~|~X (a+\der a)+Y (b+\der b)=1\},$$
and in turn defines a point $p_{P}$ in $M^2$, being the intersection of the infinitesimally moved lines $\ell_{P}$ and $\ell_{P+\der P}$. We derive the coordinates of the lines' intersection point $p_P$ to be:
$p_{P}=(\frac{\der b}{a\der b-b\der a},\frac{\der a}{b\der a-a\der b})$. 
So a point $P(t)=(a,b)$ moving in the Beltrami space $B^2$ with the velocity $V(t)=(\frac{\der a}{\der t},\frac{\der b}{\der t})=(\dot{a},\dot{b})$ defines a point $p_P(t)$ in the de Sitter spacetime $M^2$, and we have the following assignment:
$$B^2\ni P(t)=(a,b)\quad \to \quad p_P(t)=\Big(\frac{\dot{ b}}{a\dot{ b}-b\dot{ a}},\frac{\dot{ a}}{b\dot{ a}-a\dot{ b}}\Big)\in M^2.$$
\begin{figure}[h!]
\centering
\includegraphics[scale=0.6]{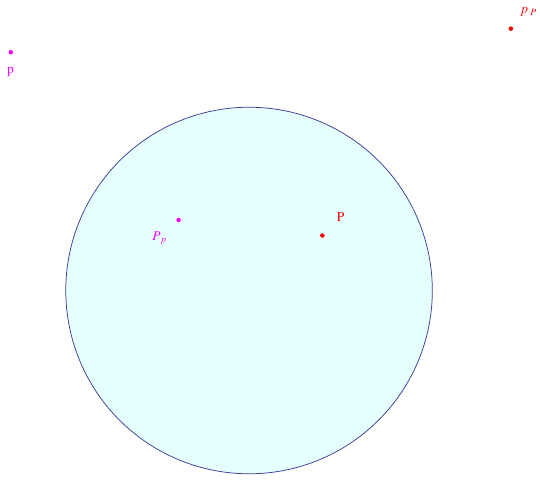}\hspace{1.5cm}
\caption{\tiny{A spacetime point $p\in M^2$ and its ghost $P_p$ in Riemannian $B^2$ (these two points are in magenta color), and a Riemannian point $P\in B^2$ and its ghost $p_P$ in spacetime $M^2$ (these two points are in red color). }}\label{fiii5}
\end{figure}
Now, we are in a position to define a `choreographic' movement of a pair of points $(p,P)$ with $p\in M^2$ and $P\in B^2$. We say that a pair $(p,P)$ \emph{dances} if the point $p$ moves in $M^2$ in the direction of the image $p_P$ of the point $P\in B^2$, and if the point $P$ moves in $B^2$ in the direction of the image $P_p$ of the point $p\in M^2$. A curve $\xi(t)=\Big(p(t),P(t)\Big)\subset M^2\times B^2$ with points $p(t)$ and $P(t)$ in a dance at every moment of time is a \emph{dancing pair}. This is a sort of a \emph{ghosts} dance: the point $p$ lives in another (Lorentzian) world than the Riemannian world point $P$. Point $p$ can not see $P$ because it can not penetrate the boundary of the circle $X^2+Y^2=1$. But it can see a `ghost image' $p_P$ of $P$ in its Lorentzian world. And, being attracted by this ghost, it tries to chase it. The same is true for the point $P$ in the Riemannian world $X^2+Y^2<1$. Point $P$ chases the ghost image $P_p$ of $p$ in its own world, although he can not see the other world $p$.
\begin{figure}[h!]
\centering
\includegraphics[scale=0.7]{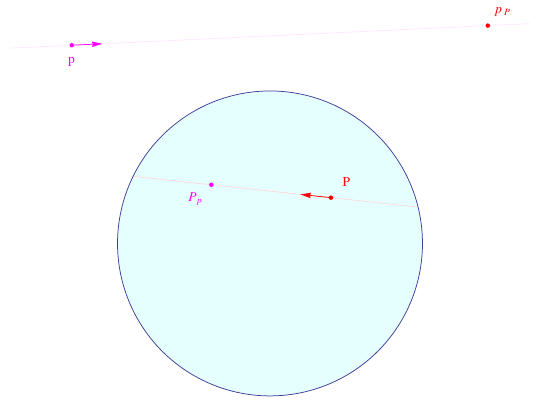}
\caption{\tiny{The dancing condition. \\ The point $p(t)\in M^2$ always moves towards the image $p_P(t)$ of the point $P(t)$ inside; also, the point $P(t)\in B^2$ always moves towards the image $P_p(t)$ of the point $p(t)$ outside.}}\label{fii5}
\end{figure}

Let us write the dancing condition for the dancing pair $\xi(t)=\Big(x(t),y(t),a(t),b(t)\Big)$ explicitly. We get: $$v(t)=\lambda \Big (p(t)-p_P(t)\Big)\quad\&\quad V(t)=\mu \Big(P(t)-P_p(t)\Big),$$
or infinitesimally and in coordinates
\be\begin{aligned}
  (\der x,\der y)\,\,=\,\,&\lambda \,\,\Big(x-\frac{\der b}{a\der b-b\der a},y-\frac{\der a}{b\der a-a\der b}\Big)\\&\\\& &\\&\\
  (\der a,\der b)\,\,=\,\,&\mu\,\, \Big(a-\frac{\der y}{x\der y-y\der x},b-\frac{\der x}{y\der x-x\der y}\Big)
\end{aligned}.\label{dcs}\ee
Eliminating $\lambda$ from the first of the above equations, or $\mu$ from the second equation we get:
$$\der a\, \Big(\,\,(1-by)\,\der x+bx\,\der y\,\,\Big)\,\,+\,\,\der b\, \Big(\,\,ay\,\der x+(1-ax)\,\der y\,\,\Big)\,\,=\,\,0.$$
This means that a dancing pair $\xi(t)=\Big(p(t),P(t)\Big)=\Big(x(t),y(t),a(t),b(t)\Big)$ moves along a \emph{null curve} in a \emph{conformal split signature} metric
\begin{equation}
  g_0=\,\,\der a\,\Big(\,\,(1-by)\,\der x+bx\,\der y\,\,\Big)\,\,+\,\,\der b\,\Big(\,\,ay\,\der x+(1-ax)\,\der y\,\,\Big).\label{dancemetp}\end{equation}
And this conformal metric, similarly as the Beltrami-de Sitter conformal metric $\der s_0^2$, is solely defined by the \emph{projective} structure of $\bbR^2$: we only used the notions of points, lines and conditions of incidence, such as points being on a line, and tangency. Thus the conformal class $[g_0]$ of $g_0$ must have at least $\slg(3,\bbR)$ conformal symmetry. Actually we have the following theorem.

\begin{theorem}\label{dancing}
  The space $T=B^2\times M^2$ consisting of pairs $(p,P)$ of points $p\in B^2=\{(X,Y)\in \bbR^2~|~X^2+Y^2<1\}$ and $P\in M^2=\{(X,Y)\in \bbR^2~|~X^2+Y^2>1\}$ is naturally equipped with a conformal class $[G]$ of an \emph{Einstein} split signature metric
  \be g=\frac{2\der a\,\Big(\,\,(1-by)\,\der x+bx\,\der y\,\,\Big)\,\,+\,\,2\der b\,\Big(\,\,ay\,\der x+(1-ax)\,\der y\,\,\Big)}{(1-a x-b y)^2}.\label{dancemet}\ee
  The dancing pairs  $\xi(t)=\Big(x(t),y(t),a(t),b(t)\Big)$ are null curves in this conformal class.

  The conformal class $[g]$ has $\sla(3,\bbR)$ as its full algebra of conformal Killing symmetries. These symmetries are all Killing symmetries of the metric $g$. The metric $g$ is locally isometric to the \emph{dancing metric}\footnote{This is a split-signature counterpart of the Fubini-Study metric from the Riemannian signature.} as defined in \cite{bor1}.
  \end{theorem}
\begin{remark}
  It is worthwhile to note that the first of the dancing conditions \eqref{dcs} - the one with the parameter $\lambda$ - is equivalent to the second one (the one with the parameter $\mu$).  
\end{remark}
\begin{remark}
  That the metric $g$ of the above theorem is isometric to the \emph{dancing metric} considered in Ref. \cite{bor1} is not a surprise. Recalling  Remark \ref{space_of_lines} and its first or the second interpretation of points of $T$, we see that our \emph{dancing condition} telling that a point in $M^2$ moves to a point defined by a moving line in $M^2$ is precisely the dancing condition defining the metric in \cite{bor1}. See the next section for further discussion of this. 
\end{remark}
\begin{remark}
  The metric $g$, as identified with the dancing metric, is \emph{para-Kahler-Einstein}, and belongs to the class of metrics considered in \cite{bor2}. As such it admits the \emph{para--Kahler potential} (see \cite{bor2} Proposition 2.14), i.e. a real function $V=V(x,y,a,b)$ on $\bbT$ such that
    $$g=\,2\der a\, \big(\,V_{ax}\der x\,+\,V_{ay}\der y\,\big)\,+\,2\der b \,\big(\,V_{bx}\der x\,+\,V_{by}\der y\,\big),$$
    where the capital letter $V$ with the double subscript $ax$ denotes second partial derivative of $V$ with respect to $a$ and $x$, $V_{ax}=\frac{\partial^2V}{\partial a\partial x}$, etc. 
    One can easily check that for our metric $g$ the para-Kahler potential is:
    $$V(x,y,a,b)=-\log(ax+by-1).$$
  \end{remark}
\begin{remark}
  In this section we considered a dance between points $p\in M^2$ and $P\in B^2$. If on the other hand we only had one point, say in $M^2$, moving along a curve $p(t)\subset M^2$ of points in the de Sitter space $M^2$, we could define yet another dance, requiring that the velocity of the point $p(t)$ is directed in the plane $\bbR^2$ to its ghost image $P_p(t)$ in the Beltrami space $B^2$. Explicitly, this condition would be:
  \be \dot{p}=\lambda (p-P_p).\label{dcs1}\ee
  As can be expected, this condition, when written explicitly in coordinates $p=(x,y)$, means that the curve $p(t)\in M^2$ is a null curve in the de Sitter metric $$\der s^2=\frac{(1-y^2)\der x^2+2xy \der x \der y+(1-x^2)\der y^2}{(1-x^2-y^2)^2}.$$
  This reflects the fact that only  when the point $p(t)$ moves along a straight line tangent to the boundary circle $S$ at some point $r$, its ghost point $P_p(t)$ remains on this line, as $P_p(t)=r$ for all times during such motion. Moreover, as we know, every such straight line $p(t)\in M^2$ is a null curve in the Beltrami metric.  
\end{remark}
\subsection{Correspondence space embeds in $(\bbR P^2)^*\times \bbR P^2$}
This section embeds the corresponding space $T$ in 4-dimensional manifolds larger-than-$T$.

\subsubsection{Embedding in $\bbR^4$}
  Note that although the metric $g$ is originally defined in $T=(B^2\times M^2)\subset \bbR^4$, we can trivially extend it to a 4-dimensional manifold $\tilde{T}=\bbR^4\setminus S$, where $S$ is the 3-dimensional singular locus
  \be S=\{(x,y,a,b)\in\bbR^4~|~ax+by=1\},\label{Ss}\ee
  which is the boundary of $\tilde{T}$, $$S=\partial\tilde{T}.$$

  Both metrics: $g$ as in \eqref{dancemet}, and $g_0$ as in \eqref{dancemetp}, are regular everywhere in $\tilde{T}$. Moreover, although the metric $g$ blows up at the boundary $\partial\tilde{T}=S$, its \emph{conformal class} $[g]$, or what is the same the conformal class $[g_0]$ of the metric $g_0$, is \emph{regular} at the closure
  $$\overline{T}=\tilde{T}\cup\partial\tilde{T}$$
  of $\tilde{T}$.
This is because the only singularity of the metric 
$g$
lies in its conformal factor

$$\tfrac{1}{(1-ax-by)^2},$$ which diverges  at $\partial\overline{T}=S$. However, this type of singularity in conformal geometry is considered irrelevant. It is merely a ‘coordinate singularity,’ caused by an unsuitable choice of the metric representative within the conformal class. There are numerous regular representatives of the conformal class $[g]$ in $\overline{T}$; one of them is given by by the metric $g_0$.    

The set $S=\partial \overline{T}$, which is the singular locus of the metric $g$ in $\overline{T}$, will be important in the Section \ref{rado}.

  \subsubsection{Further extension to $(\bbR P^2)^*\times \bbR P^2$}\label{furext}
  The extension $\overline{T}$ arises \emph{naturally} and independently of the main exposition as follows \cite{bor1}:

  Consider the space $V=(\bbR^3)^*\times \bbR^3$ with points $({\bf p},{\bf q})=(p_i,q^j)$ and the natural pairing $<{\bf p},{\bf q}>=p_iq^i$. Let $G$ be a symmetric bilinear form in $V$ defined by
  $$G=\frac{-2<{\bf p}\times\der{\bf p},{\bf q}\times\der{\bf q}>}{<{\bf p},{\bf q}>^2}.$$
  Here the symbol $\times$ denotes the cross product in the respective spaces $(\bbR^3)^*$ and $\bbR^3$,
  $${\bf p}\times\der {\bf p}=\epsilon^{kij}p_i\der p_j,\quad {\bf q}\times\der {\bf q}=\epsilon_{kij}q^i\der q^j.$$
 Using the standard identities from the vector calculus in $\bbR^3$ we can also write $G$ as:
\be G=2\frac{<{\bf p},\der{\bf q}><\der {\bf p},{\bf q}> - <{\bf p},{\bf q}><\der {\bf p},\der{\bf q}>}{<{\bf p},{\bf q}>^2}.\label{dancc}\ee

 This bilinear form is doubly degenerate on $V$, as in each pairs of the three components of the 1-forms ${\bf p}\times \der {\bf p}$ or${\bf q}\times\der {\bf q}$, only two are linearly independent at points of $V$. More specifically, the degenerate directions for $G$ in $V$ are the directions tangent to two Euler vector fields
 $$Y_p=<{\bf p},\nabla_{\bf p}>=p_i\partial_{p_i},\quad Y_q=<\nabla_{\bf q}, {\bf q}>=q^i\partial_{q^i}$$
 on $V$. It further follows that $G$ \emph{is preserved when Lie transported along} $Y_p$ and $Y_q$ in $V$. This is the infinitesimal version of the fact that $G$ is \emph{invariant} with respect to the action $$\phi_{(\lambda,\mu)}({\bf p},{\bf q})=(\lambda{\bf p},\mu{\bf q})$$ of the multiplicative group $$\bbR_0^2=\{\bbR^2\ni(\lambda,\mu)\,|\,\lambda\mu\neq 0\}\quad\mathrm{on}\quad V,$$ as we have: 
  $$\phi_{(\lambda,\mu)}^*G=\frac{-2<\lambda {\bf p}\times\der(\lambda {\bf p}),\mu {\bf q}\times\der(\mu {\bf q})>}{<\lambda {\bf p},\mu{\bf q}>^2}=G.$$
Thus, the bilinear form $G$  descends to the well defined bilinear form $g$ on the quotient space 
$${\mathcal T}=V/\bbR_0^2=(\bbR P^2)^*\times \bbR P^2$$
of orbits of $\bbR_0^2$. 
The descended bilinear form $g$ is nondegenerate, and has split signature $(+,+,-,-)$.

Due to the invariance of $G$ any 4-dimensional submanifold of $V$ transversal to the Euler vector fields $Y_p$ and $Y_q$ is equipped with a split signature metric isometric to $g$. Any such submanifold is also a good local representation of $\mathcal T$. For example, taking a 4-dimensional slice $p_1=a$, $p_2=b$, $p_3=-1$, $q^1=x$, $q^2=y$ and $q^3=1$ in $V$, i.e. an embedding  
$$\iota: \overline{T}\to V,\quad \mathrm{with} \quad \iota(x,y,a,b)=({\bf p},{\bf q})=(a,b,-1,x,y,1)\in V,$$
one gets
$$\iota^*(G)=\frac{2\der a\,\Big(\,\,(1-by)\,\der x+bx\,\der y\,\,\Big)\,\,+\,\,2\der b\,\Big(\,\,ay\,\der x+(1-ax)\,\der y\,\,\Big)}{(1-ax-by)^2},$$
  which \emph{isometrically} embeds the extended correspondence dancing metric space $(\overline{T}, g)$ in $V$. This gives a local one-to-one isometry between $(\overline{T},g)$ and $({\mathcal T},g)$.
  
  \begin{remark} \label{2.7} Regarding the Beltrami -- de Sitter model $(BdS,\der s^2)$ from Section \ref{bds} we have the following proposition.
    \begin{proposition}
      The Beltrami -- de Sitter model $(BdS,\der s^2)$ is isometrically embedded in the extended correspondence space $\overline{T}$ by
      $$ BdS\ni (x,y)\mapsto (x,y,a,b)=(x,y,x,y)\in\overline{T}.$$
      \end{proposition}
So the Beltrami -- de Sitter model $(BdS,\der s^2)$ seats in $\overline{T}$ on the diagonal, $$BdS\simeq \{\,\,(x,y,a,b)\in \overline{T}~|~a=x,\,\,b=y,\,\,(x,y)\in\bbR^2\setminus \bbS^1\,\,\}.$$ Here $\bbS^1=\{(x,y)\in\bbR^2\,|\,x^2+y^2=1\}$, which as a subset of $\overline{T}$ is the restriction of the boundary $S=\partial\overline{T}$ to the diagonal
$$S_{|BdS}=\{\,\,(x,y,a,b)\in\overline{T}~|~ax+by=1,\,\,a=x,\,\,b=y\,\,\}\simeq\bbS^1.$$
      One easily checks that the dancing metric $g$, when restricted to this diagonal, equips $BdS$ with the Beltrami metric, $g_{|BdS}=\der s^2$. \end{remark}
 
\subsection{Twistorial interpretation of the Radon-like transform} \label{rado} Having in mind the double fibration \eqref{df}
of the \emph{correspondence space} $T$,
with $\pi_1(x,y,a,b)=(x,y)$ and $\pi_2(x,y,a,b)=(a,b)$,  one can now give an interpretation of our transforms Radon--like \eqref{traf} and \eqref{traF} as a kind of \emph{Penrose transforms} \cite{eastwood} from the twistor theory. In short: pullback any function from $M^2$ (or $B^2$) to $T$ by $\pi_1^*$ (or by $\pi_2^*$), and integrate it over the \emph{distinguished curves} in the respective fibers $\pi_2^{-1}(a,b)$, or $\pi_1^{-1}(x,y)$, in $T$. What are the \emph{distinguished curves} in $T$? 
\subsubsection{Distinguished curves in $T$} These are the curves obtained by intersecting the singular locus $S_T$ of the dancing metric, with either fibers $\pi_1^{-1}(x,y)$ or fibers $\pi_2^{-1}(a,b)$.
More specifically, for every point $p=(x,y)\in M^2$, i.e. a point, whose coordinates satisfy $x^2+y^2>1$, we have a 1-dimensional set 
$$\gamma_p=\pi_1^{-1}(p)\cap S_T=\{(x,y,a,b)\,\,:\,\, ax +by=1\,\,\&\,\,a^2+b^2<1\}\subset T.$$
Likewise, for every point $P=(a,b)\in B^2$, i.e. such that $a^2+b^2<1$, we have a 1-dimensional set
$$\Gamma_P=\pi_2^{-1}(P)\cap  S=\{(x,y,a,b)\,\,:\,\, ax +by=1\,\,\&\,\,x^2+y^2>1\}\subset T.$$
The curves $\gamma_p$ and $\Gamma_P$ may be parameterized as follows:

If $y\neq 0$ in $p=(x,y)$, then we parameterize $\gamma_p$ by $a$: 
$$\gamma_p=\{(x,y,a,\frac{1}{y}-\frac{x}{y}a)\,\,\&\,\,s_x<a<r_x\},$$
where $s_x$ and $r_x$ are the $x$-components of the end points $s$ and $r$ of a cord $c_{(x,y)}$ corresponding to the point $p=(x,y)$, as in \eqref{cordcp}-\eqref{cordcp1}.

If $y=0$ in $p=(x,y)\in M^2$, then $x^2>1$, and we parameterize $\gamma_p$ by $b$:
$$\gamma_p=\{\,(x,0,\frac{1}{x},b)\,\,\&\,\,|b|<\sqrt{1-\frac{1}{x^2}}\,\}.$$
Note that the projection $\pi_2(\gamma_p)$ of any $\gamma$ curve to $B^2$ is the cord $c_p$ corresponding to $p\in M^2$ in $B^2$, as defined in \eqref{cordcp}.

Similarly, the curves $\Gamma_P$ are either parameterized by $x$ if $b\neq 0$ in $P=(a,b)\in B^2$:
$$\Gamma_P=\{(x,\frac{1}{b}-\frac{a}{b}x,a,b)\,\,\&\,\,x\in\bbR\}$$
or, if $b=0$ in $P=(a,b)\in B^2$, we parameterize by $y$: 
$$ \Gamma_P=\{(\frac{1}{a},y,a,0)\,\,\&\,\,y\in\bbR\}.$$ 
Now, note that the projection $\pi_1(\Gamma_P)$ of any $\Gamma$ curve to $M^2$ is the line $\ell_P$ corresponding to $P\in B^2$ in $M^2$, as defined in \eqref{linelp}.
\begin{figure}[h!]
\centering
\includegraphics[scale=0.4]{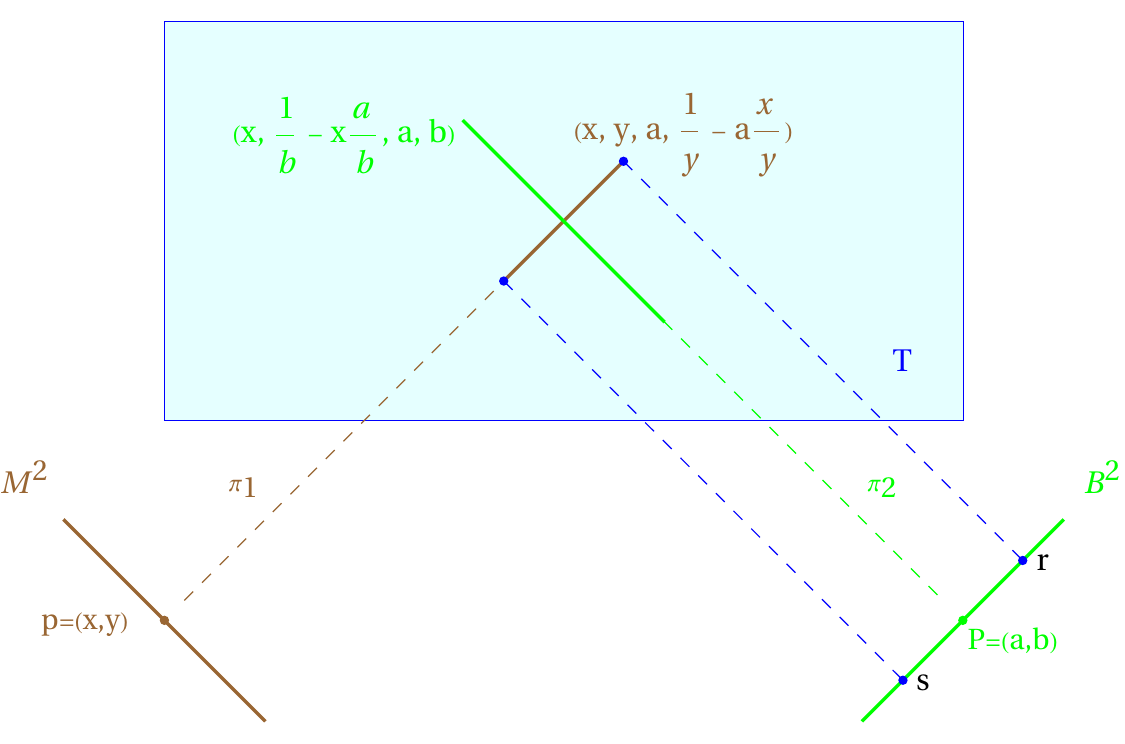}
\caption{\tiny{Radon--like transform. \\Integrate pulbacks of functions from $M^2$ or $B^2$ over the base along the distinguished curves in $M$. For functions on $M^2$, to get their transformed value from $p=(x,y)$ to $P=(a,b)$, integrate the pulbacks along the green line $\Gamma$ over the entire range of $x$s from $M^2$. Likewise, for functions on $B^2$, to get their transformed value from $P=(a,b)$ to $p=(x,y)$, integrate the pulbacks along the red line $\gamma$ over the entire range of $a$, which changes between the $x$ components of the tangency points of the cone tiped at $p=(x,y)$.}}\label{fii15}
\end{figure}

\subsubsection{The transform viewed as a Penrose transform} 
These distinguished curves enable to transform functions on $M^2$ into functions on $B^2$ and vice versa.

Indeed, suppose that we have an integrable function $f: M^2\to\bbR$ on $M^2$. Let us choose a point $p=(x,y)\in M^2$. The function $f$ has value $f(x,y)$ at this point. Knowing it, we want to know what is the transformed value $\hat{f}(a,b)$ of this function at a point $P=(a,b)\in B^2$.

To get this we first pullback $f$ by $\pi_1^*$ to a function on $T$, $(\pi_1^*f)(x,y,a,b)=f(\pi_1(x,y,a,b))$ and then we integrate this pullback along the distinguished curve $\Gamma_P=(x,\frac{1}{b}-x\frac{a}{b},a,b)$ in the fiber over the point $P=(a,b)$. In our parameterization of $\Gamma_P$, its parameter $x$ runs through the entire real line, $x\in\bbR$,  so the transformed value of $f(x,y)$ at point $P=(a,b)$ is given by:
$$\hat{f}(a,b)=\int_{-\infty}^{+\infty}f\Big(x,\frac{1}{b}-x\frac{a}{b}\Big)\sqrt{1+\frac{a^2}{b^2}}\der x.$$
Note that this is the same as the formula \eqref{traf} in which we defined this transform was defined in terms of geometry of the plane $\bbR^2=\overline{B}^2\cup M^2$.

Likewise, given an integrable function $F: B^2\to\bbR$ on $B^2$ we pullback it by $\pi_2^*$ to a function on $T$, $(\pi_2^*F)(x,y,a,b)=F(\pi_2(x,y,a,b))$, and then transform the value $F(a,b)$ of the function $F$ from point $P=(a,b)\in B^2$ to a value in a point $p=(x,y)\in M^2$ by integrating this pulback $\pi_2^*F$ along the distinguished curve $\gamma_p=(x,y,a,\frac{1}{y}-a\frac{x}{y})$ in the fiber over the point $p=(x,y)$. Thus the formula for the transformed value $\hat{F}(x,y)$ of the function $F:B^2\to \bbR$ is:
$$\hat{F}(x,y)=\int_{\frac{x-y\sqrt{x^2+y^2-1}}{x^2,y^2}}^{\frac{x+y\sqrt{x^2+y^2-1}}{x^2,y^2}}F\Big(a,\frac{1}{y}-a\frac{x}{y}\Big)\sqrt{1+\frac{x^2}{y^2}}\der a.$$
Again we recover the previously defined transform \eqref{traF}.
\section{Hidden ${\bf G}_2$ symmetry of the Beltrami--de Sitter model}\label{hidden}
Due to results of Ref. \cite{an}, Theorem \ref{dancing} enables us to associate the \emph{split real form of the exceptional simple Lie group} ${\bf G}_2$ with the 2D Beltrami--de Sitter model. This happens to be the \emph{symmetry of the twistor distribution} $\mathcal D$ on the \emph{circle twistor bundle} $\bbS^1\to  \bbT\to T$ over the correspondence space $T$. One can see how these are defined for general non-half-flat split-signature conformal metrics in 4-dimensions in \cite{an}. Here, for completeness of the presentation, we repeat in details the construction in the case of the conformal class $(T,[g])$ of the dancing metric $g$ on the correspondence space of the Beltrami--de Sitter model.
\subsection{The circle twistor bundle over the correspondence space}
Four-dimen\-sional metrics of signature $(+,+,-,-)$ have a
particular feature that, contrary to the Lorentzian or Riemannian
signatures, they admit \emph{real} 2-planes which are \emph{totally
  null}. This totally null condition means that the metric, when
restricted to a 2-plane is identically vanishing. In the case of our dancing metric $g$ given by \eqref{dancemet}, let us take for example, two vector fields $X=\partial_x$ and $Y=\partial_y$ on $T$. They define rank 2 distribution $D_0=\Span(X,Y)$ of 2-planes $D_{0\xi}$ spanned at each point $\xi$ by two vectors $X_\xi$ and $Y_\xi$ being the respective values of $X$ and $Y$ at $\xi$. The 2-planes $D_{0\xi}$ are real and totally null at every point $\xi\in T$ as $g(X,X)=g(Y,Y)=g(X,Y)=0$ at every point of $T$. We say that the rank 2-distribution $D_0$ itself is real and totally null. If we ask for any other rank 2-distribution on $T$ that is \emph{real totally null} and whose 2-planes have the \emph{same orientation} than $D_0$ we find that all such distributions must be of the form $D_u=\Span(X_u,Y_u)$, where
$$\begin{aligned}
  X_u=\partial_x+u^3\frac{(1-by)\partial_b-ay\partial_a}{(ax+by-1)^2},\quad
  Y_u=\partial_y+u^3\frac{bx\partial_b-(1-ax)\partial_a}{(ax+by-1)^2},
  \end{aligned}$$
with $u:T\to \bbR$ being an arbitrary function\footnote{For the convenience of later formulas, we prefer using the parameter $u^3$ rather than $u$.} on $T$, or of the form $D_\infty=\Span(X_\infty,Y_\infty)$, where
$$\begin{aligned}
  X_\infty=\frac{(1-by)\partial_b-ay\partial_a}{(ax+by-1)^2},\quad
  Y_\infty=\frac{bx\partial_b-(1-ax)\partial_a}{(ax+by-1)^2}.
  \end{aligned}$$
Note that for every $u\in \bbR\cup\{\infty\}$ we have $g(X_u,X_u)=g(X_u,Y_u)=g(Y_u,Y_u)=0$ at every point $\xi\in T$. Thus the distribution $D_u$ is real and totally null for every function $u$ on $T$, as claimed. 

Thus, at every point $\xi\in T$ we have a wealth of $\bbS^1=\bbR\cup\{\infty\}$ of totally null planes $D_{u\xi}$ having the same orientation as $D_{0\xi}$. The bundle of all these 2-planes over $T$ is by definition the \emph{circle twistor bundle} $$\bbS^1\to\bbT\stackrel{\pi}{\to} T$$ over the conformal manifold $(T,[g])$. Here, by $\pi$ we denoted the canonical projection, which  associates the base point $\xi\in T$ to every point $(\xi,D_{u\xi})\in \bbT$, $$\pi(\xi,D_{u\xi})=\xi \quad \forall\, u\in\bbR\cup\{\infty\}.$$

\subsection{Twistor distribution}\label{tdis} This bundle is naturally equipped with a rank 3 distribution ${\mathcal D}^1$ defined by the condition that the 3-planes ${\mathcal D}^1_{(\xi,D_{u\xi})}$ of ${\mathcal D}^1$ in $\bbT$ are such that they project to the corresponding 2-planes $D_{u\xi}$ in $T$,
$$\pi_*({\mathcal D}^1_{(\xi,D_{u\xi})})=D_{u\xi}.$$
It follows that that the \emph{canonical} distribution ${\mathcal D}^1$ on $\bbT$ \emph{uniquely defines a rank 2 distribution} $\mathcal D$ on $\bbT$ satisfying
\be [{\mathcal D},{\mathcal D}]={\mathcal D}^1.\label{twistdist}\ee
As can be easily checked, in our coordinates $(x,y,a,b,u)$ on the twistor bundle $\bbT$ the distribution $\mathcal D$ is given by
$${\mathcal D}\,=\,\Span\Big(\,X_u-\frac{u^4y}{(1-ax-by)^2}\partial_u,\,\,Y_u+\frac{u^4x}{(1-ax-by)^2}\partial_u\,\Big).$$
The distribution $\mathcal D$ is called the \emph{twistor distribution} for the conformal manifold $(T,[g])$. Note, that it follows from the defining condition \eqref{twistdist} for $\mathcal D$ that the twistor distribution encapsulates not only the zero-jet-data of the conformal metric $[g]$ but also data about its first jets.  

\subsection{The ${\bf G}_2$ symmetry associated with the Beltrami--de Sitter model}
One can check that the twistor distribution $\mathcal D$ has the property of being \emph{bracket generating}, i.e. satisfying
\be {\mathcal D}\subset [{\mathcal D},{\mathcal D}]\subset[{\mathcal D}, [{\mathcal D},{\mathcal D}]]=\mathrm{T}\bbT.\label{flag}\ee
It is therefore a \emph{$(2,3,5)$ distribution}, as the distributions ${\mathcal D}$, ${\mathcal D}^1=[{\mathcal D},{\mathcal D}]$ and ${\mathcal D}^2=[{\mathcal D},{\mathcal D}^1]=\mathrm{T}\bbT$ constituting the flag \eqref{flag} have the respective ranks $(2,3,5)$ everywhere on $\bbT$.

It turns out that the $(2,3,5)$ distributions have \emph{differential invariants}, and they can be diffeomorphically nonequivalent even locally \cite{cartan}. Their invariants are given by the curvature of a certain $\mathfrak{g}_2$-valued \emph{Cartan connection} \cite{nurdif}, where $\mathfrak{g}_2$ is the \emph{split real form of the exceptional simple Lie algebra} ${\bf G}_2$. The \emph{simplest} $(2,3,5)$ distribution \emph{has vanishing} curvature of this connection, and when this happens the distribution has the largest \emph{symmetry}. The (local) \emph{symmetry group} of a distribution $\mathcal D$ defined on a manifold $M$ is a group of (local) diffeomorphisms $\phi:M\to M$ such that $\phi_*{\mathcal D}\subset{\mathcal D}$; the Lie algebra of symmetries of the pair $(M,{\mathcal D})$ is the Lie algebra of vector fields $S$ on $M$ such ${\mathcal L}_SX\subset{\mathcal D}$ for every vector field $X\in {\mathcal D}$.

Returning to our twistor distribution $\mathcal D$ defined in the previous section we have the following theorem.
\begin{theorem}\label{the_dist_g2}
  Consider the Beltrami--de Sitter model $(BdS,\der s^2)$, its dancing conformal split-signature structure $(T,[g])$ and its circle twistor bundle $\bbS^1\to\bbT\to T$ with the twistor distribution $\mathcal D$. Then the Lie algebra of symmetries of $\mathcal D$ is the 14-dimensional split real form of the simple exceptional Lie algebra $\mathfrak{g}_2$.

  Explicitly, we have
  \begin{itemize}
  \item the Beltrami-de Sitter space $$BdS=\{\bbR^2\ni(X,Y)\,|\,X^2+Y^2\neq 1\},$$
  with
 $$\der s^2=\frac{(1-Y^2)\der X^2+2XY \der X \der Y+(1-X^2)\der Y^2}{(1-X^2-Y^2)^2},$$
\item the dancing conformal 4-manifold 
  $$T=B^2\times M^2,$$
  with $B^2=\{\bbR^2\ni (a,b)\,|\,a^2+b^2<1\}$, $M^2=\{\bbR^2\ni(x,y)\,|\,x^2+y^2>1\}$ and the conformal class $[g]$ represented by
    $$ g=\frac{2\,\der a\,\Big(\,\,(1-by)\,\der x+bx\,\der y\,\,\Big)\,\,+\,\,2\,\der b\,\Big(\,\,ay\,\der x+(1-ax)\,\der y\,\,\Big)}{(1-a x-b y)^2},$$
\item the twistor bundle
 $$\bbS^1\to\bbT\stackrel{\pi}{\to}T,$$
  of real totally null planes over $T$, parameterized by $(x,y,a,b,u)$, so that the planes $D_u=\Span(X_u,Y_u)$, with
  $$
    X_u=\partial_x+u^3 \frac{(1-by)\partial_b -ay\partial_a}{(ax+by-1)^2},\,\,
      Y_u=\partial_y+u^3\frac{bx\partial_b-(1-ax)\partial_a}{(ax+by-1)^2},$$
  from a fiber $\pi^{-1}(x,y,a,b)$, form a circle
    $
      \bbS^1\simeq\Big\{D_u\,\,|\,\,u\in\bbR\cup\{\infty\}\Big\}$, 
    \item and the twistor distribution
      $${\mathcal D}\,=\,\Span\Big(\,X_u-\frac{u^4y}{(1-ax-by)^2}\partial_u,\,\,Y_u+\frac{u^4x}{(1-ax-by)^2}\partial_u\,\Big),$$
      which has the following 14-vector fields $S_i$, $i=1,2,\dots,14$, as the Lie algebra $\mathfrak{g}_2$-symmetry generators:
      $$\small{\begin{aligned}
          S_1=&\partial_y-ab\partial_a-b^2\partial_b-bu\partial_u,\\
           S_2=&\partial_x-a^2\partial_a-ab\partial_b-au\partial_u,\\
           S_3=&\frac{2(ax+by-1)}{u}\partial_y-\frac{u^2(ax-1)}{ax+by-1}\partial_a-\frac{bx u^2}{ax+by-1}\partial_b-\frac{u^3x-abx-yb^2+b}{ax+by-1}\partial_u,\\
           S_4=&\frac{2(ax+by-1)}{u}\partial_x+\frac{ay u^2}{ax+by-1}\partial_a+\frac{u^2(by-1)}{ax+by-1}\partial_b+\frac{u^3y+aby-a}{ax+by-1}\partial_u,\\
           S_5=&3x\partial_x-3a\partial_a-u\partial_u, \\
           S_6=&3y\partial_y-3b\partial_b-u\partial_u, \\
           S_7=&\frac{2(ax+by-1)}{u}(x\partial_x+y\partial_y)+\frac{u^2}{ax+by-1}(y\partial_a-x\partial_b)+\partial_u,\\
           S_8=&\frac{(ax+by-1)}{2u^2}(b\partial_x-a\partial_y)+u(a\partial_a+b\partial_b)+\frac{u^2(2ax+2by-1)}{2(ax+by-1)}\partial_u,\\
           S_9=&\frac{bx(ax+by-1)}{2u^2}\partial_x+\frac{abxy+b^2y^2-ax-2by+1}{2u^2}\partial_y+u\partial_a+\frac{xu^2}{2(ax+by-1)}\partial_u,\\
           S_{10}=&\frac{abxy+a^2x^2-by-2ax+1}{2u^2}\partial_x+\frac{ay(ax+by-1)}{2u^2}\partial_y-u\partial_b-\frac{yu^2}{2(ax+by-1)}\partial_u,\\
           S_{11}=&y\partial_x-a\partial_b,\\
           S_{12}=&x\partial_y-b\partial_a,\\
           S_{13}=&x^2\partial_x+xy\partial_y-\partial_a,\\
           S_{14}=&xy\partial_x+y^2\partial_y-\partial_b.
      \end{aligned}}$$
  \end{itemize}
      \begin{remark}\label{symg2twi}
        Note that out of all of theses symmetries only \emph{eight}, namely $S_1,S_2,S_5$, $S_6,S_{11}$, $S_{12},S_{13}$ and $S_{14}$, project from $\bbT$ to the dancing space $T$. The eight symmetries that do project are the \emph{lifts of the conformal symmetries} of the dancing metric $g$ and, once they are projected to $T$, they span the full algebra of conformal symmetries of the conformal structure $(T,[g])$. As we claimed in Theorem \ref{dancing} this algebra is isomorphic to $\sla(3,\bbR)$.  
        \end{remark}
    
  \end{theorem}
\section{Beltrami--de Sitter turned inside out}\label{comp_tw}
\subsection{The inversion}\label{comp_two} We will now use the \emph{inversion with respect to the circle} to map the Beltrami disk $x^2+y^2<1$ to the region outside the disk, and the de Sitter region $x^2+y^2>1$ to the region inside the disk.
\begin{figure}[h!]
\centering
\includegraphics{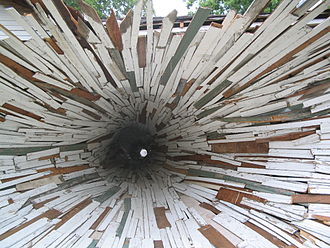}
\caption{\tiny{Inversion. A 2005 artwork by Dan Havel and Dean Ruck of Houston Alternative Art. This is a quote from \emph{Wikipedia}}}\label{fi5}
\end{figure}

\noindent
Geometrically the inversion transformation is defined on Figure \ref{fis4}.
\begin{figure}[h!]
\centering
\includegraphics[scale=0.4]{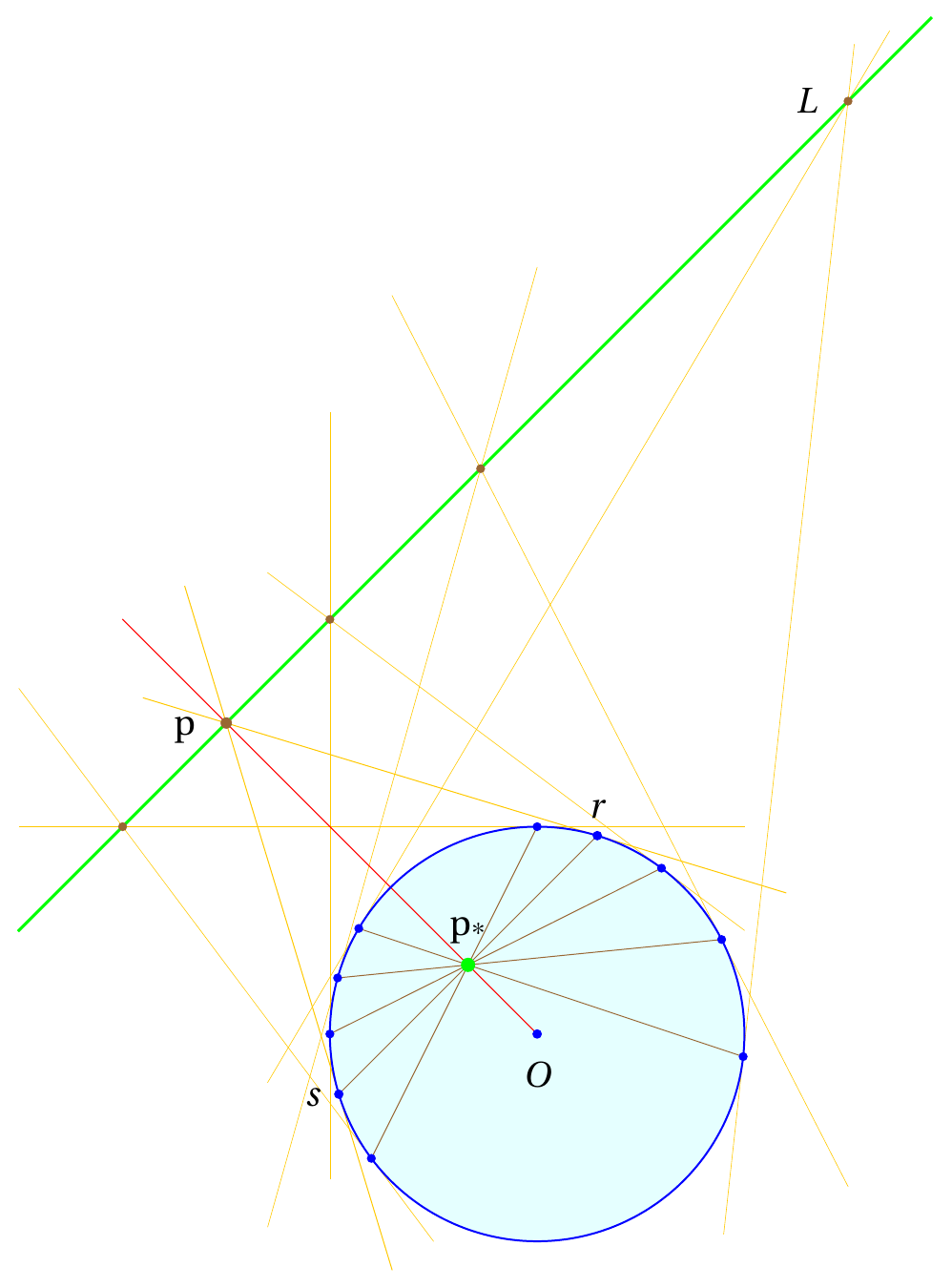}
\caption{\tiny{Inversion with respect to the circle. \\A point \( p \) outside the circle uniquely defines a point \( p^* \) inside the circle. This point \( p^* \) is the intersection of the line \( sr \), which joins the tangency points \( s \) and \( r \) of the cone from \( p \), with the line connecting \( p \) to the center \( O \) of the circle. Conversely, a point \( p^* \) inside the disk defines a line \( L \) outside the disk. This line \( L \) is the common locus of all the tips of the cones defined by all the chords passing through \( p^* \). The intersection of \( L \) with the line \( Op^* \), which joins the center \( O \) of the circle with \( p^* \), defines a point \( p \) outside the disk. The points \( p \) and \( p^* \) are inverses of each other with respect to the circle.}}\label{fis4}
\end{figure}

Looking at this figure, assuming that the circle of radius 1 is situated at the origin of coordinates and recalling that the tangency points $s$ and $r$ have coordinates as in \eqref{cordcp1}, we find that the coordinates $p=(x,y)$ and $p^*=(x^*,y^*)$ of the inversion--related points are:
\be (x^*,y^*)=(\frac{x}{x^2+y^2},\frac{y}{x^2+y^2})\quad\&\quad (x,y)=(\frac{x^*}{x^*\phantom{}^2+y^*\phantom{}^2},\frac{y^*}{x^*\phantom{}^2+y^*\phantom{}^2}).\label{inv}\ee
Note that this transformation is an \emph{identity} on the circle $x^2+y^2=1$.

We now make the \emph{inversion change of coordinates} $(x,y)=(\frac{x^*}{x^*\phantom{}^2+y^*\phantom{}^2},\frac{y^*}{x^*\phantom{}^2+y^*\phantom{}^2})$ in the Beltrami-de Sitter hybrid metric (\ref{beme}). After that we get:
$$\der s^2_{inv}=\frac{\big(1-\frac{(y^*)^2}{\big((x^*)^2+(y^*)^2\big)^2}\big)(\der x^*)^2+\frac{2x^*y^*}{\big((x^*)^2+(y^*)^2\big)^2}\der x^* \der y^*+\big(1-\frac{(x^*)^2}{\big((x^*)^2+(y^*)^2\big)^2}\big)(\der y^*)^2}{\big(1-(x^*)^2-(y^*)^2\big)^2}.$$

If we now allow the resulting $(x^*,y^*)$ coordinates to cover the Euclidean plane $\bbR^2$, we see that the metric $\der s^2_{inv}$  is \emph{Lorentzian inside the circle} $(x^*)^2+(y^*)^2=1$, and it is \emph{Riemannian outside this circle}. The regions of Riemannian and Lorentzian signatures got interchanged in the metric $\der s^2_{inv}$ on the transformed $\bbR^2$ with coordinates $(x^*,y^*)$, as compared to the metric $\der s^2$ given by formula \eqref{beme} in the original $\bbR^2$. We thus obtained the \emph{inverted} Beltrami -- de Sitter hybrid, which we call the \emph{de Sitter -- Beltrami model}.
\begin{figure}[h!]
\centering
\includegraphics[scale=0.42]{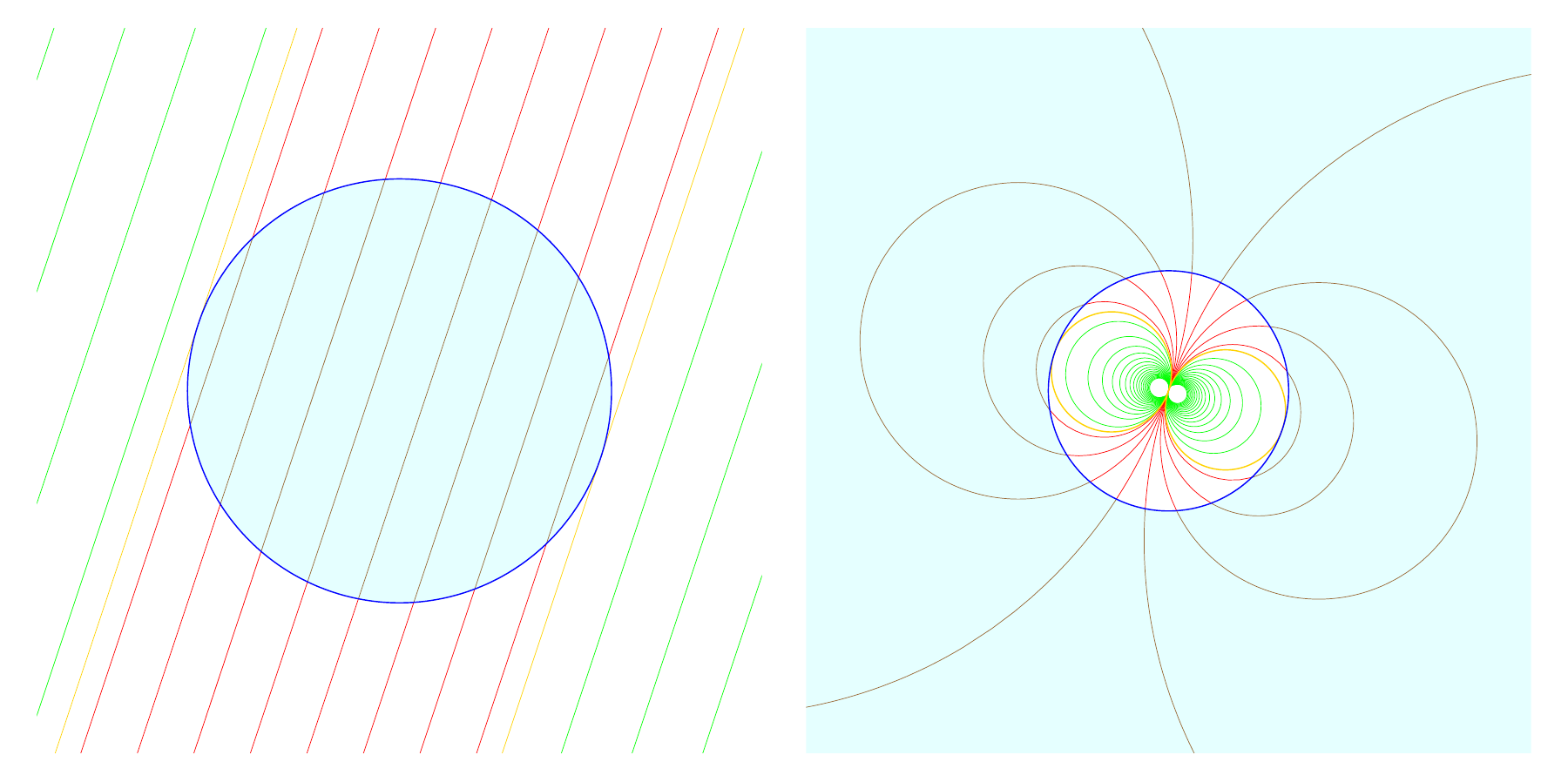}
\caption{\tiny{Comparing geodesics in the Beltrami -- de Sitter and the de Sitter -- Beltrami models; see Proposition \ref{nullrr} and the discussion at the beginning of Section \ref{51}.\\ {\bf On the left:} four types of geodesics (spacelike -- green, timelike -- red, and null -- orange; brown on the blue background are the usual geodesics in Beltrami space) in the model where the de Sitter space is outside the disk and the Beltrami space (in cyan) is inside.\\{\bf On the right:} the same four types of geodesics in the Beltrami--de Sitter model, turned inside out. Here, the Beltrami space is outside the disk (in cyan), and the de Sitter space is inside. In this de Sitter--Beltrami model, all geodesics are parts of circles centered at the center of the de Sitter white disk.\\What is absent from both the left and right pictures are specific timelike geodesics, which would appear as straight lines passing through the center of the cyan disk. These geodesics would be straight lines passing through the origin in both the Lorentzian--Riemannian and Riemannian--Lorentzian frameworks. (To visualize this in the right picture, consider the red circle as its radius approaches infinity.)\\Note also that in both pictures, the geodesics are smooth curves in the entire \( \mathbb{R}^2 \).}}\label{comp_geo}
\end{figure}
\begin{proposition}\label{nullrr}
  Geodesics curves in the de Sitter -- Beltrami model are the Euclidean circles passing through the origin in $\bbR^2$.
\end{proposition}
\begin{proof}
  The circles of radius $r$ centered at a point $(-r\sin\alpha,r\cos\alpha)$ in the de Sitter -- Beltrami $(x^*,y^*)$ plane are given by $(x^*+r\sin\alpha)^2+(y^*-r\cos\alpha)^2=r^2$, or by means of a simple algebra, are given by $(x^*)^2+(y^*)^2+2r(x^*\sin\alpha-y^*\cos\alpha)=0$.
  Transforming this equation by the inversion $(x^*,y^*)\mapsto (x,y)$, as in \eqref{inv}, to the $(x,y)$ plane of the Beltrami -- de Sitter model, we obtain \be
  x\sin\al-y\cos\al+\frac{1}{2r}=0,\label{nullr}\ee
  i.e. we obtain the equation of a generic straight line in the Beltrami -- de Sitter $(x,y)$ plane.
  \end{proof}
\begin{remark}
  In parallel with Remark \ref{2.7} the de Sitter -- Beltrami model is embedded in the \emph{extended corresponding space} $\overline{T}$ as a 2-dimensional subset
  $$\small{\begin{aligned}
    \mathrm{dSB}=&\Big\{
    \,\overline{T}\,\ni\,\big(
    \,\frac{x^*}{(x^*)^2+(y^*)^2},\,\frac{y^*}{(x^*)^2+(y^*)^2},\,\frac{x^*}{(x^*)^2+(y^*)^2},\,\frac{y^*}{(x^*)^2+(y^*)^2}\,
    \big)\\
    &\mathrm{where}\,\,(x^*,y^*)\in\bbR^2\setminus\bbS^1
    ,\,\,\mathrm{with}\,\, \bbS^1=\{(x^*,y^*)\,|\,(x^*)^2+(y^*)^2=1\}\,
      \Big\},
  \end{aligned}}$$
  on which the metric $\der s^2_{inv}$ is the restriction of the dancing metric $g$ from \eqref{dancemet}, i.e. $\der s^2_{inv}=g_{|\mathrm{dSB}}$.
  \end{remark}
\subsection{The de Sitter space inside the disk}
The inversion transformation \eqref{inv} used here provides a nice \emph{compact model} of the 2-dimensional de Sitter space. Contrary to the \emph{noncompact} model discussed in Section \ref{Belde}, in which the light cones are given by any two lines tangent to the circle $x^2+y^2=1$ and crossing at some point outside the circle, this compact model is now confined within the circle $(x^*)^2+(y^*)^2=1$, having the light cones consisting of pairs of intersecting semicircles of radius $\tfrac12$ centered at points of the circle $(x^*)^2+(y^*)^2=\tfrac14$.
\begin{figure}[h!]
\centering
\includegraphics[scale=0.4]{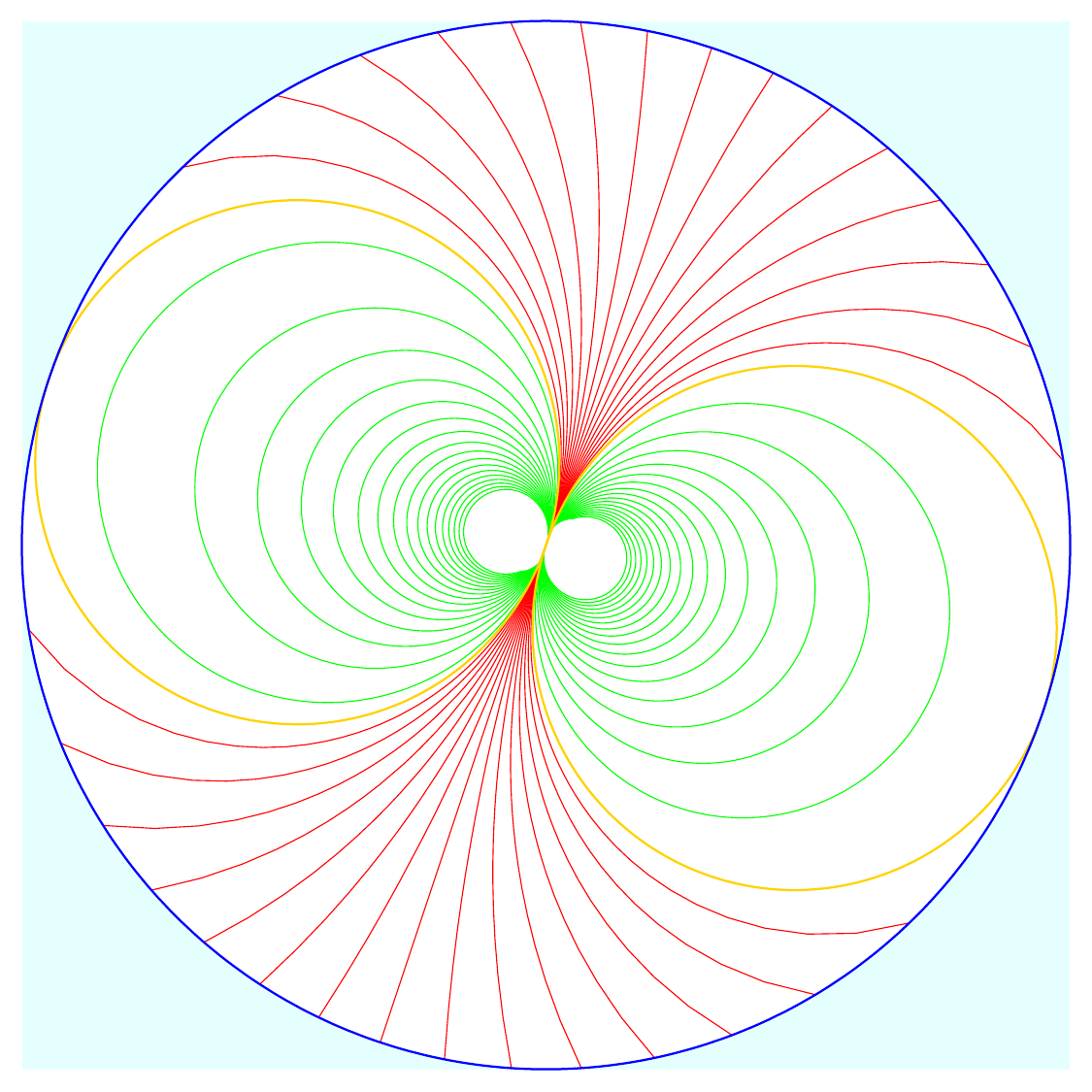}
\caption{\tiny{Geodesics in the de Sitter space wrapped in the unit disk.\\
Two orange half--circles form a null cone emanating from the de Sitter's Big Bag surface, which is at the center of the white disk. The red curves are the future oriented timelike geodesics, and the green curves are spacelike geodesics. All green, orange and red curves are parts of circles passing through the center of the white unit disk.
}}\label{comp_des1}
\end{figure}
\subsection{The Beltrami space outside the disk} The inversion gives also an interesting description of the Beltrami -- Cayley -- Klein \emph{noneuclidean geometry} \cite{bel1,bel2,cayley,klein}. One realizes it as the \emph{complement of the unit disk in the plane} with \emph{geodesics as parts of circles passing through the center of the disk}, see Figure \ref{belge}. We also make an analogous picture for one of the three possibilities of geodesics, namely the \emph{null geodesics}, in both de Sitter spaces, inside and outside the disk, on Figure \ref{comp_des}.

\begin{figure}[h!]
\centering
\includegraphics[scale=0.42]{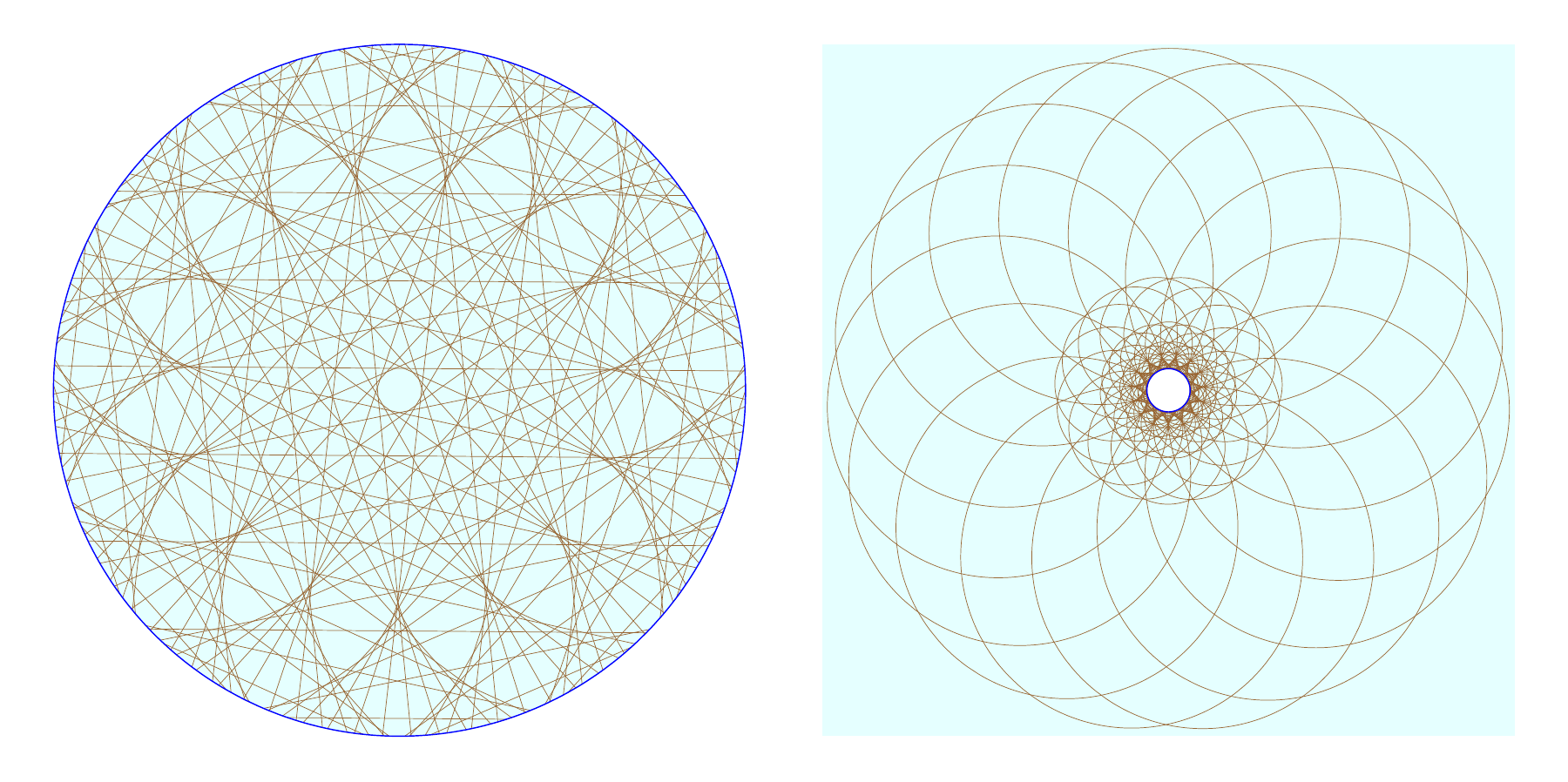}
\caption{\tiny{Beltrami -- Cayley -- Klein model inside and outside the unit disk.\\
    {\bf On the left:} The brown lines are 120 chords in the unit disk. These are 120 geodesics in the original Beltrami -- Cayley -- Klein model. The picture was made by rotating eight brown chords depicted in Figure \ref{comp_geo} fifteen times by multiples of 24 degrees around the origin. It is why the picture has 15-fold symmetry.\\
    {\bf On the right:} The same 120 Beltrami geodesics as seen in the complement of the unit disk. Now the brown curves are parts of circles passing through the center of the picture. They constitute 120 geodesics in the Beltrami metric $\der s^2_{inv}$ defined in the complement of the unit disk.}}\label{belge}
\end{figure}
\begin{figure}[h!]
\centering
\includegraphics[scale=0.4]{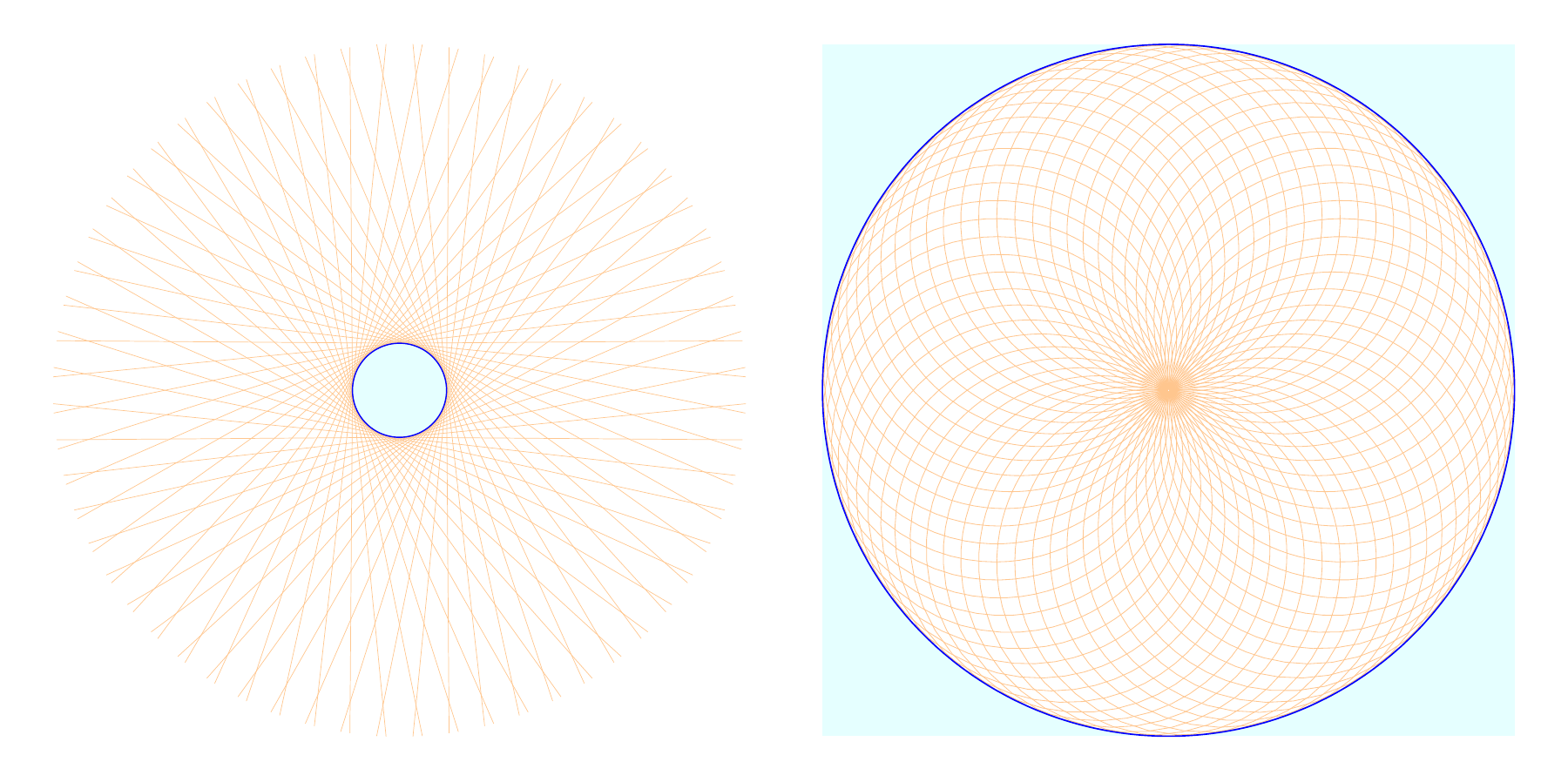}
\caption{\tiny{Comparison of light cone structures in de Sitter spaces as viewed from outside and inside the unit disk.\\
{\bf On the left:} The de Sitter space outside the unit disk. The light cones are represented by pairs of intersecting straight lines tangent to the unit circle, depicted as the blue circle in the center of the image.\\
{\bf On the right:} The de Sitter space inside the unit disk. This image is an inversion, with respect to the blue circle, of the image on the left. Here, the light cones appear as pairs of intersecting semicircles with a radius of $\tfrac12$, centered on points along a circle of radius $\tfrac12$, which is concentric with the blue circle of radius 1.\\ The arrow of time in the left picture is tangent to the \emph{inward} radial directions, while in the right picture it is tangent to the \emph{outward} radial directions.\\
In both images, the blue circles represent the \emph{null scri} -- the set where null geodesics terminate. The circle at infinity in the left image and the center of the right image are \emph{spacelike} sets, marking the origins of both the left and right universes. 
}}\label{comp_des}
\end{figure}


\subsection{Four possibilities for the 2D Beltrami--de Sitter hybrid}
On top of the Beltrami -- de Sitter and the de Sitter -- Beltrami models, which have their respective metrics $\der s^2$ or $\der s^2_{inv}$ and their visual realizations as in Figure \ref{comp_geo}, we can also make two other hybrids starting from the initial Beltrami -- de Sitter construction.  

This is done by \emph{replacing the de Sitter metric outside the unit disk} of the original construction \emph{with the outside part of the metric} $\der s^2_{inv}$, or by \emph{replacing the Beltrami metric inside the unit disk} of the original construction \emph{by the inside part of the metric} $\der s^2_{inv}$. 
This process brings two more conformal structures in \( \mathbb{R}^2 \): a \emph{purely Riemannian} one in the entire plane in the former case, and a \emph{purely Lorentzian} one in the entire plane in the later case.

More specifically, on $\bbR^2\setminus \bbS^1$ we have:
\begin{itemize}
\item  the Riemannian metric
$$
\der s^2_{RR} = \left\{
\begin{array}{ll}
\frac{(1-y^2)\der x^2+2xy \der x \der y+(1-x^2)\der y^2}{(1-x^2-y^2)^2} & \text{if }x^2+y^2<1\\
&\\
\frac{(1-\frac{y^2}{(x^2+y^2)^2})\der x^2+\frac{2xy}{(x^2+y^2)^2}\der x \der y+(1-\frac{x^2}{(x^2+y^2)^2})\der y^2}{(1-x^2-y^2)^2} & \text{if }x^2+y^2>1
\end{array} \right\},
$$
\item and the Lorentzian metric:
$$
\der s^2_{LL} = \left\{
\begin{array}{ll}
\frac{(1-\frac{y^2}{(x^2+y^2)^2})\der x^2+\frac{2xy}{(x^2+y^2)^2}\der x \der y+(1-\frac{x^2}{(x^2+y^2)^2})\der y^2}{(1-x^2-y^2)^2} & \text{if } x^2+y^2<1\\
&\\
\frac{(1-y^2)\der x^2+2xy \der x \der y+(1-x^2)\der y^2}{(1-x^2-y^2)^2} & \text{if } x^2+y^2>1
\end{array} \right\}.
$$
\end{itemize}
These, together with the \emph{Beltrami--de Sitter model}, whose metric structure is defined on $\bbR^2\setminus \bbS^1$ by $\der s^2$, and the \emph{de Sitter--Beltrami model}, whose metric structure is defined on $\bbR^2\setminus \bbS^1$ by $\der s^2_{inv}$, provide the \emph{four possibilities} of the $\bbR^2$ conformal structures  mentioned in the title of this section.

\section{The Lorentzian--Lorentzian case as a model for Penrose's CCC}\label{sec6} 
In Penrose's Conformal Cyclic Cosmology (CCC) the Universe is modeled on a sequence of eons, each of which is a Lorentzian $4-$manifold and such that the next eon is \emph{conformally} matched through its Big Bang spacelike surface with the Big Crunch spacelike surface of the previous eon. In this sense our Lorentzian hybrid of two de Sitter 2-dimensional manifolds equipped with the metric $\der s^2_{inv}$ and considered in the previous section gives a 2-dimensional model of a Universe from Penrose's CCC.

In this section we analyze this model in more detail. We start with the  Lorentzian metric $\der s^2_{LL}$ and first consider a metric $\der s^2_{0LL}$ conformally related to it, given by:
$$\der s^2_{0LL}~=~(1-x^2-y^2)^2~\der s^2_{LL}.$$
We have the following proposition. 
\begin{proposition}\label{501}
The Lorentzian metric 
\be
\der s^2_{0LL} = \left\{
\begin{array}{ll}
(1-\frac{y^2}{(x^2+y^2)^2})\der x^2+\frac{2xy}{(x^2+y^2)^2}\der x \der y+(1-\frac{x^2}{(x^2+y^2)^2})\der y^2 & \text{if } x^2+y^2\leq 1\\
&\\
(1-y^2)\der x^2+2xy \der x \der y+(1-x^2)\der y^2 & \text{if } x^2+y^2>1,
\end{array} \right\}\label{g0l}\ee
and its second fundamental form 
\be
K_{0LL}=\left\{
\begin{array}{ll}
-\frac{y^2}{(x^2+y^2)^{3/2}}\der x^2+\frac{2xy}{(x^2+y^2)^{3/2}}\der x \der y-\frac{x^2}{(x^2+y^2)^{3/2}}\der y^2 & \text{if } x^2+y^2\leq 1\\
&\\
-\frac{y^2(1-2(1-x^2-y^2))}{(x^2+y^2)^{3/2}}\der x^2+\frac{2xy(1-2(1-x^2-y^2))}{(x^2+y^2)^{3/2}}\der x \der y-\frac{x^2(1-2(1-x^2-y^2))}{(x^2+y^2)^{3/2}}\der y^2 & \text{if } x^2+y^2>1,
\end{array} \right\}\label{K0l}\ee
are continuous across the circle $x^2+y^2=1$, and as such are well defined Lorentzian metrics everywhere on the plane, except the point $(x,y)=(0,0)$. 
\end{proposition}
\begin{proof}
The continuity across $x^2+y^2=1$ follows directly from the formulas (\ref{g0l}), (\ref{K0l}). For completeness we recall the definition of the second fundamental form. We have: $$K_{0L}=K_{\mu\nu}~\der x^\mu\der x^\nu,\quad\quad (x^\mu)=(x,y),$$ and 
$$K_{\mu\nu}=h_{\rho \mu} h_{\sigma\nu}\nabla^\rho n^\sigma, $$
where 
$$n=n^\mu\partial_\mu=\frac{x}{(x^2+y^2)^{1/2}}\partial_x+\frac{y}{(x^2+y^2)^{1/2}}\partial_y$$ is the unit normal to the circle $x^2+y^2=1$, and 
$$h=h_{\mu\nu}~\der x^\m\der x^\nu$$ defined via:
$$g_{0L}~=~(h_{\mu\nu}+n_\mu n_\nu)~\der x^\mu\der x^\nu$$
is the projector onto the space orthogonal to $n$. Using this definition we calculate $K_{0L}$ and obtain (\ref{K0l}). 
\end{proof}
Thus, the Lorentzian hybrid of two de Sitters living in $\bbR^2$ and given by the metric $\der s^2_{LL}$ is a 2-dimensional model of the Penrose's CCC. Let us call $\bbR^2$ with the metric $\der s^2_{LL}$ as the \emph{de Sitter -- de Sitter model}. 
\subsection{Null geodesics in the de Sitter--de Sitter model}\label{51}
Among all geodesics, which in the de Sitter space of the Beltrami -- de Sitter model are described by the equation \eqref{nullr}, the ones that have only one point of contact with the circle $x^2+y^2=1$, namely the \emph{null geodesics},  must satisfy $r=1/2$. Reading back the proof of Proposition \ref{nullrr} one thus see that the null geodesics inside the unit disk of the de Sitter -- Beltrami model governed by the metric $\der s^2_{inv}$ satisfy
$$(x^*+\tfrac12 \sin\al)^2+(y^*-\tfrac12\cos\al)^2=\tfrac14.$$
Thus, the null geodesics inside the unit disk of the de Sitter -- de Sitter model, are the \emph{circles of radius} $\tfrac12$ centered at points belonging the circle concentric with the disk and having radius $\tfrac12$. Outside the disk the null geodesics are the straight lines tangent to it. See Figure \ref{mw}.  

\begin{figure}[h!]
\centering
\includegraphics[scale=0.2]{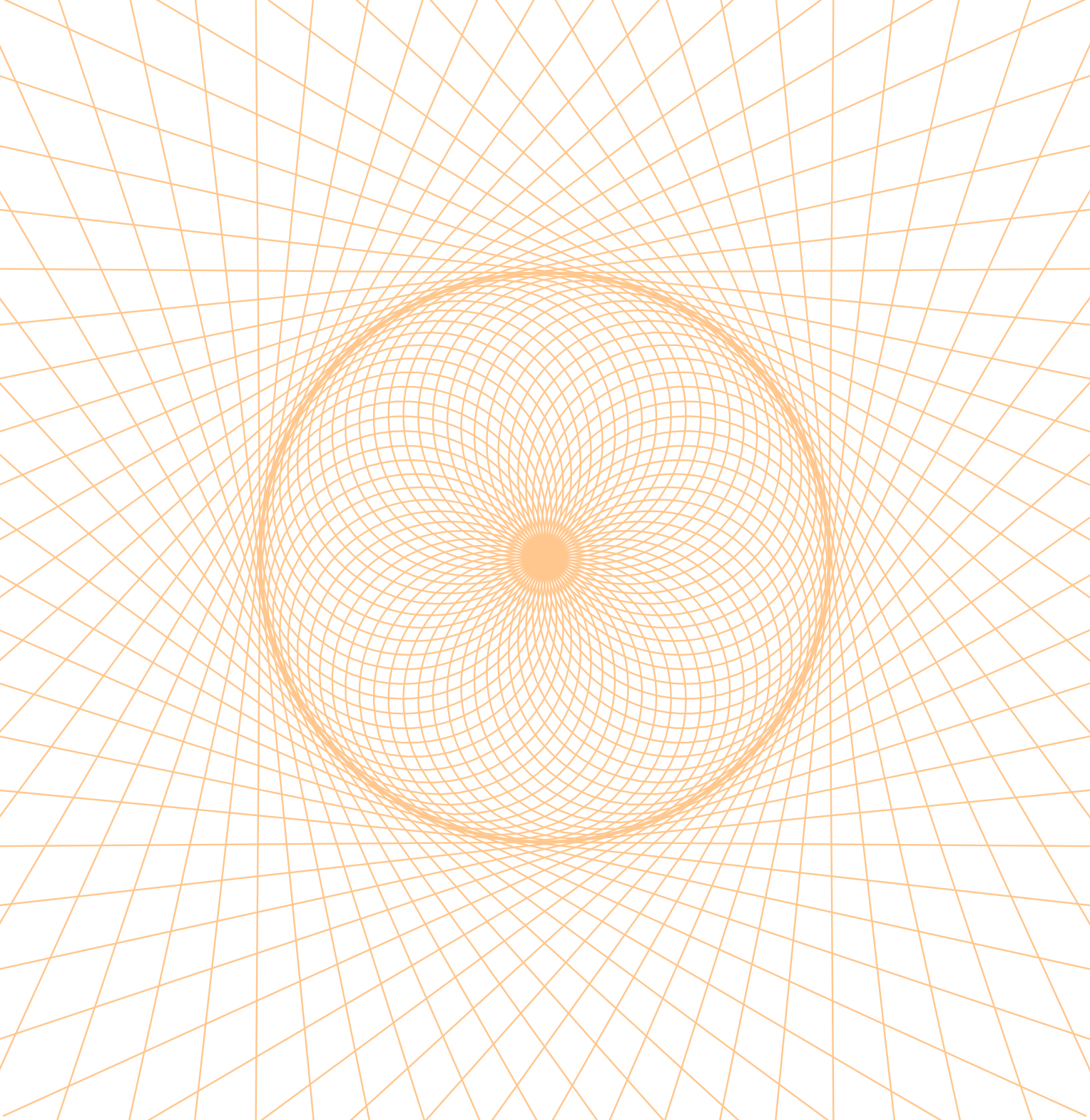}
\caption{\tiny{The Lorentzian-Lorentzian hybrid.\\ The orange light rays emanate from the Big Bang surface in the center, then they approach their envelope -- the orange unit circle, which is considered as the Big Crunch of this eon in the Universe. Reaching the envelope of the Big Crunch, the light rays can continue their life within the unit circle violating the arrow of time there, or they can escape to the next eon becoming the straight lines tangent to the envelope there. 
}}\label{mw}
\end{figure}
We should mention that the inversion of a null geodesic in the outside-the-disk de Sitter space is the semicircle passing through the tangency point on the unit disk and the origin, which forms an acute angle (i.e. $0^{\circ}$) with the light ray outside the unit disk; see Figure \ref{circgeo}.
\begin{figure}[ht!]
\centering
\includegraphics[scale=0.22]{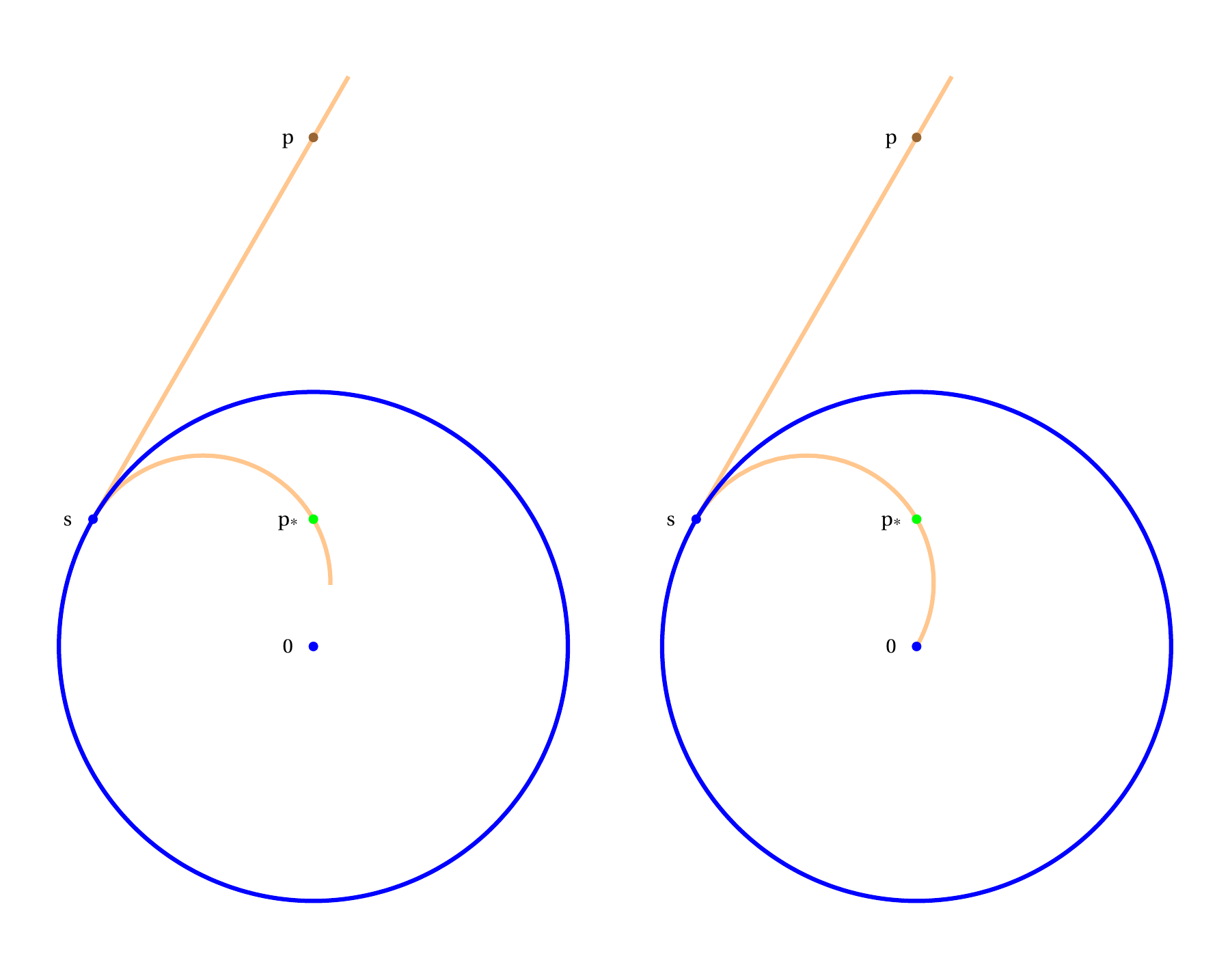}
\caption{\tiny{A light ray outside the disk and its image by the inversion, which is again a light ray in the de Sitter -- de Sitter model. \\
    {\bf On the left:} A null geodesic passing through the brown point $s$ in the external region, and its inversion in the internal region of the disk. The orange line in the disk does not reach the central point $0$, because we inverted the outside line only beyond point $p$ to the limit shown on the picture. \\ {\bf On the right:} An entire null geodesic inside the disk. Null geodesics in the internal region of the disk are semicircles of radius $1/2$ tangent to the boundary circle of the de Sitter boundary disk $x^2+y^2=1$. The middle point $0$ of the disk is the image of the null infinity (depending on the interpretation: in the past or in the future) of the external region.}}\label{circgeo}
\end{figure}

It is worthwhile to note that in the de Sitter -- de Sitter hybrid spacetime trajectories of photons - curves corresponding to null geodesics - can split at \emph{scries}, i.e. at the circle $x^2+y^2=1$ or at its center $0$, or at the circle at infinity, as at Figure \ref{fig12}, where a possible path of a photon in the de Sitter -- de Sitter model is discussed.
\begin{figure}[ht!]
\centering
\includegraphics[scale=0.3]{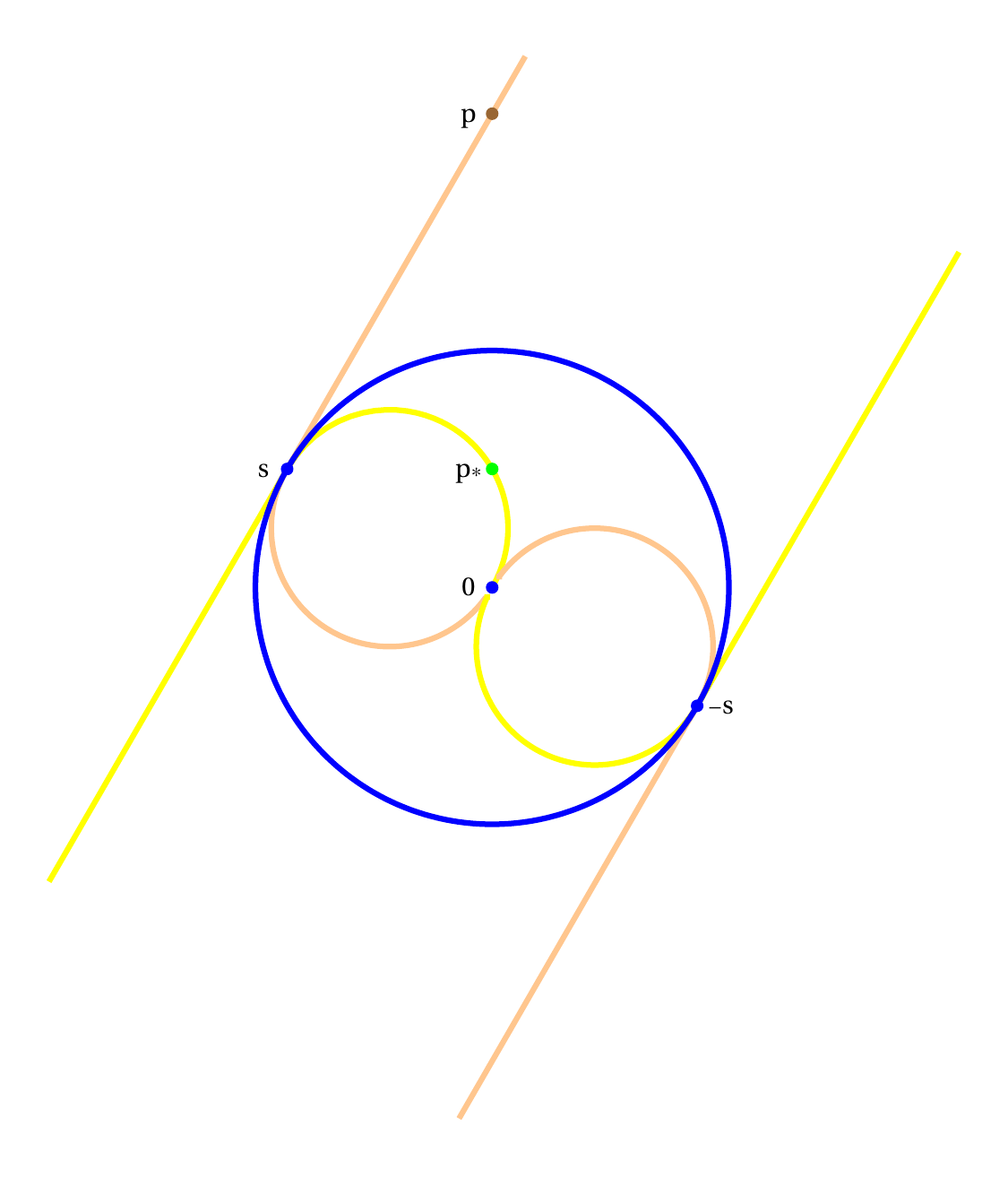}
\caption{\tiny{A possible travel of a photon in the Lorentzian--Lorentzian hybrid.\\
A photon emitted at a point 
$p$ in this figure, after reaching the boundary of the de Sitter circle at the tangency point $s$, can either continue its journey in the external de Sitter space along the yellow geodesics or enter the internal region of the disk and travel along the orange null geodesic inside the disk. Since traveling along the yellow geodesic would violate the arrow of time in the external spacetime (unless one subscribes to the idea of photons traveling back in time), we assume that at 
$s$, the photon chooses the orange world line and continues its travel along it.  Traveling along the orange null geodesic inside the disk, the photon eventually reaches the center of the disk at $0$. For timelike curves, the journey from $s$ to 
$0$ would take infinite time, but for a photon, it is essentially instantaneous. At 
$0$, the photon faces another choice: it can continue either along the yellow or the orange null geodesic, either to the left or to the right, or even reappear on one of the yellow/orange straight world lines at infinity in the external region.\\
For instance, if the photon chooses the arrow-of-time-violating orange null geodesic, it reaches the boundary of the de Sitter circle again at the point $-s$. Here, it must again decide whether to remain in the internal region and follow the yellow null geodesic or exit the interior and follow the orange straight line. If it chooses the orange line, it may eventually reappear either at the blue point at the center or at the endpoint of one of the orange/yellow straight lines in the exterior. This process can repeat indefinitely.}}\label{fig12}
\end{figure}
\begin{remark}
Note that the possible path of a photon, described by the orange trajectory passing through points 
$p$, $s$, $0$, and $-s$ in Figure \ref{fig12}, forms a geodesic that is piecewise smooth, but not very smooth at the switching points. For instance, at the point 
$s$, the path is merely $C^1$--continuous. This highlights a significant difference between the present setting and the Lorentzian–Riemannian Beltrami -- de Sitter model. In this model, the Euclidean straight lines in the plane serve as geodesics in both the Lorentzian de Sitter $M^2$ and the Riemannian Beltrami $B^2$. These geodesics are smooth curves -- indeed, smooth conformal geodesics -- in both regions as well as at the boundary between them. The same is true for the de Sitter -- Beltrami model: the Euclidean circles passing through the origin in $\bbR^2$ are smooth curves and serve as smooth conformal geodesics in the conformal class of the metric $\der s^2_{inv}$. This distinction is why we consider the original Beltrami -- de Sitter model, despite being Riemannian -- Lorentzian, to be significantly more physical than the Lorentzian–Lorentzian CCC-like  hybrid discussed in this Section.

Consequently, we favor the Lorentzian–Riemannian type of Conformal Cyclic Cosmology over the purely Lorentzian version.
  \end{remark}

\subsection{Geodesics in the Beltrami -- Beltrami model}

The Beltrami -- Beltrami model is $\bbR^2$ equipped with the metric $\der s^2_{RR}$. One can analyze its smoothness properties obtaining an analogous result as in Proposition \ref{501}. Here we only present a picture of a grid of geodesics in this model similar to the Lorentzian grid given at Figure \ref{mw}.
\begin{figure}[h!]
\centering
\includegraphics[scale=0.53]{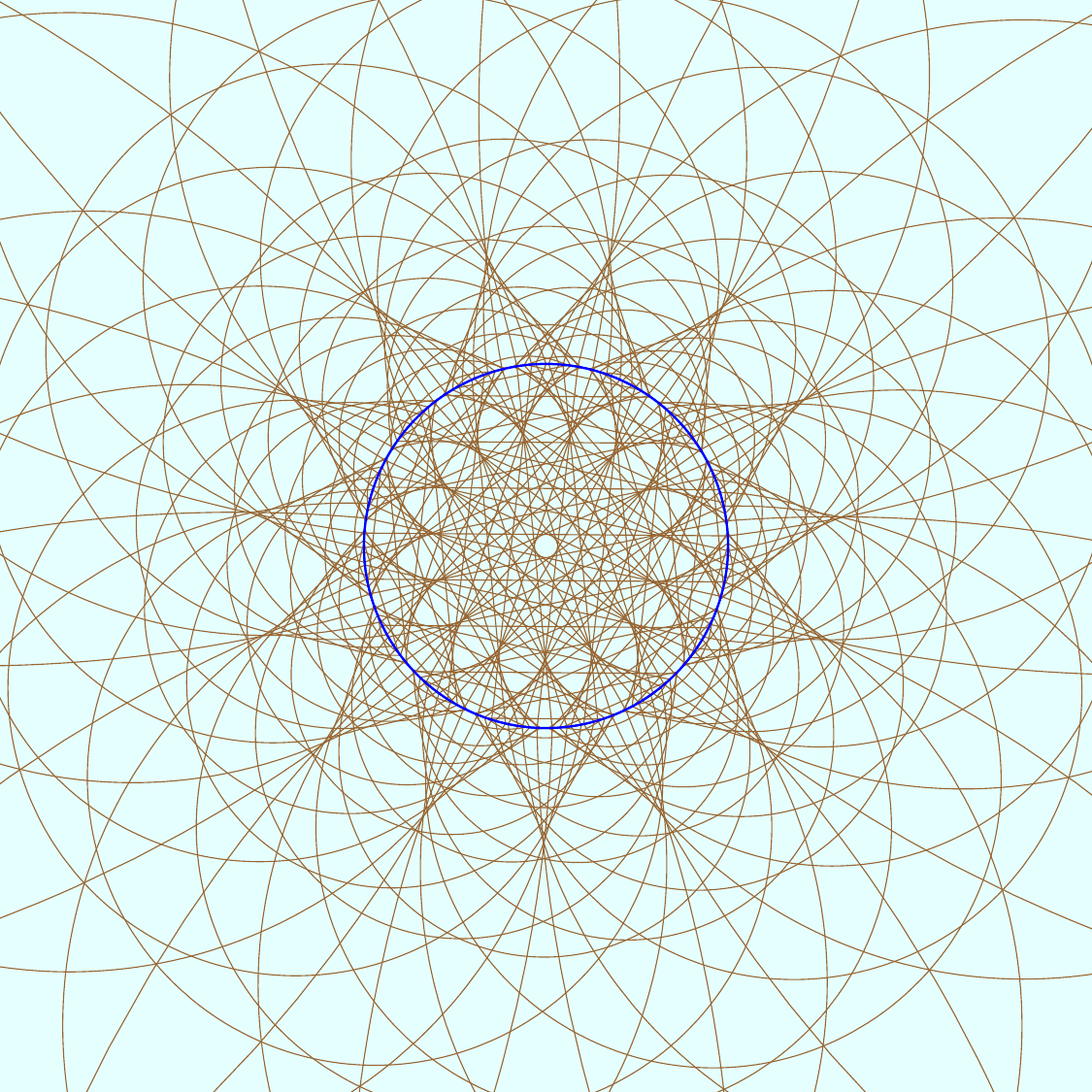}
\caption{\tiny{The Riemannian–Riemannian hybrid: \\Inside the blue circle are the brown Euclidean chords, which are the usual Beltrami geodesics in the disk. Outside the blue circle are their inversions -- these are portions of Euclidean circles passing through the center of the picture. The blue circle is the boundary of the two Riemannian regions, where the metric $\der s^2_{RR}$ is singular.
}}\label{mw1}
\end{figure}

\section{Generalization to any dimension $n\geq 2$}\label{sec7}
In this Section we show that majority of constructions discussed so far can be performed in any dimension $n\geq 2$. In particular, the Beltrami -- de Sitter model, and the Radon like transform defined previously in dimension $n=2$, can be generalized to dimension $n=4$, i.e. to the dimension of the standard physics and cosmology.

Let us see that the same principle, which brought us to the 2-dimensional de Sitter metric in Section \ref{Belde},  will give a conformal de Sitter metric in any dimensions.

For this we consider $\bbR^n$ and an $(n-1)$-dimensional unit sphere $$\bbS^{(n-1)}=\{\bbR^n\ni p=(x^1,x^2,\dots, x^n)~|~(x^1)^2+(x^2)^2+\dots+(x^n)^2=1\}$$ centered at the origin $0=(0,0,\dots,0)$.  We also define the ball $$\overline{B}{}^n=\{\bbR^n\ni p=(x^1,x^2,\dots, x^n)~|~(x^1)^2+(x^2)^2+\dots+(x^n)^2\leq 1\},$$ its interior  $$B^n=\{\bbR^n\ni p=(x^1,x^2,\dots, x^n)~|~(x^1)^2+(x^2)^2+\dots+(x^n)^2< 1\},$$ and its complement $$M^n=\bbR^n\setminus \overline{B}{}^n.$$ The conformal Lorentzian structure in $M^n$ will be defined, as in the 2D case, requiring that the light cones of the Lorentzian metric $\der s^2$ in $M^n$, be the usual cones with tips at points $p\in M^n$ and with generators being straight lines emanating from $p$ and being tangent to the sphere $\bbS^{(n-1)}=\overline{B}{}^n\setminus B^n$.

Thus the $n$-dimensional situation of getting the conformal structure on $M^n$ can refer again to the Figure \ref{fi3}. Looking at this figure, we have $p=(x^1,x^2,\dots ,x ^n)$ and $q=p+dp=(x^1+\der x^1,x^2+\der x^2,\dots x^n+\der x^n)$. The usual dot product in $\bbR^n$ gives:
\begin{itemize}
\item[(i)] $r^2+(r-p)^2=p^2$, by Pythagoras theorem,
\item[(ii)] $r dp=0$, because the vectors $r$ and $ q-p=\der p$ are orthogonal,
\item[(iii)] $ r-p=\lambda \der p$, with some scalar $\lambda$,  because the displacement $q-p=\der p$ along the light cone generator must be aligned with the generator vector $r-p$,
  \item[(iv)] $r^2=1$, since the vector $r$ points to the surface of the unit sphere $\bbS^{(n-1)}$.
\end{itemize}
By (i) and (iv) we get $p^2-1=(r-p)^2$, an thus $p^2-1=(r-p)\lambda \der p$ by (iii). This determines $\lambda=\frac{p^2-1}{(r-p)\der p}$, and by (iii) yields $r-p=\frac{p^2-1}{(r-p)\der p}\der p$, or multiplying both sides by the denominator: $[(r-p)\der p](r-p)=(p^2-1)\der p$. Taking the dot product of both sides of this last relation with the displacement $\der p$ we get $[(r-p)\der p]^2=(p^2-1)(\der p)^2$, or equivalently
\be \der s^2_0=(p^2-1)(\der p)^2-(p\der p)^2=0\label{g0met}\ee
because of (ii). This last expression, when written in coordinates $(x^1,x^2,\dots, x^n)$ of the point $p$, reads
\be \der s^2_0=\Big(-1+\sum_{i=1}^n(x^i)^2\Big)\Big(\sum_{i=1}^n(\der x^i)^2\Big)-\Big(\sum_{i=1}^nx^i\der x^i\Big)^2.\label{conmetn}\ee
Thus we have shown that the vanishing of $\der s^2_0$ is the condition for the displacement $\der p$ of a point $p$ to be along the generator of the light cone in $M^n$. Thus on $M^n$ the conformal class $[\der s^2_0]$ of metrics $\der s^2=\mathrm{e}^{2\phi}\der s^2_0$, with $\der s^2_0$ given by \eqref{conmetn}, is the Lorentzian class whose light cones are cones containing the ball $\overline{B}{}^n$ as their tangent.
\begin{remark}
  Note that in case of $n=2$ and $(x^1,x^2)=(x,y)$, the metric \eqref{conmetn} becomes precisely the same $\der s^2_0$ as in \eqref{conmet2}. Note also that one can write \eqref{conmetn} in the following equivalent form:
  $$\der s^2_0=\sum_{i=1}^n\Big(\big(-1+\sum_{j\neq i}^n(x^j)^2\big)\big(\der x^i\big)^2\Big)-2\sum_{i< j}^n\Big(x^ix^j\der x^i\der x^j\Big).$$
  This generalizes the 2D formula \eqref{conmet2} in a more straightforward way.
\end{remark}
\begin{remark}
  As in the case of two dimensions, the metric $\der s^2_0$ is conformal to the de Sitter metric. To see this we consider
  $$\der s^2=\frac{-1}{(p^2-1)^2}\der s^2_0,$$
  where $\der s^2_0$ is as in \eqref{g0met}, and use the polar parameterization of the point $p\in M^n$, i.e. we take $p=\rho N$, with $\rho\geq 0$ and $N$ the unit vector, $N^2=1$, on the sphere $\bbS^{(n-1)}$. Then, using the relation $N\der N=0$, one easily gets that
  $$\der s^2=\frac{\der \rho^2}{(\rho^2-1)^2}-\frac{\rho^2}{\rho^2-1}(\der N)^2.$$
  Now the change of the $\rho$ coordinate $\rho=\coth(t)$ brings the metric $g$ to the standard de Sitter form \cite{HawEl}
  $$\der s^2=\der t^2-\cosh^2t\,\,\der \Omega^2,$$
  where we used the fact that $(\der N)^2$ is the standard metric on the unit sphere $\bbS^{(n-1)}$, i.e. $(\der N)^2=\der\Omega^2$. Please note that the range of the \emph{time} coordinate $t$ is from $t\to 0^+$, corresponding to points $p\in M^n$ being very far from the ball $\overline{B}{}^n$, to $t\to +\infty$, when the points $p$ in $M^n$ approach the ball. 
\end{remark}
\subsection{The hybrid of the de Sitter and Beltrami in $\bbR^n$}
Now the metric
\be \der s^2=\frac{\sum_{i=1}^n\Big(\big(1-\sum_{j\neq i}^n(x^j)^2\big)\big(\der x^i\big)^2\Big)+2\sum_{i< j}^n\Big(x^ix^j\der x^i\der x^j\Big)}{\Big(1-\sum_{i=1}^n (x^i)^2\Big)^2}\label{ndimdan}\ee
is a regular \emph{Einstein} metric in the entire set $\bbR^n\setminus \bbS^{(n-1)}=B^n\cup M^n$. However, its signature is Lorentzian only in $M^n$; in $B^n$ the signature is Riemannian. Due to the blow up of the conformal factor $\big(1-\sum_{i=1}^n (x^i)^2\big)^{-2}$ at the hypersurface $p^2=1$, the metric blows up at the sphere $\bbS^{(n-1)}$. Note that the change of the signature of the conformal class $[\der s^2]$ of the metric $\der s^2$ occurs also at this sphere.

We call $\bbR^n$ equipped with the conformal class $[\der s^2]$ of the hybrid-signature Einstein metric $\der s^2$ the \emph{Beltrami -- de Sitter model}.
\subsection{The $n$-dimensional hyperbolic spaces and Beltrami -- de Sitter model}
We relate the Beltrami -- de Sitter model to the geometries on hyperboloids 
$$-T^2+\sum_{i=1}^n (X^i)^2=-\epsilon, \quad\quad \epsilon=\pm1,$$
in Minkowski space
$$\bbM=\{(T,X^i)\,|\,T,X^i\in\bbR, i=1,2,\dots n\},$$ with Minkowski metric $$\eta=-\der T^2+\sum_{i=1}^n(\der X^i)^2.$$
As in the case of $n=2$ there are two kinds of these hyperboloids:
the one-sheeted $\bbH_{\shortminus1}=\{(T,X^i)\,|\,T^2-\sum_{j=1}^n (X^j)^2=-1\}$, and
  the two-sheeted $\bbH_{1}=\{(T,X^i)\,|\,T^2-\sum_{j=1}^n (X^j)^2=1\}$.
From now on, we will only consider them for $T>0$.

After restriction to these hyperboloids, the Minkowski metric $\eta$ becomes
\be \eta_{|\bbH_\epsilon}=-\frac{(\sum_{i=1}^nX^i\der X^i)^2}{\epsilon + \sum_{j=1}^n(X^j)^2}+\sum_{i=1}^n(\der X^i)^2.\label{etih}\ee
This metric has the Riemannian signature on the hyperboloid $\bbH_1$, and it has Lorentzian signature on the one-sheeted hyperboloid $\bbH_{\shortminus1}$. It further follows that the pair $(\bbH_{\shortminus1},\eta_{|\bbH_{\shortminus1}})$ is locally isometric to the de Sitter space $(M^n,\der s^2)$, and that the entire one sheet of $\bbH_1$, with its metric $\eta_{\bbH_1}$, is locally isometric to the Beltrami space $(B^n,\der s^2)$.

Similarly as in $n=2$ case, we now use the \emph{central projection}
$pr:\bbH_\epsilon\to \Pi$ that maps points $(\sqrt{\epsilon+\sum_{j=1}^n(X^j)^2},X^i)$ of the hyperboloids $\bbH_\epsilon$ to points $(1,x^i)$ in $\bbM$, which lie on the plane $\Pi=\{(T,X^i)\,|\,T=1, X^i=x^i, i=1,2,\dots, n\}$ tangent to the hyperboloid $\bbH_1$ at its tip $(T,X^i)=(1,0)$. The formal definition of this projection is:\\
 \emph{Given a point $(T,X^i)$ on $\bbH_\epsilon$, use a line $\ell$ in the Minkowski space $\bbM$ passing through this point and the origin $(T,X^i)=(0,0)$, and  define a projected point $(1,x^i)$ on the plane $T=1$ as the intersection of the line $\ell$ and this plane.}
 
The explicit formula is:
$$pr
  \big(
  \sqrt{\epsilon+\sum_{j=1}^n(X^j)^2\,\,},X^i
  \big)
  =
  \big(
  1,\frac{X^i}{\sqrt{\epsilon+\sum_{j=1}^n(X^j)^2}}\big),$$
which maps $n$ numbers $(X^i)$, $i=1,2,\dots,n$, parameterizing points of the hyperboloid $H_\epsilon$  to the $(x^i)$'s paramtrizing points of the plane $\Pi$ via $pr  \big(X^i\big) =\big(x^i\big)$
with 
$$x^i=\frac{X^i}{\sqrt{\epsilon+\sum_{j=1}^n(X^j)^2}}.$$
See the geometry of this transformation on Figure \ref{fig455}.\\
\begin{figure}[h!]
\centering
\includegraphics[scale=0.22]{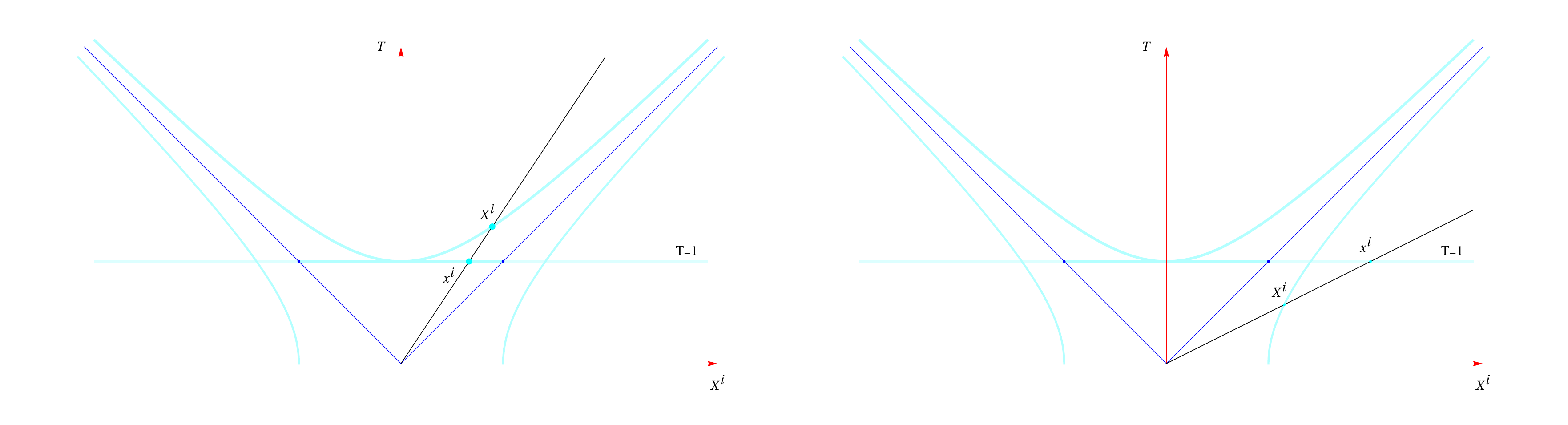}
\caption{\tiny{Two types of hyperboloids in the $n$-dimensional Minkowski space with and the corresponding projections.}}\label{fig455}
\end{figure}
Now, passing from coordinates $(X^i)$ to the projected-to-the-plane-$\Pi$ coordinates $(x^i)$, and inserting them in the hyperbolic metric \eqref{etih}, we obtain:
$$\eta_{\bbH_\epsilon}=\epsilon\frac{(1-\sum_{i=1}^n(x^i)^2)(\sum_{i=1}^n(\der x^i)^2)+(\sum_{i=1}^nx^i\der x^i)^2}{(1-\sum_{i=1}^n(x^i)^2)^2},$$
which after a short algebra is the metric $\der s^2$ of the Beltrami -- de Sitter model.

And we have a generalization of Corollary \ref{colo}.
\begin{corollary}
  The plane $\Pi=\{\bbM\ni(T,X^i)\,|\,T=1\}$ in the Minkowski space $(\bbM,\eta)$ is a realization of the Beltrami -- de Sitter model $BdS$.
  
  In this realization the Beltrami space $B^n$ consists of points $(1,x^i)\in \bbM$ such that $\sum_{i=1}^n(x^i)^2<1$, and is equipped with the \emph{Riemannian} Beltrami metric $\eta_{|\bbH_1}$. It can be identified with the space of all \emph{inertial observers} passing through the origin of the Minkowski space $(\bbM,\eta)$.

  The de Sitter space $M^n$ consists of points $(1,x^i)\in \bbM$ such that $\sum_{i=1}^n(x^i)^2>1$, and is equipped with the \emph{Lorentzian} de Sitter metric $-\eta_{|\bbH_{\shortminus1}}$. It can be identified with  the space of all \emph{tachyons} passing through the origin of the Minkowski space $(\bbM,\eta)$. 
  \end{corollary}

\subsection{The Radon-like transform} \label{rlt_n} We now define the $point \leftrightarrow hyperplane$ duality in the $n$-dimensional Beltrami -- de Sitter model, which is an analog of the $point\leftrightarrow line$ duality known from the 2-dimensional model discussed in Section \ref{proj_du}. Here it is:

A point $p\in M^n$ defines a cone ${\mathcal C}_p\subset M^n$, whose tip is $p$ and whose generators are straight lines in $M^n$ passing through $p$ and tangent to $\partial B^n$. This cone, in turn, defines an $(n-1)$--dimensional disk $c_p$ in $B^n$, whose boundary is an $(n-2)$--dimensional sphere  being the intersection  ${\mathcal C}_p\cap\partial B^n=\bbS^{(n-2)}$ of the cone ${\mathcal C}_p$ with the boundary $\partial B^n$ of $\overline{B}^n$. Explicitly, in $\bbR^n$ with coordinates $(X^i)$, $i=1,\dots ,n$, and with $B^n=\{\bbR^n\ni (X^i)\,|\,\sum_{i=1}^n(X^i)^2<1\}$, $M^n=\{\bbR^n\ni (X^i)\,|\,\sum_{i=1}^n(X^i)^2>1\}$ and $\overline{B}^n=\bbR^n\setminus M^n$, 
if the point $p=(x^i)$, $i=1,\dots,n$, is in $M^n$, i.e. when $\sum_{i=1}^n(x^i)^2>1$, then the disk $c_p$ is
$$c_p\,=\,\Big\{\,\bbR^n\ni(X^i)\,|\,\sum_{i=1}^n x^i X^i=1\,\Big\}\,\cap\, \overline{B}^n.$$

Likewise, given a point $P\in B^n$, one has a wealth of $(n-1)$--dimensional disks $c_P$ in $B^n$ containing $P$, obtained by intersecting $B^n$ with hyperplanes $H_P$ in $\bbR^n$ passing through $P$. Each of these disks $c_p$ has its cone in $M^n$, which is tangent to $\partial B^n$ at points of the $(n-2)$--dimensional sphere $c_P\cap\partial B^n$. The \emph{tips} of all these cones define an $(n-1)$ dimensional hyperplane $\ell_P$ in $M^n$. The explicit formula for the hyperplane $\ell_P$ in $M^n$ corresponding to the point $P$ in $B^n$, i.e. such that $P=(a^i)$, $i=1,\dots,n$, with $\sum_{i=1}^n (a^i)^2<1$, is given by
$$\ell_P\,=\,\Big\{\,\bbR^n\ni(X^i)\,|\,\sum_{i=1}^n a^i X^i=1\,\Big\}.$$

This enables for the $n$-dimensional version of the Radon transform between functions defined on $B^n$ and on $M^n$. It is defined as follows: 

 {\small\bf Transforming functions from $M^n$ to $B^n$}:
Having an integrable function $f:M^n\to\bbR$ defined in the de Sitter part $M^n$ of the $n$-dimensional Beltrami--de Sitter model we \emph{define its transformed function} $\hat{f}: B^n\to \bbR$ by 
 $$\hat{f}(P)=\int_{\ell_P}f.$$
 Here $P$ is a point in $B^n$, and the integration is the multi-dimensional integral over the hyperplane $\ell_P$ in $M^n$ defined in $M^n$ by $P$.
 
 {\small\bf Transforming functions from $B^n$ to $M^n$}: Likewise, given an integrable function $F: B^n\to\bbR$ on the Beltrami part $B^n$ of the $n$-dimensional Beltrami--de Sitter model we get the value $\hat{F}(p)$ of its \emph{transform} $\hat{F}$ at the point $p$ in $M^2$ as the multidimensional integral
 $$\hat{F}(p)=\int_{c_p}F$$  of the function $F$ along the disk $c_p$ in $B^n$, which in $B^n$ corresponds to $p$ from $M^n$.

 \subsection{The correspondence space and the dancing metric} We can define the \emph{corresponding space} $T=B^n\times M^n$ of the $n$-dimensional Beltrami -- de Sitter model and its \emph{dancing metric conformal class}  geometrically, as we did in Section \ref{sec210}. This is done by observing that a differentiable curve $p(t)$ in $M^n$ defines a $(n-2)$-dimensional plane $P_{p(t)}$ in $B^n$ and that a differentiable curve $P(t)$ in $B^n$ defines a part of an $(n-2)$-dimensional plane $p_{P(t)}$ in $M^n$. Then the $n$-dimensional Beltrami -- de Sitter model \emph{dancing conditions} for the \emph{movements} of points $p(t)$ in $M^n$ and $P(t)$ in $B^n$ are, that given a pair $(p(t),P(t))$ the point $p(t)$ in $M^n$ should always move in the direction of $p_{P(t)}$ in $M^n$, and the point $P(t)$ in $B^n$ should always move in the direction of $P_{p(t)}$ in $B^n$. This defines a conformal structure $[g_0]$ in $T$ in which \emph{any curve} $\gamma(t)=(p(t),P(t))\in T$ \emph{performing the dancing movement is a null curve}. However we will not show further details of this derivation here. Instead, we will use a simpler method of defining the dancing metric, immediately on the extension $(\bbR P^n)^*\times\bbR P^n$ of $T$, generalizing the approach presented in Section \ref{furext}.

 Thus, in analogy with this section, we start with $V=(\bbR^{n+1})^*\times\bbR^{n+1}$ and its points $({\bf p},{\bf q})=(p_i,q^i)$, $i=1,\dots,n$, and use the formula \eqref{dancc} define the bilinear form
 $$G=2\,\frac{<{\bf p},\der{\bf q}><\der {\bf p},{\bf q}>-<{\bf p},{\bf q}><\der {\bf p},\der{\bf q}>}{<{\bf p},{\bf q}>^2}$$
 on  $(\bbR^{n+1})^*\times\bbR^{n+1}$. Similarly as in the case of $n=2$, this degenerate form on $(\bbR^{n+1})^*\times\bbR^{n+1}$ descends to the split signature metric $g$ on the $2n$-dimensional manifold ${\mathcal T}\simeq (\bbR P^n)^*\times\bbR P^n$ obtained from $V$ by identifying points on the directions in $(\bbR^n)^*$ and $\bbR^n$. Explicitly, on the open set in $\mathcal T$, with coordinates $p_i=a_i$, $i=1,\dots,n$, $p_{n+1}=-1$, $q^i=x^i,$ $i=1,\dots, n$, $q^{n+1}=1$, we get
 \be g=\frac{\sum_{j=1}^n\,2\,\der a_j\,\,\Big(\,\big(1-\sum_{i=1}^na_i x^i\big)\,\der x^j\,+\,x^j\,\big(\sum_{i=1}^n a_i\der x^i\big)\,\Big)}{(1-\sum_{i=1}^na_ix^i)^2}.\label{dann}\ee
 This, by definition, is the \emph{dancing metric} on $\overline{T}=(\bbR^n)^*\times\bbR^n$, with coordinates $(a_i,x^i)$, which is an extension of the correspondence space $T=B^n\times M^n$, of the $n$-dimensional Beltrami -- de Sitter model. Here $B^n=\{\,\bbR^n\ni(a_i)\,|\,\sum_{i=1}^n a_i^2<1\,\}$ and $M^n=\{\,\bbR^n\ni(x^i)\,|\,\sum_{i=1}^n (x^i)^2>1\,\}$

 In particular, one can easily see that if $n=2$ and $a_1=a$, $a_2=b$, $x^1=x$ and $x^2=y$, this formula coincides with the original dancing metric given at \eqref{dancemet}. Also, on the $n$-dimensional diagonal set in $\mathcal T$, where $a_i=x^i$, $i=1,\dots, n$, we see that this metric becomes \eqref{ndimdan}, i.e. the Einstein metric of the $n$-dimensional  Beltrami -- de Sitter model.

 \begin{theorem}\label{dancingn}
  The space $T=B^n\times M^n$ is naturally equipped with a conformal class $[G]$ of a split signature metric $g$ as in \eqref{dann}, 
  the dancing metric, which is Einstein, $$Ric(g)=-(n+1)g.$$ 

  The conformal class $[g]$ has $\sla(n+1,\bbR)$ as its full algebra of conformal Killing symmetries. These symmetries are all Killing symmetries of the metric $g$.
  \end{theorem}
 \section{The twistor distribution for general $n$}
 The local ${\bf G}_2$ symmetry of the twistor distribution of the dancing metric of the 2-dimensional Beltrami -- de Sitter model is exceptional and at first glance seems to be related to the case $n=2$ only. Although the symmetry group of the twistor distribution of the dancing metric \eqref{dann} for $n\geq 3$ is different from ${\bf G}_2$ and  depends on $n$, we found some surprising behavior for particular $n$s. In this section we will show details of this issue for $n=3$ and leave the physics' interesting case $n=4$ for the forthcoming paper with Iand Anderson \cite{ian}.

 \subsection{The twistor bundle of totally null planes} As a preparation for this discussion we write the dancing metric in the null coframe $(\theta^I)=(\theta^1,\dots,\theta^n,\theta^{n+1},\dots,\theta^{2n})$, in which the metric assumes the form
 $$g=2\theta^1\theta^{2n}+2\theta^2\theta^{2n-1}+\dots+2\theta^n\theta^{n+1}=g_{IJ}\theta^I\theta^J,$$
 with $(g_{IJ})$, $I,J=1,2,\dots,2n$, the \emph{anti}diagonal matrix having all nonzero elements equal to one, i.e.
 $$(g_{IJ})=\bma 0&S\\S&0\ema,$$
 with $S$ the $n\times n$-matrix given by: 
\be S=\bma 0&0&\dots&0&1\\0&0&\dots&1&0\\\dots&\dots&\dots&\dots&\dots\\\dots&\dots&\dots&\dots&\dots\\0&1&\dots&0&0\\1&0&\dots&0&0\ema.\label{ssss}\ee
 
 Explicitly, the coframe is given by:
 \be
 \begin{aligned} 
 \theta^i&=\frac{\der a_i}{1-\sum_{j=1}^na_jx^j}\\
   \theta^{2n+1-i}&=\frac{(1-\sum_{j=1}^na_jx^j)\der x^i+x^i(\sum_{j=1}^na_j\der x^j)}{1-\sum_{j=1}^na_jx^j},
   \end{aligned}\label{nuc}\ee
 where $i=1,2,\dots n$. We denote be $(e_i)$ the dual frame to $(\theta^I)$, $e_I\hook \theta^I=\delta_I^J$, $I,J=1,2,\dots,2n$.
 Now, in the correspondence space $T$, we consider a rank $n$ distribution $D_0$ spanned by the frame vectors $e_k$ with $k=n+1,n+2,\dots,2n$,
 $$D_0=\Span(e_{n+1},e_{n+2},\dots,e_{2}).$$
 This distribution is \emph{totally null}, i.e. $g(X,Y)=0$ for all $X$,$Y$ in $D_0$, and has maximal possible rank among all totally null distributions in $T$. It defines an $n$-dimensional totally null plane $D_{0\xi}$ at every point $\xi\in T$. The connected component of identity $\sog_0(n,n)$ of the orthogonal group $\sog(n,n)$ acts on such planes, and in particular, acts on the plane $D_{0\xi}$ producing an orbit
 $$\mathcal{O}_{D_{0\xi}}=\big\{\,\Span(a\cdot e_{n+1},a\cdot e_{n+2},\dots,a\cdot e_{2n})\,\,\mathrm{s.t.}\,\,a\in\sog_0(n,n)\,\big\}.$$
 For every two points $\xi$ and $\xi'$ in $T$, these orbits are diffeomorphic to $\sog(n)$. In particular, they have the same dimension $N=\frac{n(n-1)}{2}$, and locally can be parametrized by elements of the Lie algebra $\soa(n)$. One easily sees the following proposition:
 \begin{proposition}\label{proa}
   Let $A=(A^i{}_j)$, $i,j=1,2,\dots n$, be an $n\times n$ real matrix from the matrix Lie algebra $\soa(n)$, i.e. $A$ is a matrix such that $A^T=-A$. Consider a point $\xi\in T$, the frame vectors $(e_{n+1}, e_{n+2},\dots,e_{2n})$ at this point. Then the space 
   $$D_{A\xi}=\Span\big(\,Y_1,\,Y_2,\,\dots,\,Y_n\,\big)$$
   where $$Y_i=e_{n+i}+\sum_{j=1}^ne_j(SA)^j{}_i,\quad\mathrm{with}\quad i=1,2,\dots,n,\,\, \mathrm{and}\,\, S\,\,\mathrm{as\,\, in}\,\,\eqref{ssss},$$
   is a totally null $n$-plane at $\xi$ belonging to the same orbit $\mathcal{O}_{D_{0\xi}}$  as $D_{0\xi}$. Moreover, locally, every element of this orbit is in this form.
 \end{proposition}
 \begin{remark}\label{remaa}
 In particular, the oribt of $D_{0\xi}$ can be parametrized according to the dimension $n$ of the Beltrami -- de Sitter space, and 
 \begin{itemize}
    \item if $n=2$ depends on \emph{one} real parameter $t_3$, with
   $A=\bma 0&-t_3\\t_3&0\ema$; in this case for each $t_3\in\bbR$, the corresponding $2$-plane is spanned by
   $$\begin{aligned}
     Y_1&=e_3+t_3 e_1,\\Y_2&=e_4-t_3 e_2;\end{aligned}$$
     \item  if $n=3$ the orbit depends on \emph{three} real parameters $(t_1,t_2,t_3)$, with $A=\bma 0&-t_1&t_2\\t_1&0&-t_3\\-t_2&t_3&0\ema$; here for each $t_1,t_2,t_3)\in\bbR^3$, the corresponding 3-plane is spanned by
  $$\begin{aligned}
  Y_1&=e_4-t_2 e_1+t_1 e_2,\\
  Y_2&=e_5+t_3 e_1-t_1e_3,\\
  Y_3&=e_6-t_3 e_2+t_2 e_3;\end{aligned}$$
     \item  if $n=4$ the orbit depends on \emph{six} real parameters $(t_1,t_2,t_3,t_4,t_5,t_6)$, with $A=\bma 0&-t_4&t_5&-t_6\\t_4&0&-t_1&t_2\\-t_5&t_1&0&-t_3\\t_6&-t_2&t_3&0\ema$; and here, for each $(t_1,t_2,t_3,t_4,t_5,t_6)\in\bbR^6$, its corresponding 4-plane is spanned by
$$\begin{aligned}
    Y_1&=e_5+t_6 e_1-t_5 e_2+t_4 e_3,\\Y_2&=e_6-t_2 e_1+t_1 e_2-t_4 e_4,\\
    Y_3&=e_7+t_3 e_1-t_1e_3+t_5 e_4,\\
    Y_4&=e_8-t_3 e_2+t_2 e_3-t_6 e_4;
\end{aligned}$$
    \item And so on, if $n>4$.
   \end{itemize}
\end{remark}
 We define the \emph{twistor bundle} of the correspondence space $T$ of the $n$-dimensional Beltrami -- de Sitter model, to be a bundle $\sog(n)\to \bbT\stackrel{\pi}{\to} T$ whose fiber over every point $\xi\in T$ consists of all totally null planes at $\xi$ belonging to the orbit $\mathcal{O}_{D_{0\xi}}$.
 
 \subsection{The horizontal lift} The \emph{Levi--Civita connection} of the dancing metric $g$ on $T$ defines a \emph{horizontal space} $H$ in the twistor bundle $\mathbb{T}$. Specifically, if we parametrize points in the fiber $\pi^{-1}(\xi)$ of $\mathbb{T}$ over a point $\xi \in T$ by matrices $A \in \mathrm{SO}(n)$, as in Proposition 8.1, and if we wish to horizontally lift a tangent vector $X_\xi$ from $T_\xi T$ to a point $(\xi, A) \in \mathbb{T}$, we proceed as follows: First, we take a curve $\gamma(t)$ in $T$ that is tangent to $X_\xi$ at $t = 0$, i.e., $\gamma(0) = \xi$ and $\frac{d\gamma}{dt}\big|_{t=0} = X_\xi$. Then, the horizontal lift of $X_\xi$ to a point $(\xi, A)$ in $\mathbb{T}$ is a vector $\tilde{X}_{(\xi,A)}$ in $T_{(\xi,A)} \mathbb{T}$, which is tangent to a curve $\tilde{\gamma}(t)$, with $\tilde{\gamma}(0) = (\xi, A)$, whose points in $\mathbb{T}$ correspond to a 1-parameter family of totally null planes in $T$, which are the parallel transports of the totally null plane $D_{A\xi}$ from the point $\xi$ along $\gamma(t)$ in $T$, according to the Levi--Civita connection. This procedure defines a horizontal lift of any tangent vector from $T$, thereby defining the \emph{horizontal space} $H_{(\xi,A)} \subset T_{(\xi,A)} \mathbb{T}$ at each point $(\xi, A)$ in $\mathbb{T}$.

 Below we give formulae for the horizontal lifts of the frame vectors $(e_I)$ of the dancing metric null coframe $(\theta^I)$, as defined in \eqref{nuc}, in low dimensions $n=2$, $n=3$ and $n=4$.

 \begin{proposition}
   Let $(a_i,x^j,t_\al)$ be coordinates of the points of the twistor bundle $\bbT$, where the $(a_i,x^j)$ are coordinates on the correspondence space $T$ appearing in the dancing metric, as in \eqref{nuc}, and $(t_\alpha)$, $\alpha=1,2,\dots,\tfrac12 n(n-1)$, be the antisymmetric matrix $A$ coordinates of the fiber, which in the case of $n=2,3,4$ are defined in Remark \ref{remaa}. Then the horizontal lift $\tilde{e}_I$ of the frame vectors $e_I$ from a point $\xi=(a_i,x^j)$ in $T$ to the point $p=(a_i,x^j,t_\al)$ in $\bbT$ are,
   \begin{itemize}
   \item when $n=2$:
     $$ \begin{aligned}
     \tilde{e}_1&=e_1-x_1t_3\partial_{t_3},\\
     \tilde{e}_2&=e_2-x_2t_3\partial_{t_3},\\
     \tilde{e}_3&=e_3+2a_2t_3\partial_{t_3},\\
     \tilde{e}_4&=e_4+2a_1t_3\partial_{t_3},
   \end{aligned}$$
  \item when $n=3$:
     $$ \begin{aligned}
    \tilde{e}_1&=e_1+(x_2t_1-x_1t_2)\partial_{t_2}+(x_3t_1-x_1t_3)\partial_{t_3},\\
    \tilde{e}_2&=e_2+(x_1t_2-x_2t_1)\partial_{t_1}+(x_3t_2-x_2t_3)\partial_{t_3},\\
    \tilde{e}_3&=e_3+(x_1t_3-x_3t_1)\partial_{t_1}+(x_2t_3-x_3t_2)\partial_{t_2},\\
     \tilde{e}_4&=e_4+2a_3(t_1\partial_{t_1}+t_2\partial_{t_2}+t_3\partial_{t_3}),\\
     \tilde{e}_5&=e_5+2a_2(t_1\partial_{t_1}+t_2\partial_{t_2}+t_3\partial_{t_3}),\\
     \tilde{e}_6&=e_6+2a_1(t_1\partial_{t_1}+t_2\partial_{t_2}+t_3\partial_{t_3}),
   \end{aligned}$$
  \item when $n=4$:
     $$ \begin{aligned}
    \tilde{e}_1&=e_1+(x_2t_1-x_1t_2-x_4t_4)\partial_{t_2}+(x_3t_1-x_1 t_3-x_4t_5)\partial_{t_3}+(-x_3t_4+x_2t_5-x_1t_6)\partial_{t_6},\\
    \tilde{e}_2&=e_2+(x_1t_2-x_2t_1+x_4t_4)\partial_{t_1}+(x_3t_2-x_2 t_3-x_4t_6)\partial_{t_3}+(x_3t_4-x_2t_5+x_1t_6)\partial_{t_5},\\
    \tilde{e}_3&=e_3+(x_1t_3-x_3t_1+x_4t_5)\partial_{t_1}+(x_2t_3-x_3 t_2+x_4t_6)\partial_{t_2}+(-x_3t_4+x_2t_5-x_1t_6)\partial_{t_4},\\
    \tilde{e}_4&=e_4+(x_2t_1-x_1t_2-x_4t_4)\partial_{t_4}+(x_3t_1-x_1 t_3-x_4t_5)\partial_{t_5}+(x_3t_2-x_2t_3-x_4t_6)\partial_{t_6},\\
    \tilde{e}_5&=e_5+2a_4(t_1\partial_{t_1}+t_2\partial_{t_2}+t_3\partial_{t_3}+t_4\partial_{t_4}+t_5\partial_{t_5}+t_6\partial_{t_6}),\\
    \tilde{e}_6&=e_6+2a_3(t_1\partial_{t_1}+t_2\partial_{t_2}+t_3\partial_{t_3}+t_4\partial_{t_4}+t_5\partial_{t_5}+t_6\partial_{t_6}),\\
    \tilde{e}_7&=e_7+2a_2(t_1\partial_{t_1}+t_2\partial_{t_2}+t_3\partial_{t_3}+t_4\partial_{t_4}+t_5\partial_{t_5}+t_6\partial_{t_6}),\\
     \tilde{e}_8&=e_8+2a_1(t_1\partial_{t_1}+t_2\partial_{t_2}+t_3\partial_{t_3}+t_4\partial_{t_4}+t_5\partial_{t_5}+t_6\partial_{t_6}).
   \end{aligned}$$
     \end{itemize}
   \end{proposition}
\subsection{The twistor distribution}
The (tautological) horizontal lift of of a totally null plane $D_{A\xi}$ from the point $\xi$ in $T$ to its point $(\xi,A)$ in the fiber $\pi^{-1}(\xi)$ in $\bbT$ defines an $n$-dimensional plane ${\mathcal D}_{(\xi,A)}$ at $(\xi,A)$. This, when done at every point $\xi\in T$ with every $n$-plane in the orbit ${\mathcal O}_{D_0\xi}$, defines a rank $n$-distribution ${\mathcal D}$ in $\bbT$. This is the \emph{twistor distribution} on $\bbT$. This distribution, in the case $n=2$, coincides with the one which was defined in Section \ref{tdis}.
\begin{remark}\label{holift}
  Although the \emph{horizontal lift} of vectors from $T$ to $\bbT$ was defined by means of the Levi -- Civita connection of the dancing \emph{metric}, i.e. by a Lorentzian \emph{metric}, the \emph{twistor distribution} $\mathcal D$, as defined here, is a \emph{conformal object}. This follows from the fact that the twistor bundle $\bbT$ itself is a conformal object (we only used the concept of \emph{nullity} of vectors to define it), and the fact that $\mathcal D$  coincides with the unique distribution ${\mathcal D}$ having the \emph{property} that at every point $(\xi,A)$ in $\bbT$ its derived distribution ${\mathcal D}^1=[{\mathcal D},{\mathcal D}]$, satisfies $$\pi_*({\mathcal D}^1{}_{(\xi,A)})=D_{A\xi}, \quad \forall (\xi,A)\in \bbT,$$where $D_{A\xi}$ is the totally null plane from the orbit ${\mathcal O}_{D_0\xi}$, which is labeled by the skew symmetric matrix $A$, as in Proposition \ref{proa}. This property of $\mathcal D$, which can be taken as its defining property, does not use the metric at all, and since the totally null planes $D_{A\xi}$ on $T$ are conformal objects, so is the distribution $\mathcal D$.  
\end{remark}

We have the following theorem.
\begin{proposition}\label{pro85}
  In coordinates $(\xi,A)$ of Remark \ref{remaa} the twistor disribution $\mathcal D$, on the twistor bundle $\sog(n)\to\bbT\stackrel{\pi}{\to} T$ is spanned,\begin{itemize}
    \item[(i)]  when $n=2$, by:
$$\begin{aligned}
      Y_1&=\partial_{x^2}-a_2(x^1\partial_{x^1}+x^2\partial_{x^2})+(1-a_1x^1-a_2 x^2)t_3\partial_{a_1}+(2a_2-t_3x^1)t_3
      \partial_{t_3},\\
     Y_2&=\partial_{x^1}-a_1(x^1\partial_{x^1}+x^2\partial_{x^2})+(1-a_1x^1-a_2 x^2)(-t_3\partial_{a_2})+(2a_1+t_3x^2)t_3\partial_{t_3},
   \end{aligned}$$
    \item[(ii)] when $n=3$, by:
      $$\begin{aligned}
    Y_1=&\partial_{x^3}-a_3(x^1\partial_{x^1}+x^2\partial_{x^2}+x^3\partial_{x^3})+(1-a_1x^1-a_2 x^2-a_3x^3)(t_1\partial_{a_2}-t_2\partial_{a_1})+\\&(2a_3+t_2x^1-t_1x^2)(t_1\partial_{t_1}+t_2\partial_{t_2}+t_3\partial_{t_3}),\\
    Y_2=&\partial_{x^2}-a_2(x^1\partial_{x^1}+x^2\partial_{x^2}+x^3\partial_{x^3})+(1-a_1x^1-a_2 x^2-a_3x^3)(t_3\partial_{a_1}-t_1\partial_{a_3})+\\&(2a_2+t_1x^3-t_3x^1)(t_1\partial_{t_1}+t_2\partial_{t_2}+t_3\partial_{t_3}),\\
    Y_3=&\partial_{x^1}-a_1(x^1\partial_{x^1}+x^2\partial_{x^2}+ x^3\partial_{x^3})+(1-a_1x^1-a_2 x^2-a_3x^3)(t_2\partial_{a_3}-t_3\partial_{a_2})+\\&(2a_1+t_3x^2-t_2x^3)(t_1\partial_{t_1}+t_2\partial_{t_2}+t_3\partial_{t_3}),\end{aligned}$$
    \item[(iii)] and, when $n=4$, by:
      $$\begin{aligned}
    Y_1=&\partial_{x^3}-a_3(x^1\partial_{x^1}+x^2\partial_{x^2}+x^3\partial_{x^3}+x^4\partial_{x^4})+(1-a_1x^1-a_2 x^2-a_3x^3-a_4x^4)(t_1\partial_{a_2}-t_2\partial_{a_1}-t_4\partial_{a_4})+\\&(2a_3+t_2x^1-t_1x^2+t_4x^4)(t_1\partial_{t_1}+t_2\partial_{t_2}+t_4\partial_{t_4})+\big(t_3(2a_3+t_2x^1-t_1x^2)+x^4(t_2t_5-t_1t_6)\big)\partial_{t_3}+\\&\big(t_5(2a_3-t_1x^2+t_4x^4)+x^1(t_3t_4+t_1t_6)\big)\partial_{t_5}+\big(t_6(2a_3+t_2x^1+t_4x^4)+x^2(t_3t_4-t_2t_5)\big)\partial_{t_6},\\
    Y_2=&\partial_{x^2}-a_2(x^1\partial_{x^1}+x^2\partial_{x^2}+x^3\partial_{x^3}+x^4\partial_{x^4})+(1-a_1x^1-a_2 x^2-a_3x^3-a_4x^4)(t_3\partial_{a_1}-t_1\partial_{a_3}+t_5\partial_{a_4})+\\&(2a_2+t_1x^3-t_3x^1-t_5x^4)(t_1\partial_{t_1}+t_3\partial_{t_3}+t_5\partial_{t_5})+\big(t_2(2a_2+t_1x^3-t_3x^1)-x^4(t_3t_4+t_1t_6)\big)\partial_{t_2}+\\&\big(t_4(2a_2+t_1x^3-t_5x^4)+x^1(t_1t_6-t_2t_5)\big)\partial_{t_4}+\big(t_6(2a_2-t_3x^1-t_5x^4)+x^3(t_2t_5-t_4t_3)\big)\partial_{t_6},\\
    Y_3=&\partial_{x^1}-a_1(x^1\partial_{x^1}+x^2\partial_{x^2}+x^3\partial_{x^3}+x^4\partial_{x^4})+(1-a_1x^1-a_2 x^2-a_3x^3-a_4x^4)(t_2\partial_{a_3}-t_3\partial_{a_2}-t_6\partial_{a_4})+\\&(2a_1+t_3x^2-t_2x^3+t_6x^4)(t_2\partial_{t_2}+t_3\partial_{t_3}+t_6\partial_{t_6})+\big(t_1(2a_1+t_3x^2-t_2x^3)+x^4(t_2t_5-t_3t_4)\big)\partial_{t_1}+\\&\big(t_4(2a_1-t_2x^3+t_6x^4)+x^2(t_2t_5-t_1t_6)\big)\partial_{t_4}+\big(t_5(2a_1+t_3x^2+t_6x^4)-x^3(t_3t_4+t_1t_6)\big)\partial_{t_5},\\
     Y_4=&\partial_{x^4}-a_4(x^1\partial_{x^1}+x^2\partial_{x^2}+x^3\partial_{x^3}+x^4\partial_{x^4})+(1-a_1x^1-a_2 x^2-a_3x^3-a_4x^4)(t_6\partial_{a_1}-t_5\partial_{a_2}+t_4\partial_{a_3})+\\&(2a_4-t_6x^1+t_5x^2-t_4x^3)(t_4\partial_{t_4}+t_5\partial_{t_5}+t_6\partial_{t_6})+\big(t_1(2a_4+t_5x^2-t_4x^3)+x^1(t_3t_4-t_2t_5)\big)\partial_{t_1}+\\&\big(t_2(2a_4-t_6x^1-t_4x^3)+x^2(t_3t_4+t_1t_6)\big)\partial_{t_2}+\big(t_3(2a_4-t_6x^1+t_5x^2)+x^3(t_1t_6-t_2t_5)\big)\partial_{t_3},\\
    \end{aligned}$$
      \end{itemize}
  \end{proposition}
As mentioned in the Remark \ref{holift} in each case $n\geq 2$ the twistor distribution $\mathcal D$ is a \emph{canonical object} on the twistor bundle that is  \emph{canonically defined} by means of the \emph{conformal structure} of the dancing metric $g$ on $T$. Thus all the \emph{conformal symmetries} of the dancinhg metric, should have their \emph{lifts} to the \emph{symmetries of the twsitor distribution} $\mathcal D$ \emph{on the twistor bundle} $\sog(n)\to\bbT\to T$. Also, by the \emph{very construction of the dnacing} which uses only the notions of $k$-\emph{planes} and their \emph{incidence} in $\bbR^n$, it is obvious that the dancing metric has its Lie algebra of conformal symmetries at least as large as the Lie algebra $\sla(n,\bbR)$. Actually, the dancing metric has the Lie algebra of conformal Killing vectors \emph{precisely} isomorphic to $\sla(n+1,\bbR)$.

We know (see Remark \ref{symg2twi}), that in the case of $n=2$, the Lie algebra of symmetries of the twistor distribution on $\bbT$ is larger than $\sla(3,\bbR)$, namely it is isomorphic to the split real form of the simple excdeptional Lie algebra $\mathfrak{g}_2$, which contains $\sla(3,\bbR)$.

Although if $n\geq 4$ the situation of the full symmetry algebra for the twistor distribution $\mathcal D$ is not fully understood, the following proposition holds.

\begin{proposition}\label{pro86}
  For every $n\geq 2$, the Lie algebra of infinitesimal symmetries of the twistor distribution $\mathcal D$ on $\bbT$, i.e. the Lie algebra  $\mathfrak{g}_{\mathcal D}$  of vector fields $S$ on $\bbT$ such $[S,{\mathcal D}]\subset {\mathcal D}$, contains a subalgebra $\mathfrak{h}\subset \mathfrak{g}_{\mathcal D}$ isomorphic to $\sla(n+1,\bbR)$, $\mathfrak{h}=\sla(n+1,\bbR)$.   
  \end{proposition}
 
Of course the Lie algebra $\mathfrak{h}$ is in general not the full symmetry algebra of the twistor distribution $\mathcal D$, as we know form the $n=2$ case, discussed in the Remark \ref{symg2twi}, where only eight symmetries of the distribution $\mathcal D$, namely $S_1,S_2,S_5,S_6,S_{11},$ $\dots, S_{14}$, out of all fourteen, were the lifts of conformal symmetries of the dancing metric $g$ from Theorem \ref{the_dist_g2}.

\subsection{The $n=3$ case: the hidden ${\bf G}_2$ symmetry again}\label{again}
Here is the situation in the  $n=3$ case. We have the following proposition.
\begin{proposition}\label{pro857}
  If $n=3$, consider the twistor bundle $\sog(3)\to \bbT\stackrel{\pi}{\to} T$, and the twistor distribution ${\mathcal D}=\Span(Y_1,Y_2,Y_3)$ on $\bbT$
as in Proposition \ref{pro85} point (ii). Then 
  the following fifteen vector fields $S_I$, $I=1,2,\dots,15$, on $\bbT$ span a subalgebra $\mathfrak{h}=\sla(4,\bbR)\subset \mathfrak{g}_{\mathcal D}$ of the Lie algebra $\mathfrak{g}_{\mathcal D}$ of infinitesimal symmetries of $\mathcal D$  
  $$\begin{aligned}
    S_1=&x^2\partial_{x^1}-a_1\partial_{a_2}+t_2\partial_{t_1}\\
    S_2=&x^1\partial_{x^2}-a_2\partial_{a_1}+t_1\partial_{t_2}\\
     S_3=&x^1\partial_{x^3}-a_3\partial_{a_1}+t_1\partial_{t_3}\\
     S_4=&x^3\partial_{x^1}-a_1\partial_{a_3}+t_3\partial_{t_1}\\
      S_5=&x^3\partial_{x^2}-a_2\partial_{a_3}+t_3\partial_{t_2}\\
      S_6=&x^2\partial_{x^3}-a_2\partial_{a_2}+t_2\partial_{t_3}\\
      S_7=&x^1\partial_{x^1}-a_1\partial_{a_1}-t_2\partial_{t_2}-t_3\partial_{t_3}\\
      S_8=&x^2\partial_{x^2}-a_2\partial_{a_2}-t_1\partial_{t_1}-t_3\partial_{t_3}\\
      S_9=&x^3\partial_{x^3}-a_3\partial_{a_3}-t_1\partial_{t_1}-t_2\partial_{t_2}\\
      S_{10}=&\partial_{x^1}-a_1(a_1\partial_{a_1}+a_2\partial_{a_2}+a_3\partial_{a_3})+(a_2t_2+a_3t_3)\partial_{t_1}-a_1(t_2\partial_{t_2}+t_3\partial_{t_3})\\
      S_{11}=&\partial_{x^2}-a_2(a_1\partial_{a_1}+a_2\partial_{a_2}+a_3\partial_{a_3})+(a_1t_1+a_3t_3)\partial_{t_2}-a_2(t_1\partial_{t_1}+t_3\partial_{t_3})\\
      S_{12}=&\partial_{x^3}-a_3(a_1\partial_{a_1}+a_2\partial_{a_2}+a_3\partial_{a_3})+(a_1t_1+a_2t_2)\partial_{t_3}-a_3(t_1\partial_{t_1}+t_2\partial_{t_2})\\
      S_{13}=&\partial_{a_1}-x^1(x^1\partial_{x^1}+x^2\partial_{x^2}+x^3\partial_{x^3})+2x^1(t_1\partial_{t_1}+t_2\partial_{t_2}+t_3\partial_{t_3})\\
      S_{14}=&\partial_{a_2}-x^2(x^1\partial_{x^1}+x^2\partial_{x^2}+x^3\partial_{x^3})+2x^2(t_1\partial_{t_1}+t_2\partial_{t_2}+t_3\partial_{t_3})\\
      S_{15}=&\partial_{a_3}-x^3(x^1\partial_{x^1}+x^2\partial_{x^2}+x^3\partial_{x^3})+2x^3(t_1\partial_{t_1}+t_2\partial_{t_2}+t_3\partial_{t_3}).
  \end{aligned}$$
  These vector fields are lifts from $T$ to $\bbT$ of conformal symmetries of the dancing metric $g$ on $T$, in the sense that the pushforwards $\pi_*(S_I)$ of symmetries $S_I$ by $\pi$, from $\bbT$ to $T$, form the full Lie algebra $\sla(4,\bbR)$ of conformal symmetries of the corresponding dancing metric $g=2\theta^1\theta^6+2\theta^2\theta^5+2\theta^3\theta^4$ on $T$.
  
  The Lie algebra $\mathfrak{h}=\sla(4,\bbR)$ of the fifteen vectir feilds $S_I$ is not the full algebra $\mathfrak{g}_{\mathcal D}$ of the infinitesimal symmetries of the twistor distribution ${\mathcal D}=\Span(Y_1,Y_2,Y_3)$, as any vector field $S=f Y$, tangent to the vector field \be Y=t_2Y_2+t_3Y_1+t_1Y_3,\label{YY}\ee
  where $Y_1,Y_2,Y_3$ are as in point (iii) of Proposition \ref{pro85}, is a symmetry of $\mathcal D$, and we have that $\Span(Y)\cap\mathfrak{h}=\{0\}$.

  The distribution ${\mathcal D}=\Span(Y_1,Y_2,Y_3)$ on $\bbT$ has an infinite dimensional Lie algebra $\mathfrak{g}_{\mathcal D}$ of symmetries.
\end{proposition}

Another interesting feature of the twistor distribution $\mathcal D$ in the case when $n=3$ is that its \emph{derived flag}, i.e. the sequence ${\mathcal D}^1=[{\mathcal D},{\mathcal D}]$, ${\mathcal D}^{i+1}=[{\mathcal D},{\mathcal D}^{i}]$ for $i\in\bbN$, \emph{stabilizes} at $i=2$, which means that the distribution ${\mathcal D}^2$ is \emph{integrable}. Moreover, we have
$$\rank{({\mathcal D})}=3,\quad\rank{({\mathcal D}^1)}=4,\quad \rank{({\mathcal D}^2)}=6\quad \mathrm{and}\quad [{\mathcal D}^2,{\mathcal D}^2]\subset{\mathcal D}^2,$$
so in the langauge of the theory of distributions the twistor distribution ${\mathcal D}=\Span(Y_1,Y_2,Y_3)$  \emph{is a $(3,4,6)$ distribution}.

For this distribution the vector field $Y$ appearing in Proposition \ref{pro857} playes a particular role. It is its \emph{Cauchy characteristic}\footnote{For completeness: if ${\mathcal D}^\perp=\{\Lambda^1\bbT\ni \lambda\,\,|\,\, Y\hook\lambda =0, \, \forall Y\in\mathcal{D}\}$ is the \emph{annihilator of} $\mathcal D$, then a vector field $S$ on $\bbT$ is a \emph{Cauchy characteristic for} ${\mathcal D}$, if and only if $S\in{\mathcal D}$ and $S\hook (\der\lambda)\in {\mathcal D}^\perp$ for all $\lambda\in{\mathcal D}^\perp$}. This is defined as a special kind of a symmetry of a distribution tangent to it,  see e.g. \cite{Bryant}, Chapter II. \textsection 2, p.30.

So, the twistor distribution ${\mathcal D}=\Span(Y_1,Y_2,Y_3)$ on $\bbT$ \emph{is a $(3,4,6)$ distribution with a Cauchy characteristic} $Y$ \emph{on it}. Actually the module of Cauchy characteristics for the twistor distribution $\mathcal D$ is 1-dimensional, and is spanned by the vector field $Y$ defined in \eqref{YY}.

This last information about $Y$ being a characteristic for $\mathcal D$ suggest a change of coordinates in $\bbT$ from $(a_i,x^j,t_\alpha)$, to $(a,b,c,x,y,z,u_1,u_2,u_3)$ in which, in particular $Y$ is ramified. Thus for this reason, and for further convenience, we introduce new coordinates on $\bbT$ related to  $(a_i,x^j,t_\alpha)$ by:
$$
\begin{aligned}
  &x=\frac{a_1 t_3}{a_1 t_1+a_2 t_2+a_3 t_3},\quad  y=\frac{a_2 t_3}{a_1 t_1+a_2 t_2+a_3 t_3},\quad  z=\frac{a_1 t_1+a_2 t_2+a_3 t_3}{t3},\\
  &\\
  &a=\frac{x^1 t_3-x^3t_1}{\sqrt{t_3}},\quad  b=\frac{x^2 t_3-x^3t_2}{\sqrt{t_3}},\quad  c=\sqrt{t_3}\,\Big(\frac{x^1a_1+x^2 a_2+x^3 a_3-1}{a_1 t_1+a_2 t_2+a_3 t_3}\Big)^{1/3},\\
  &u_1=\frac{t_1}{t_3},\quad \hspace{1.25cm} u_2=\frac{t_2}{t_3},\quad \hspace{1.25cm} u_3=t_3.
\end{aligned}
$$
In these coordinates the three vector fields  $Y_1,Y_2,Y_3$ spanning the twistor distribution ${\mathcal D}$ can be replaced by the three vector fields
$$\begin{aligned}
  Y_{}{}'=&\partial_{u_3},\\
  Y_2{}'=&c^3\partial_x+\tfrac12 a^2\partial_a+(\tfrac12 ab-1)\partial_b+\tfrac12 ac\partial_c+(au_3-2yz\sqrt{u_3})\partial_{u_3},\\
  Y_3{}'=&c^3\partial_y+(\tfrac12 ab+1)\partial_a+\tfrac12 b^2\partial_b+\tfrac12 bc\partial_c+(bu_3+2xz\sqrt{u_3})\partial_{u_3}.
\end{aligned}$$
Here, $Y'=\frac{Y}{2u_3^2 z}$, $Y_2{}'=\frac{-Y_2}{\sqrt{u_3}}$ and $Y_3{}'=\frac{Y_3}{\sqrt{u_3}}$, so they are proportional to the respective $Y$, $Y_2$ and $Y_3$, and as such they also span the twistor distribution $\mathcal D$,
$${\mathcal D}=\Span(Y',Y_2{}',Y_3{}').$$
We will use this basis for the twistor distribution $\mathcal D$ from now on.

\begin{remark}
  Note that in the new coordinates $(x,y,z,a,b,c,u_1,u_2,u_3)$ on $\bbT$, the twistor distribution $\mathcal D$, represented at each point of $\bbT$ by the 3-plane $Y'\dz Y_2{}'\dz Y_3{}'$, does \emph{not} depend on the coordinates $(z,u_1,u_2,u_3)$, which are tangent to the fibers of $\sog(3)\to \bbT\stackrel{\pi}{\to} T$ (along $(z,u_1,u_2)$), and tangent to the characteristic direction $Y'$ (along $u_3$). Moreover, since the characteristic direction $Y'$ preserves the distribution $\mathcal D$  we have the following proposition.
  \end{remark}
\begin{proposition}
  The rank \emph{two} disttribution ${\mathcal E}$, spanned by the vector fields $Y_2{}'$ and $Y_3{}'$, $${\mathcal E}=\Span(Y_2{}',Y_3{}'),$$
  is well defined on $\bbT$ and 
  has the derived flag ${\mathcal E}^1=[{\mathcal E},{\mathcal E}]$ and ${\mathcal E}^2=[{\mathcal E},{\mathcal E}^1]$ such that
  $$\rank({\mathcal E})=2,\quad \rank({\mathcal E}^1)=3,\quad \rank({\mathcal E}^2)=5, \quad [{\mathcal E}^2,{\mathcal E}^2]\subset{\mathcal E}^2.$$  This means that $\mathcal E$ is a $(2,3,5)$ distribution in the 9-dimensional twistor bundle $\bbT$.

Moreover, the distribution ${\mathcal E}$ is preserved, when Lie transported along the characteristic direction $Y'$, and thus it defines a $(2,3,5)$ distribution $${\mathcal F}=\sigma_*{\mathcal E}$$ on the 5-dimensional quotient $Q$ of the twistor bundle $\bbT$ by its fibers, tangent to $\partial_{u_1},\partial_{u_2},\partial_z$, and the characteristic symmetry direction $Y'=\partial_{u_3}$.

In coordinates $(a,b,c,x,y,z,u_1,u_2,u_3$ in $\bbT$ the cannonical projection $$\sigma:\bbT\to Q$$ is given by $$ \sigma(a,b,c,x,y,z,u_1,u_2,u_3)=(a,b,c,x,y),$$
so that $(a,b,c,x,y)$ are the local coordinates on $Q$. 
\end{proposition}

Finally we have the following, a bit surprising theorem.
\begin{theorem}
  The $(2,3,5)$ quotient twistor distribution $\mathcal F$ on the 5-dimensional manifold $Q$ canonnically associated with the $n=3$ Beltrami -- de Sitter model $BdS$ has it Lie algebra of all infinitesimal symmetries $\mathfrak{g}_{\mathcal F}$ isomorphic to the split real form of the exceptional simple Lie algebra $\mathfrak{g}_2$.

  In coordinates $(a,b,c,x,y)$ on $Q$ the distribution $\mathcal F$ is $${\mathcal F}=\Span(X_2,X_3)$$
  with the spanning vector fields $X_2, X_3$ given by :
  $$\begin{aligned}X_2=&c^3\partial_x+\tfrac12 a^2\partial_a+(\tfrac12 ab-1)\partial_b+\tfrac12 ac\partial_c,\\
  X_3=&c^3\partial_y+(\tfrac12 ab+1)\partial_a+\tfrac12 b^2\partial_b+\tfrac12 bc\partial_c.
  \end{aligned}$$
  The fourteen $\mathfrak{g}_2$ symmetry generators are:
  $$\begin{aligned}
    R_1=&\partial_x\\
    R_2=&\partial_y\\
    R_3=&a\partial_b-y\partial_x\\
    R_4=&b\partial_a-x\partial_y\\
    R_5=&a\partial_a-b\partial_b-x\partial_x+y\partial_y\\
    R_6=&c\partial_c+3x\partial_x+3y\partial_y\\
    R_7=&4c^2\partial_x-\frac{a^2}{c}\partial_a-\frac{ab-2}{c}\partial_b+a\partial_c\\
    R_8=&4c^2\partial_y-\frac{ab+2}{c}\partial_a-\frac{b^2}{c}\partial_b+b\partial_c\\
    R_9=&2x^2\partial_x+2xy\partial_y+(2c^3-xa-2by)\partial_a+bx\partial_b+cx\partial_c\\
    R_{10}=&2y^2\partial_y+2xy\partial_x+(2c^3-2xa-by)\partial_b+ay\partial_a+cy\partial_c\\
    R_{11}=&2bcx\partial_x+2c(c^3-ax)\partial_y+\frac{c^3(ab-4)+xa}{c^2}\partial_a+\frac{b(bc^3+x)}{c^2}\partial_b+\frac{bc^3-x}{c}\partial_c\\
    R_{12}=&2acy\partial_y+2c(c^3-by)\partial_x+\frac{c^3(ab+4)-by}{c^2}\partial_b+\frac{a(ac^3-y)}{c^2}\partial_a+\frac{ac^3+y}{c}\partial_c\\
    R_{13}=&2bc\partial_x-2ac\partial_y+\frac{a}{c^2}\partial_a+\frac{b}{c^2}\partial_b-\frac{1}{c}\partial_c\\
    R_{14}=&4c^2(x\partial_x+y\partial_y)+\frac{a(3c^3-ax-by)-2y}{c}\partial_a+\frac{b(3c^3-ax-by)+2x}{c}\partial_b+(c^3+ax+by)\partial_c.
        \end{aligned}$$
  \end{theorem}
\begin{remark}
  It is interesting that the $\mathfrak{g}_2$ symmetries of the the canonical $(2,3,5)$ twistor distribution ${\mathcal F}$ on the twistor quotient $Q$ have nothing to do with the $\sla(4,\bbR)$ symmetry of the `parent' canonical $(3,4,6)$ twistor distribution $\mathcal D$ on $\bbT$. Only two, out of 15 generators $S_I$ of the $\sla(4,\bbR)$ symmetries of $\mathcal D$ given in Proposition \ref{pro857} have pushforwards that are vector fields on $Q$, these are $\sigma_*(S_1)=R_4$ and $\sigma_*(S_2)=R_3$. 
\end{remark}
\begin{remark}
  Another interesting feature of the `parent' $(3,4,6)$ twistor distribution $\mathcal D$ on the twistor bundle $\bbT$ is that its \emph{any} rank 2 \emph{subdistribution} transversal to the Cauchy characteristic $Y$ is $(2,3,5)$. It is difficult to check if any of these subdistributions contains $\mathfrak{g}_2$ as their Lie algebra of symmetries. 
  \end{remark}

 \subsection{The 4-dimensional Beltrami -- de Sitter model} We only mention that in $n=4$ case the twistor distribution $\mathcal D$ on the 14-dimensional twistor bundle $\bbT$, as given in Proposition \ref{pro85}, point (iii), is $(4,10,14)$ so, contrary to the distribution $\mathcal D$ for $n=3$, it is a \emph{bracket generating} distribution on $\bbT$. As stated in Proposition \ref{pro86} it surely has $\mathfrak{h}=\sla(5,\bbR)$ as its Lie algebra of symmetries. Further properties of this twistor space and this distribution, as well as the general case of $n\geq 4$ will be discussed in a paper with Ian Anderson \cite{ian}.

\end{document}